\newif\iffullversion
\newif\ifdraft
\newif\ifanonymous
\fullversiontrue

\documentclass{article}

\usepackage[ascii]{inputenc}
\usepackage[T1]{fontenc}
\usepackage{microtype}
\usepackage{paralist}
\usepackage{amsmath,amssymb}
\usepackage[framemethod=tikz]{mdframed}
\usepackage{stmaryrd}
\usepackage{mathpartir}
\usepackage{tocloft}
\usepackage{relsize}
\usepackage{declmath}
\usepackage{makeidx}\makeindex
\usepackage[a4paper]{geometry}
\usepackage{version}
\usepackage{xcolor}
\usepackage{xr}
\usepackage[pdfborderstyle={/S/U/W 0.3}]{hyperref}
\usepackage{cleveref}
\usepackage[backend=bibtex,maxbibnames=100]{biblatex}

\newif\iftexht
\ifx\HCode\undefined\else
\texhttrue
\fi

\input{macros}

\iffullversion
\newcommand\fullshort[2]{#1}
\includeversion{fullversion}
\excludeversion{shortversion}
\else
\newcommand\fullshort[2]{#2}
\excludeversion{fullversion}
\includeversion{shortversion}
\fi

\newcommand\fullonly[1]{\fullshort{#1}{}}
\newcommand\shortonly[1]{\fullshort{}{#1}}

\newtheorem{definition}{Definition}

\newtheorem{lemma}{Lemma}

\crefname{lemma}{Lemma}{Lemmas}
\crefname{definition}{Definition}{Definitions}

\bibliography{ghost}

\newcommand{\mathcolorbox}[2]{\mathchoice%
  {\colorbox{#1}{$\displaystyle#2$}}%
  {\colorbox{#1}{$\textstyle#2$}}%
  {\colorbox{#1}{$\scriptstyle#2$}}%
  {\colorbox{#1}{$\scriptscriptstyle#2$}}}%

\def\symbolindexmarkhighlight#1{%
  \setlength{\fboxsep}{2pt}%
  {\TextOrMath{\colorbox{gray!20}{#1}}{\mathcolorbox{gray!20}{#1}}}}

\AtBeginDocument{
  \ifx\PreviewMacro\undefined\else
  \PreviewMacro*\footnote
  \PreviewMacro*\fullonly
  \fi
}

\begin{document}

\renewcommand\sectionautorefname{Section}
\renewcommand\subsectionautorefname{Section}
\renewcommand\subsubsectionautorefname{Section}
\newcommand\lemmaautorefname{Lemma}
\newcommand\conjectureautorefname{Conjecture}
\newcommand\definitionautorefname{Definition}
\newcommand\corollaryautorefname{Corollary}
\newcommand\claimautorefname{Claim}

\title{Quantum Hoare Logic with Ghost Variables}
\author{Dominique Unruh\\\small University of Tartu}

\maketitle

\ifdraft
\begin{center}
  \bfseries\Huge\fboxsep=10pt
  \framebox{THIS IS A DRAFT}
\end{center}
\fi

\begin{abstract}
  \noindent
  Quantum Hoare logic allows us to reason about quantum programs.  We
  present an extension of quantum Hoare logic that introduces ``ghost
  variables'' to extend the expressive power of pre-/postconditions.
  Ghost variables are variables that do not actually occur in the
  program and are allowed to have arbitrary quantum states (in a
  sense, they are existentially quantified), and be entangled with
  program variables.  Ghost variables allow us to express properties
  such as the distribution of a program variable or the fact that a
  variable has classical content.  And as a case study, we show how quantum Hoare logic
  with ghost variables can be used to prove the security of the
  quantum one-time pad.
\end{abstract}

\tableofcontents

\section{Introduction}

Designing algorithms is an inherently error-prone process.  This is
especially true for quantum algorithms.  Quantum algorithms can solve
certain computational problems much faster than classical computers
(e.g., \cite{shor,grover,harrow09quantum}), and most
likely will be of great impact once quantum computers are available.
And already now, analyzing quantum algorithms is of practical
relevance when proving the security of cryptosystems against future
quantum attackers (post-quantum cryptography).  But quantum algorithms
are difficult to get right: Quantum mechanics has many properties that
go against human intuition, and quantum programs are difficult to test
and debug since we cannot directly observe their state (except in
small-scale simulations).  A solution to this problem is formal
verification where we prove the behavior of the algorithm.  In the
classical realm, Hoare logic \cite{hoare} (and its close relative, the predicate
transformers \cite{dijkstra75guarded}, which we treat as the same for the sake of
this introduction) has proven to be an invaluable tool for the
analysis of imperative programs.  In Hoare logics, we analyze the
behavior of a program by investigating what ``postcondition'' the
final state of a program satisfies if the initial state satisfies a
certain ``precondition''.  This allows us to formally prove the
behavior of a complex program by first deriving the behavior of
individual lines of code and then modularly plugging these together to
get a description of the behavior of the whole program.  Hoare logics
have been developed also for probabilistic programs
\cite{kozen83probabilistic,McIverMorgan2005} and quantum programs
\cite{dhondt06weakest,ying12floyd,chadha06reasoning,feng07proof,kakutani09logic}.

However, existing Hoare logics still have limitations at to what can be
expressed within a pre-/postcondition.  For example, we cannot express
that the value of a certain variable is uniformly distributed.\footnote{%
  Classical/quantum calculi that use ``expectations'' (predicates that
  do not just hold/not hold, but hold to a certain degree,
  \cite{kozen83probabilistic,McIverMorgan2005,dhondt06weakest,ying12floyd})
  allow us to reason about probabilistic behavior.  But they only
  allow us to reason about the probability of a certain event (or the
  expectation value of a quantity), but not about the distribution of
  value.  (I.e., we can express ``$\xx$ has value $0$ with probability
  at least $\alpha$'' but not ``$\xx$ is uniform''.)}  (E.g., to
state that inside a while-loop, we have the invariant that $\xx$
is a uniformly random bit.)  Or, specific to the quantum case, that a
quantum variable has a certain distribution (e.g., is in the
``completely mixed state'', the quantum analogue to a uniform
distribution). Or
that a quantum variable is not entangled with other variables. Or that
a quantum variable contains classical data at a certain point in the
program.  (Some logics distinguish
classical and quantum variables, e.g., 
\cite{qrhl},
at the costs of more complex semantics.  But this does not allow us to
reason about dynamic properties, e.g., that a variable becomes
classical after a measurement.)

In this article, we present an extension of quantum Hoare logic that
removes these limitations.  We introduce ``ghost variables'' and show
that using ghost variables, we can encode properties such as ``$\xx$
has distribution $D$''
or ``$\xx$
is separable (unentangled)'' or ``$\xx$
is classical''.  (That is, all of these are emerging properties, not
hardcoded into our logic.)  A ghost variable is a variable that does
not actually occur in the program but is introduced merely in a
predicate (pre-/postcondition).  The ghost variable is then allowed to
take any value that makes the predicate true (effectively
existentially quantified).  E.g., the classical predicate
$\xx=\gggg^2$
(where $\gggg$
is a ghost variable) would express that $\xx$
is a square.  Classical ghost variables, however, do not yield any new
expressive power since existential quantifiers are already allowed in
most Hoare logics (so we could state the predicate as
$\exists g.\,\xx=g^2$).
In the quantum setting, however, a ghost variable can have a quantum
state, and possibly be entangled with other variables!  For example,
if $\psi$
is a maximally entangled state between two variables, then looking at
only one of those variables, we would see a uniformly distributed
variable.  And a uniformly distributed variable can always be seen as
part of a system with two variables in state $\psi$.
Thus the predicate ``$\xx\ee$
together are in state $\psi$''
(where $\ee$
is a ghost variable) models the fact that $\xx$
is uniformly distributed.  Similarly we can encode classicality and
separability. Thus, using ghost variables, we can continue reasoning
about quantum programs using Hoare logic, but additionally have
program invariants that state that variables are distributed in
certain ways, are classical, are separable, and more. (We stress that
this even is an advance over the state of the art for classical
programs.  Our logic could be used in the analysis of classical
programs to express the fact that certain variables have certain
distributions.  This is quite unexpected since we would be using a
quantum phenomenon to analyze purely classical programs.)

Additionally the introduction of ghost variables makes the foundations
of the investigated programming language simpler.  Many operations that
one thinks of as elementary (such as random sampling, measurements)
can actually be built from more elementary operations (such as
applying a unitary operation, initializing a quantum register).  Using
ghost variables we can then derive the properties of the derived
operations from the properties of the elementary one.  (E.g., show
that after random sampling the assigned variable has a certain
distribution and is classical.)  This means that the language is
simpler (and thus arguably more foundationally elegant), and the core
set of rules of our logic is quite small (eleven rules).

Finally, we demonstrate that our logic can be applied to problems that
seem out of reach of existing Hoare logics: We analyze quantum
one-time pad encryption and show that it is secure, i.e., that an encrypted
quantum message indeed ``looks random''.

\paragraph{Related work.} Hoare logic was first introduced by Hoare
\cite{hoare}.  A different view was provided by Dijkstra
\cite{dijkstra75guarded} using predicate transformers.  Hoare
logics/predicate transformers were generalized by Kozen
\cite{kozen83probabilistic} (and \cite{McIverMorgan2005} for the case
of combined probabilism/nondeterminism).
Quantum Hoare logics and predicate transformer calculi for quantum programs
have been presented by D'Hondt and Panangaden \cite{dhondt06weakest},
Chadha, Mateus and Sernadas \cite{chadha06reasoning}, Feng, Duan, Ji,
and Ying \cite{feng07proof}, Ying \cite{ying12floyd},
and Kakutani \cite{kakutani09logic}. Unruh \cite{qrhl} gives a quantum
Hoare logic for analyzing pairs of programs (based on the classical
pRHL \cite{certicrypt}).  A different approach is taken by
pictorial calculi where quantum processes can be formalized and
rewritten as diagrams, starting with Abramsky and Coecke
\cite{abramsky-coecke}.  \cite[\fullshort{Example }{Ex.~}4.91]{coecke17picturing}
applies this approach to the classical one-time pad, but only to its
correctness, not its security.  The quantum one-time pad was
discovered by \cite{boykin03qotp,mosca00qotp}.

\paragraph{Organisation.}
\autoref{sec:var.mem.pred} introduces some quantum basics as well as
important notation and auxiliary concepts.  \autoref{sec:qprogs}
introduces \fullonly{syntax and semantics of }the simple imperative quantum
language we use for our calculus.  (And explains how random sampling,
measurements, etc.~are encoded using more basic language features.)
\autoref{sec:hoare.ghosts} introduces our Hoare logic with
ghosts. (The concept of ghost variables and the semantics of Hoare judgments.) 
\autoref{sec:pred.ghosts} shows how important properties such as
distributions of variables, classicality, separability can be encoded
in pre-/postconditions using ghost variables.
\autoref{sec:core.rules} presents and explains the eleven core rules
of the logic from which all other rules can be derived.
\autoref{sec:deriv.rules} \fullshort{derives a number of additional rules}{presents additional
  rules that can be derived} from
the core rules. (For reasoning about derived language features, and
for convenient reasoning about programs with classical variables.)
\autoref{sec:qotp} analyses the quantum one-time pad.
\shortonly{The appendix contains material such as full proofs that will be included
  in the arXiv version of this paper.}

\section{Preliminaries: Variables, Memories, and Predicates}
\label{sec:var.mem.pred}

In this section, we introduce some fundamental concepts and notations
needed for this paper, and recap some of the needed quantum background
as we go along.
When introducing some notation $X$, the place of definition is marked like this: \symbolindexmarkhighlight{$X$}. All symbols are listed in the symbol index.\shortonly{\interfootnotelinepenalty=0 \footnote{Symbol index:
    \fboxsep=0pt
    \renewcommand\symbolindexpage[2]{\hyperref[#2]{p.#1}}%
    \renewcommand\symbolindexentry[4]{%
      \ifx\@nil#4\@nil\else
      \mbox{\colorbox{gray!20}{$#2$}\,#4}%
      \space
      \fi
    }%
  \renewenvironment{thesymbolindex}{}{}%
  \printsymbolindex
}}

\paragraph{Variables.}
Before we introduce the syntax and semantics of programs, we first
need to introduce some basic concepts. A
\emph{variable}\index{variable} is described by a variable name $\xx$
that identifies the variable, and a type $T$.  The \emph{type}%
\index{type!(of a variable)} of $\xx$
is simply the set of all (classical) values the variable can
take. E.g., a variable might have type $\bit$,
or $\setN$.\fullonly{\footnote{We
  stress that we do not assume that the type is a finite or even a
  countable set. Consequently, the Hilbert spaces considered in this
  paper are not necessarily finite dimensional or even
  separable. However, all results can be informally understood by
  thinking of all sets as finite and hence of all Hilbert spaces as
  $\setC^N$
  for suitable $N\in\setN$.}}
We will assume that there is always some distinguished value in $T$
that we denote \symbolindexmark0{$0$}.
  
We distinguish between three kinds of variables: program variables
\symbolindexmark\xx{$\xx,\yy,\zz$}
(that can occur in programs), entangled ghost variables
\symbolindexmark\ee{$\ee$},
and unentangled ghost variables \symbolindexmark\uu{$\uu$}. We write
\symbolindexmark\gggg{$\gggg$} for variables that are entangled ghosts or unentangled ghosts.
(The meaning of these kinds will become clear later, for now they
simply form a partition of the set of all variables.)  We use
\symbolindexmark\vv{$\vv,\ww$} when we do not wish to specify the kind of variable.

Lists\pagelabel{page:variable.conventions} or sets of variables will be denoted \symbolindexmark\VV{$\VV,\WW$} or
\symbolindexmark\XX{$\XX,\YY,\ZZ$} or
\symbolindexmark\EE{$\EE$} or
\symbolindexmark\UU{$\UU$} or
\symbolindexmark\GG{$\GG$}
(depending on the kind of variable they contain). 
Given a list
$\VV=\vv_1\dots\vv_n$
of variables, we say its \emph{type} \index{type!(of a list of
  variables)} is $T_{1}\times\dots\times T_{n}$
if $T_{i}$ is the type of $\vv_i$.
We write \symbolindexmark\progvars{$\progvars\VV$} for the program
variables in $\VV$.

\paragraph{Memories and quantum states.}
An \emph{assignment}\index{assignment} assigns to each variable a
classical value. Formally, for a set $\VV$,
the \emph{assignments over $\VV$}
are all functions \symbolindexmark\mm{$\mm$}
with domain $\VV$
such that: for all $\xx\in\VV$
with type $T_\xx$,
$\mm(\xx)\in T_\xx$.
That is, assignments can represent the content of classical memories.

To model quantum memories, we simply consider superpositions of
assignments: A \emph{(pure) quantum memory}%
\index{pure quantum memory}%
\index{quantum memory!(pure)}%
\index{memory!(pure) quantum} is a superposition of
assignments. Formally, \symbolindexmark\elltwov{$\elltwov\VV$},
the set of all quantum memories over $\VV$,
is the Hilbert space with basis\fullonly{\footnote{When we say ``basis'', we
  always mean orthonormal basis.}} $\{\ket\mm\}_\mm$
where $\mm$
ranges over all assignments over~$\VV$. Here \symbolindexmark\ket{$\ket\mm$} simply denotes the
basis vector labeled $\mm$, we often write $\ket\mm_{\VV}$ to stress
which space we are talking about.
Intuitively, a quantum memory $\psi$
over $\VV$
with $\norm\psi=1$
represents a state a quantum computer with variables $\VV$
could be in.  \fullonly{(We do not require $\norm\psi=1$
for a quantum memories unless this is explicitly mentioned.)}

We also consider quantum states over
arbitrary sets $X$ (as opposed to sets of assignments).
Namely, \symbolindexmark\elltwo{$\elltwo X$}
denotes the Hilbert space with orthonormal basis $\{\ket x\}_{x\in X}$.
(In that notation, $\elltwov\VV$
is simply $\elltwo A$
where $A$
is the set of all assignments on $\VV$.)
Elements $\psi\in\elltwo X$
with $\norm\psi=1$
represent quantum states.

We often treat elements of $\elltwo T$ and $\elltwov\VV$ interchangeably if $T$ is
the type of $\VV$ since there is a natural isorphism between those spaces.

The tensor product $\otimes$ combines two quantum states
$\psi\in\elltwo X,\phi\in\elltwo Y$ into a joint system
$\psi\otimes\phi\in\elltwo{X\times Y}$. In the case of quantum
memories $\psi,\phi$ over $\VV,\WW$, respectively,
$\psi\otimes\phi\in\elltwov{\VV\WW}$.
(And $\psi\otimes\phi=\phi\otimes\psi$ since we are composing
``named'' systems.)

For a vector (or operator) $a$, we write \symbolindexmark\adj{$\adj
  a$} for its adjoint.  (In the finite dimensional case, the adjoint
is simply the conjugate transpose of a vector/matrix. The literature
also knows the notation $a^\dagger$.)  The adjoint of
$\ket x$ is written $\bra x$.
We abbreviate \symbolindexmarkonly\proj$\symbolindexmarkhighlight{\proj{\psi}}:=\psi\adj\psi$.
This is the projector onto $\psi$ when $\norm\psi=1$.

\paragraph{Mixed quantum memories.}
In many situations, we need to model probabilistic quantum states
(e.g., a quantum state that is $\ket0$
with probability $\frac12$
and $\ket1$
with probability $\frac12$). This is modeled using \emph{mixed
  states}%
\index{mixed state}%
\index{state!mixed} (a.k.a.~\emph{density operators}%
\index{density operator}%
\index{operator!density}). Having state $\psi_i$
with probability $p_i$
is represented by the operator $\rho:=\sum_i p_i\proj{\psi_i}$.\footnote{Mathematically,
  these are the set of all positive Hermitian trace-class operators on
  $\elltwo X$.
  The requirement ``trace-class'' ensures that the trace exists and
  can be ignored in the finite-dimensional case.}\fullonly{\,\footnote{Sums without
    index set are always assumed to have an arbitrary (not necessarily finite or even
    countable) index set.  In the case of sums of vectors in a Hilbert space,
    convergence is with respect to the Hilbert space norm, and in the case of sums
    of positive operators, the convergence is with respect to the Loewner order.}}
Then $\rho$
encodes all observable information about the distribution of the
quantum state (that is, two distributions of quantum states have the
same~$\rho$
iff they cannot be distinguished by any physical process).  And
$\tr\rho$
is the total probability $\sum_ip_i$.
(That is, $\tr\rho=1$
unless we wish to represent the outcome of a non-terminating program.)
We will often need
to consider mixed states of quantum memories (i.e., mixed states with
underlying Hilbert space $\elltwov\VV$). We call them
\emph{mixed (quantum) memories} over $\VV$.%
\index{mixed (quantum) memory}%
\index{quantum memory!mixed}%
\index{memory!mixed (quantum)}

For a mixed memory $\rho$
over $\VV\supseteq\WW$ the \emph{partial trace}%
\index{trace!partial}%
\index{partial trace}
\symbolindexmark\partr{$\partr{\WW}\rho$}\label{page:partr}
is the result of throwing away variables $\WW$
(i.e., it is a mixed memory over $\VV\setminus\WW$).
Formally, $\partr{\WW}$ is defined as the continuous linear function satisfying
$\partr{\WW}(\sigma\otimes\tau):=\sigma\cdot\tr\tau$ where $\tau$ is
an operator over $\WW$.

A mixed memory $\rho$
is \emph{$(\VV,\WW)$-separable}\index{separable}
(i.e., not entangled between $\VV$ and $\WW$) iff it can be written as
$\rho=\sum_i\rho_i\otimes\rho_i'$ for mixed memories
$\rho_i,\rho_i'$ over $\VV,\WW$, respectively.

\paragraph{Operations on quantum states.}
An operation on a quantum state is modeled by an isometry $U$
on $\elltwo X$.\footnote{That
  is, a norm-preserving linear operation. Often, one models quantum
  operations as unitaries instead because in the finite-dimensional
  case an isometry is automatically unitary. However, in the
  infinite-dimensional case, unitaries are unnecessarily
  restrictive. Consider, e.g., the isometry $\ket i\mapsto\ket{i+1}$
  with $i\in\setN$
  which is a perfectly valid quantum operation but not a unitary.} If
we apply such an operation on a mixed state $\rho$,
the result is $U\rho \adj U$.

Most often,
isometries will occur in the context of operations that are performed
on a single variable or list of variables, i.e., an isometry $U$
on $\elltwov\VV$.
Then $U$
can also be applied to $\elltwov\WW$
with $\WW\supseteq\VV$:
we identify $U$
with $U\otimes\id_{\WW\setminus\VV}$.
Furthermore, if $\VV$
has type $T$,
then an isometry $U$
on $\elltwo T$
can be seen as an isometry on $\elltwov\VV$
since we identify $\elltwo T$ and $\elltwov\VV$. If we want to make
$\WW$ explicit, we write \symbolindexmark\opon{$\opon
  U\WW$}\label{page:opon} for the isometry $U$ on $\elltwov\VV$.
For example, if $U$ is a $2\times 2$-matrix and $\xx$ has type bit,
then $\opon U\xx$ can be applied to quantum memories over $\xx\yy$,
acting on $\xx$ only.
This notation is not limited to isometries, of course, but applies
to other operators, too.
(By ``operator'' we always mean a bounded linear operator in this paper.)

An important operation is \symbolindexmark\CNOT{$\CNOT$}
on $\VV\WW$
(where $\VV,\WW$
both have type $\bits n$),
defined by
$\CNOT(\ket x_{\VV}\otimes\ket y_{\WW}):=\ket
x_{\VV}\otimes\ket{x\oplus y}_{\WW}$. (That is, we allow $\CNOT$
not only on single bits but bitstrings.)

\paragraph{Predicates.} In Hoare judgments, we need to express
properties of the state of a quantum memory. In this paper, we only
consider properties that are closed under superpositions of quantum
states.  That is, a \emph{predicate}\index{predicate} \symbolindexmark\PA{$\PA,\PB,\PC$} on $\VV$
is a subspace\fullonly{\footnote{By subspace, we always mean closed
  subspaces. In the finite-dimensional case, all subspaces are closed
  anyway.}} of $\elltwov\VV$.
The syntax of predicates will not be fixed to a specific language,
i.e., any mathematically expressible subspace is a valid
predicate. But we fix some syntactic sugar for expressing
predicates succinctly:
\begin{itemize}
\item \symbolindexmark\top{$\top$},
  \symbolindexmark\zero{$\zero$}:
  The predicate that is always satisfied is denoted
  $\top:=\elltwov\VV$. The predicate that is never satisfied is $\bot:=\{0\}$.
\item \symbolindexmark\wedge{$\PA\wedge\PB$},\
  \ \symbolindexmarkhighlight{$\PA,\PB$},\
  \ \symbolindexmark\vee{$\PA\vee \PB$}:\pagelabel{page:vee}
  To model that both $\PA$
  and $\PB$
  hold (conjunction), we simply use the intersection of $\PA$
  and $\PB$
  (as sets). That is, we write $\PA\land \PB$
  to denote $\PA\cap \PB$.
  We will often also write this as ``$\PA,\PB$''
  instead of ``$\PA\land \PB$''
  where ``$,$''
  is understood to bind less closely than ``$\vee$''.
  To model that $\PA$
  or $\PB$
  holds (disjunction), we use the sum $\PA+\PB$
  (the space of all linear combinations from $\PA$
  and $\PB$).
  That is, we write $\PA\vee \PB$
  to denote $\PA+\PB$.
  \fullonly{This choice of connectives corresponds to Birkhoff-von Neumann
  quantum logic.  But we stress that we are not restricted to using
  only these connectives, we may use any well-defined operations on
  subspaces. These are just the ones that will turn out useful in the
  remainder of the paper.}
\item \symbolindexmark\oppred{$\oppred M\PA$}:
  For an operator $M$
  (typically an isometry or a projector) and a predicate $\PA$,
  we write $\oppred M\PA$
  for the subspace $\{M\psi:\psi\in\PA\}$.
  Thus, $\oppred M\PA$
  is satisfied if we apply $M$ to a quantum memory in $\PA$.
\item \symbolindexmark\quantin{$\XX\quantin S$}:
  For some variable list $\XX$
  of type $T$,
  we may wish to express the fact that the value of $\XX$
  lies in a certain subspace $S\subseteq\elltwo T$.
  Notice that $S$
  can be naturally seen as a subspace of $\elltwov\VV$.
  Then $\VV$
  being in a state in $S$
  means that the state of the whole quantum memory is in
  $S\otimes\elltwov{\VV\setminus\XX}$.
  We introduce the syntactic sugar ``$\XX\quantin S$''
  to denote $S\otimes\elltwov{\VV\setminus\XX}$.
  Note that even though it looks like a Boolean expression, it actually is a subspace
  of $\elltwov\VV$ and thus a predicate in our sense.
\item \symbolindexmark\quanteq{$\WW\quanteq\psi$}.\pagelabel{page:quanteq}
  Often, we will also want to express that the variables $\WW$
  are in a specific state $\psi$.
  This means the variables lie in $\SPAN\{\psi\}$,
  using the previous syntactic sugar we can write this as
  $\WW\quantin\SPAN\{\psi\}$.
  We introduce the abbreviation $\WW\quanteq\psi$
  for this common case.
\item \symbolindexmark\psubst{$\psubst\PA\WW{\WW'}$}\pagelabel{page:def:psubst}.
  Renaming variables $\WW$
  to $\WW'$
  in predicate $\PA$.
  This assumes $\PA$
  is a predicate over $\VV\supseteq\WW$,
  that $\WW\cap(\VV\setminus\WW)=\varnothing$,
  and that $\WW$
  and $\WW'$
  have the same type. Then $\psubst\PA\WW{\WW'}$
  is the predicate over $\paren{\VV\setminus\WW}\dotcup\WW'$
  defined by
  $\oppred{\Urename\WW{\WW'}}\PA$
  where \symbolindexmark\Urename{$\Urename\WW{\WW'}$}\pagelabel{page:def:Urename}
  is the natural isomorphism between $\elltwov{\WW}$
  and $\elltwov{\WW'}$,
  i.e., $\Urename\WW{\WW'}\ket i_{\WW}:=\ket i_{\WW'}$
  for all $i$.
  \fullonly{(${\Urename\WW{\WW'}}$
  renames $\WW$ to $\WW'$ when applied to a quantum memory.)}
\end{itemize}
An example $\PA_{\mathrm{example}}$ of a predicate would be:
\[
  \pb\paren{\xx\yy\quanteq\fsq\ket{00}+\fsq\ket{11}} \vee
  \pb\paren{\xx\quanteq\ket0},\
  \zz\quanteq\ket1
\]
This means, intuitively, that $\xx\yy$
are maximally entangled (in state $\fsq\ket{00}+\fsq\ket{11}$
up to a global phase factor) or $\xx$
has state $\ket0$, and in addition $\zz$ has state $\ket1$.

Note that our predicates seems to be lacking in expressiveness
compared with the predicates, e.g., from \cite{qrhl}: It is not
possible to parameterize the predicate using the values of classical
variables. E.g., we cannot write $\yy\quanteq\ket{\xx}$
where $\xx$
is a classical variable. This is because our semantics does not
hardcode the distinction between classical and quantum variables (all
variables are quantum by default). But, in \autoref{sec:deriv:classical}, we will see
how to express classical variables as a derived feature, and introduce
additional syntactic sugar that allows us to recover the full
expressiveness of the predicates from \cite{qrhl}.

Given a predicate $\PA$,
we will often wish to indicate which variables it talks about, i.e.,
what are its \emph{free variables}%
\index{free variables}%
\index{variables!free}. Since our definition of predicates is semantic
(i.e., we are not limited to predicates expressed using the syntax
above) we cannot simply speak about the variables occurring in the
expression describing $\PA$.
Instead, we say $\PA$
contains only variables from $\WW$
(written: \symbolindexmarkonly\fv$\symbolindexmarkhighlight{\fv(\PA)}\subseteq\WW$)
iff there exists a subspace $S$
such that $\PA=(\WW\quantin S)$.
(That is, if $\PA$
can be described solely in terms of the content of the variables
$\WW$.)
Note that there is a certain abuse of notation here: We formally
defined ``$\fv(\PA)\subseteq\WW$'',
but we do not define $\fv(\PA)$;
$\fv(\PA)\subseteq\WW$
should formally just be seen as an abbreviation for
$\exists S.\,\PA=(\WW\quantin S)$.\fullonly{\footnote{In
  fact, defining $\fv(A)$
  is possible only if there is a smallest set $\WW$
  such that $\exists S.\ \PA=(\WW\quantin S)$.
  This is not necessarily the case. For example, assume that $\VV$
  is infinite, let $\PA$
  be the space spanned by all $\ket \mm$
  where $\mm(\vv)\neq0$
  for only finitely many $\vv\in\VV$.  Then $\PA=(\WW\quantin S)$
  for any cofinite $\WW$
  (by defining $S$
  as the span of all $\ket\mm$
  with only finitely many $\mm(\vv)\neq0$),
  but $\PA\neq(\WW\quantin S)$
  whenever $\WW$
  is not cofinite.  Since there is no smallest cofinite $\WW$,
  $\fv(A)$
  cannot be defined, but we can still meaningfully use the notation
  $\fv(A)\subseteq\WW$.
}} For example, $\fv(\PA_{\mathrm{example}})\subseteq\{\xx,\yy,\zz\}$.
Similarly, we treat $\fv(\PA)\cap\WW=\varnothing$.

If\pagelabel{page:pred.identify} two predicates $\PA$
and $\PA'$
on variables $\VV$
and $\WW$,
respectively, satisfy
$\PA\otimes\elltwov{\WW\setminus\VV}=\PA'\otimes\elltwov{\VV\setminus\WW}$,
then $\PA$
and $\PA'$
intuitively describe the same property on the shared variables
$\VV\cap\WW$
(and say nothing about the remaining variables). Therefore we will
\emph{identify}%
\index{identify predicates}%
\index{predicate!(identifying)}
such $\PA$
and $\PA'$
throughout this paper.  In particular, any predicate $\PA$
can be seen as a predicate $\PA'$ on $\fv(\PA)$.

\section{Quantum programs}
\label{sec:qprogs}

\paragraph{Syntax.}
We will now define a small imperative quantum language. The set of all
programs is described by the following syntax:\symbolindexmarkonly\bc
\[
  \symbolindexmarkhighlight{\bc,\bd} ::=
  \apply U\XX
  \ |\
  \init \xx
  \ |\
  \ifte\yy\bc\bd
  \ |\
  \while\yy\bc
  \ |\
  \bc;\bd
  \ |\
  \SKIP
\]
Here $\XX$
is a list of program variables, $\xx$
a program variable, $\yy$
a program variable of type $\bit$,
and $U$
an isometry on $\elltwov\XX$
(there is no fixed set of allowed isometries, any isometry that we can
describe can be used here).\footnote{We will assume throughout the
  paper that all programs satisfy those well-typedness constraints. In
  particular, rules may implicitly impose type constraints on the
  variables and constants occurring in them by this assumption.}

Intuitively, \symbolindexmark\apply{$\apply U\XX$}
means that the operation $U$
is applied to the quantum variables~$\XX$.
E.g., $\apply H\xx$
would apply the Hadamard gate to the variable $\xx$
(we assume that $H$
denote the Hadamard matrix). It is important that we can apply $U$
to several variables $\XX$ simultaneously,
otherwise no entanglement between variables can ever be produced.

The program \symbolindexmark\init{$\init\xx$}
initializes $\xx$
with the quantum state $\ket0$.
(Remember that we assumed that every variable type contains a
distinguished element $0$.) 

The program \symbolindexmark\ifte{$\ifte\yy\bc\bd$}
will measure the qubit $\yy$,
and, if the outcome is $1$, execute~$\bc$, otherwise execute~$\bd$.

The program \symbolindexmark\while{$\while\yy\bc$}
measures $\yy$,
and if the outcome is $1$,
it executes $\bc$. This is repeated until the outcome is $0$.

Finally, \symbolindexmark;{$\bc;\bd$}
executes $\bc$
and then $\bd$.
And \symbolindexmark\SKIP{$\SKIP$}
does nothing. We will always implicitly treat ``$;$''\pagelabel{page:seq.assoc.skip.neutral}
as associative and $\SKIP$ as its neutral element.

\fullonly{\paragraph{On the minimalism of the language.}
This language is intentionally minimalistic. It seems to lack a number
of features that are present, e.g., in \cite{qrhl}: Initializing
variables with states other than $\ket0$.
Performing measurements. Probabilism (i.e., random
sampling). Parameterizing operations/states using classical variables
(e.g., $\apply{R_\xx}\yy$
where $\xx$
is a variable of type $\setR$,
and $R_\theta$
a rotation by angle $\theta$).
All these features are very important if we want to model anything but
the simplest programs. Yet, as we will see, these features are not
actually lacking. Using syntactic sugar (introduced in this section
and in \autoref{sec:deriv:classical}), we can recover all those
features. Keeping the language minimal and encoding all advanced
features allows us to get a much simpler core logic. Rules for working
with the advanced features can then be derived from the core features.}

\begin{shortversion}
    \paragraph{Semantics.}
The \emph{denotational semantics}%
\index{denotational semantics}%
\index{semantics!denotational} of our programs $\bc$
are represented as functions
\symbolindexmark\denot{$\denot\bc$} on the mixed memories over $\XXall$.
That is, if the program $\bc$ is in state $\rho$ initially,
then it is in state $\denot\bc(\rho)$ after execution.
The actual definition of the semantics $\denot\bc$ is straightforward and unsurprising
and deferred to \autoref{app:semantics}.
\end{shortversion}

\paragraph{Semantics.}
The \emph{denotational semantics}%
\index{denotational semantics}%
\index{semantics!denotational} of our programs $\bc$
are represented as functions
\symbolindexmark\denot{$\denot\bc$} on the mixed memories over
$\XXall$, defined
by recursion on the structure of the programs.  Here \symbolindexmark\XXall{$\XXall$}\pagelabel{page:XXall}
is a fixed set of program variables, and we will assume that
$\fv(\bc)\subseteq\XXall$
for all programs in this paper.\fullonly{\footnote{We fix some set $\XXall$
  in order to avoid a more cumbersome notation $\denot\bc^{\XX}$
  where we explicitly indicate the set $\XX$ of program
  variables with respect to which the semantics is defined.}}
The obvious cases are $\denot\SKIP:=\id$
and $\denot{\bc;\bd}:=\denot\bd\circ\denot\bc$.
And application of an isometry $U$
is also fairly straightforward given the syntactic sugar introduced
above:
$\denot{\apply U\XX}(\rho):=\paren{\opon U\XX}\rho\adj{\paren{\opon
    U\XX}}$.

Initialization of a quantum variable is slightly more complicated:
$\init\xx$
initializes the variable $\xx$
with $\ket0$,
which is the same as removing $\xx$,
and then creating a new variable $\xx$
with content $\ket0$.
Removing $\xx$
is done by the operation $\partr\xx$
(partial trace, see \autopageref{page:partr}). And creating a new variable $\xx$
is done by the operation $\otimes\proj{\ket0_{\xx}}$.
Thus we define $\denot{\init\xx}(\rho):=\partr\xx\rho\otimes\proj{\ket0_{\xx}}$.

The if-command first performs a measurement and then branches. A
measurement is described by one projector for each outcome. In our
case, $\opon{\proj{\ket i}}\yy$
corresponds to outcome $i=0,1$.
We then have that the state after measurement (without
renormalization) is
$\paren{\opon{\proj{\ket i}}\yy}\rho\adj{\paren{\opon{\proj{\ket i}}\yy}}$. Then $\bc$
or $\bd$
is applied to that state and the resulting states are added together
to get the final mixed state. Altogether:\symbolindexmarkonly\restrict
\begin{multline*}
  \denot{\ifte\yy\bc\bd}(\rho)
  :=
  \denot\bc\pb\paren{\restrict1(\rho)}
  +
  \denot\bd\pb\paren{\restrict0(\rho)}
  \\
  \text{where}\qquad
  \symbolindexmarkhighlight{\restrict i(\rho)}:=
  \paren{\opon{\proj{\ket i}}\yy}\rho\adj{\paren{\opon{\proj{\ket i}}\yy}}
\end{multline*}
While-commands are modeled similar: In an execution of a while
statement, we have $n\geq 0$
iterations of ``measure with outcome $1$
and run $\bc$''
(which applies $\denot\bc\circ\restrict1$
to the state), followed by ``measure with outcome $0$''
(which applies $\restrict0$
to the state). Adding all those branches up, we get the
definition:
\begin{equation*}
  \denot{\while\yy\bc}(\rho) :=
  \sum_{n=0}^\infty \restrict0 \pb\paren{ (\denot\bc\circ\restrict1)^n(\rho) }
\end{equation*}

\paragraph{Syntactic sugar.} \fullshort{To work productively with the minimal
language from above, we introduce some syntactic sugar:}{Our language
is intentionally minimalistic because we can derive more complex
statements using syntactic sugar:}
\begin{itemize}
\item \symbolindexmark\inits{$\inits\XX\psi$}
  (initialization with quantum state): To assign a quantum state
  $\psi\in\elltwo T$
  to variables $\XX$ of type $T$,
  we have to do the following: We fix an isometry $U_\psi$
  with $U\ket{0,\dots,0}=\psi$.
  And then we initialize all $\XX$
  with $\ket0$
  and apply $U$.
  That is, $\inits\XX\psi$
  abbreviates ``$\init{\xx_1};\dots;\init{\xx_n};\apply{U_\psi}\XX$''
  with $\xx_1\dots\xx_n:=\XX$
  and some arbitrary isometry $U_\psi\ket{0,\dots,0}:=\psi$.
  (Such $U_\psi$ is not unique but always exists.)
\item \symbolindexmark\initc{$\initc\XX z$}
  (classical initialization / assign-statement): This is short for
  $\inits\XX{\ket z}$. (We assume that $z$ is in the type of $\XX$.)
\item \symbolindexmark\measure{$\measure\YY\XX$}\pagelabel{page:measure}
  (measurement).  We wish to simulate a measurement in the
  computational basis using the commands from our minimal language.
  It is a well-known (and easy to check) fact that measuring $\XX$
  and assigning $\ket z$
  to $\YY$
  (where $z$
  is the outcome) is equivalent to performing a CNOT from $\XX$ to a
  $\ket0$-initialized $\YY$ and to a $\ket0$-initialized auxiliary
  register and discarding the auxiliary register. This can be expressed
  using our language: We define $\measure\YY\XX$
  to denote
  ``$\inits\YY{\ket0};\inits{\zz}{\ket0};\apply\CNOT{\XX\YY};\apply{\CNOT}{\XX\zz};\inits{\zz}{\ket0}$''
  where $\zz$
  is a fresh variable of the same type as $\YY$
  and $\XX$.
  (I.e., $\zz$ is a variable that is used nowhere else.)

  Similarly, we can also define a measurement of $\XX$
  that does not remember the outcome. (That is, its effect is merely
  to change the measured variables.) We write \symbolindexmark\measuref{$\measuref\XX$}\pagelabel{page:measuref}
  to denote
  ``$\inits{\zz}{\ket0};\apply{\CNOT}{\XX\zz};\inits{\zz}{\ket0}$''
  where $\zz$
  is a fresh variable of the same type as $\XX$.

  \fullonly{(Of course, it is also possible to model measurements other than
  computational basis measurements. To implement a projective
  measurement described by projectors $\{P_i\}_i$,
  we simply replace $\CNOT$
  by the unitary $\sum_i P_i\otimes U_{\oplus i}$
  where $U_{\oplus i}:\ket z\mapsto\ket{z\oplus i}$.
  We do not fix a specific syntax for this construction.)}
\item \symbolindexmark\sample{$\sample\XX D$}
  (random sampling). Here $D$
  is a discrete probability distribution over $T$,
  the type of $\XX$.
  Sampling for $D$
  is easily done by initializing $\XX$
  in the state \symbolindexmarkonly\psiD$\symbolindexmarkhighlight{\psiD D}:=\sum_{i\in T}\sqrt{D(i)}\ket{i}$
  and then measuring that state in the computational basis (leaving
  $\XX$
  in state $\ket i$
  with probability $D(i)$).
  That is, $\sample\XX D$
  is shorthand for ``$\inits\XX{\psiD D};\measuref\XX$.''
\end{itemize}

\section{Hoare Logic with Ghosts}
\label{sec:hoare.ghosts}

\paragraph{Recap Hoare logic.}\pagelabel{page:recap.hoare}
Before we introduce our Hoare logic with ghost variables, we quickly
recap regular quantum Hoare logic.\fullonly{\footnote{Strictly speaking,
  ``recap'' is not the right word since as far as we know this variant
  of quantum Hoare logic has not explicitly been spelled out in the
  literature. However, we still consider it folklore because it is a
  relatively simple generalization of \cite{brunet03dynamic} (allowing for more
  general programs, density operator based semantics, and changing the
  presentation from weakest precondition transformers to Hoare
  triples), it is a simplification of \cite{qrhl} (which
  considers pairs of programs instead of single programs), and it is a
  special case of \cite{ying12floyd} (by considering only ``strict''
  predicates there, i.e., predicates that are projectors, we get a
  logic that is roughly the same).}}  Intuitively, a Hoare triple
$\hl \PA\bc\PB$
means: If the initial state $\rho$
of the program $\bc$
satisfies~$\PA$,
and we run the program $\PB$,
then the final state $\denot\bc(\rho)$
satisfies~$\PB$.
Since the states of programs in our semantics are mixed 
memories $\rho$
(i.e., density operators), we need to first define what it means for a
mixed memory to satisfy a predicate. For this, the notion of
support of a density operator comes in handy: A density operator can
always be represented as $\rho=\sum_i\proj{\psi_i}$,
and intuitively this means that $\rho$
is a mixture of states $\psi_i$.
(But note that this decomposition is not unique!) Then
\symbolindexmarkonly\suppo$\symbolindexmarkhighlight{\suppo\rho}:=\SPAN\{\psi_i\}_i$
is simply the subspace spanned by all the vectors that constitute
$\rho$.\fullonly{\footnote{The
  usual formal definition of $\suppo\rho$
  is $\suppo\rho:=\im P$
  where $P$
  is the smallest projector such that $P\rho\adj P=\rho$.
  (This definition has the advantage of not requiring a specific
  choice of decomposition of $\rho$.)
  But it is easy to verify that this definition coincides with
  $\SPAN\{\psi_i\}_i$.}}
(Fortunately, this definition turns out to be independent of the
choice of $\psi_i$.)
Now, if $\PA$
is a predicate over the program variables $\XXall$
(formally: a subspace of $\elltwov\XXall$),
and $\rho$
is a mixed memory over $\XXall$,
then $\rho$
satisfies $\PA$
(written $\sats\rho\PA$)
iff $\suppo\rho\subseteq\PA$.
(I.e., iff $\rho$
is a mixture of quantum memories in $A$.)
With this notation (that will be changed somewhat later to accommodate
ghost variables), we can formally define $\hl\PA\bc\PB$
as: for all $\sats\rho\PA$
we have $\sats{\denot\bc(\rho)}\PB$. (We call $\PA$ the \emph{precondition}%
\index{precondition}
and $\PB$ the \emph{postcondition}%
\index{postcondition}.)
For this logic, we can then prove a number of rules that allow us to
derive the behavior of a complex quantum program from the behavior of
its elementary building blocks. For example, the (very easy to
prove) \rulerefx{Seq} rule shows that $\hl\PA\bc\PB$
and $\hl\PB\bd\PC$
implies $\hl\PA{\bc;\bd}\PC$.
This allows us to break down the analysis of a sequence of commands
into an analysis of the individual commands. (And similar rules exist
for quantum operations, while-loops, etc.)

However, this Hoare logic is somewhat limited in its expressivity. For
example, we cannot express the fact that the variable $\xx$
is uniformly randomly distributed. Say $\bc$
samples $\sample\xx D$
where $D$
is the uniform distribution. Then the final state is
$\rho=\sum_{i\in T}\frac1{\abs T}\proj{\ket i}$
where $T$
is the type of~$\xx$,
and $\suppo\rho=\SPAN\{\ket i\}_{i\in T}=\elltwov{\xx}=\top$.
So the only postcondition for this $\bc$
is $\top$,
the trivial postcondition. So the above quantum Hoare logic forgets
about the distribution of $\xx$
and remembers only what values have non-zero probability. (I.e.,
probabilism is treated as possibilistic nondeterminism.)  For similar
reasons, we cannot express, say, that $\xx$
is classical (e.g., after the program $\measuref\xx$).

Extensions of this basic quantum Hoare logic can make some statements
about probabilities: Quantum Hoare logic whose pre-/postconditions are
expectations \cite{dhondt06weakest,ying12floyd} can, e.g., express
that a certain predicate will hold with a certain probability. And
\cite{qrhl} can express that the outputs of two programs are
identical, even taking into account their distributions. But both
still lack the possibility of stating, as part of a
pre-/postcondition, e.g., that a variable has a particular
distribution.
\fullonly{See also discussion in \autoref{sec:qotp.tricky} for
further discussion on the limitations of those~logics.}

\paragraph{Ghost variables.}\index{ghost variable}%
\index{variable!ghost}%
\index{ghost|shortfor{ghost variable}}
Our solution to this problem is the introduction of ``ghost
variables'' (which will will often simply call ``ghosts'' for
brevity). In our context, a ghost variable is a variable that cannot
occur in the program (nor in the memory of the program) but only in
predicates. The intuitive meaning of a ghost variable is that it can
take any value that makes a predicate true. To illustrate the idea,
let us forget about quantum programs for a moment and consider the
classical case: For example, the classical postcondition $\xx=\gggg^2$
would mean that after the execution of
the program, the variable $\xx$
contains the square of $\gggg$,
if $\xx$
and $\gggg$
are both program variables of type $\setN$. But if $\gggg$
is a ghost, then $\xx=\gggg^2$
is true whenever there is some way to assign an integer to $\gggg$
that makes $\xx=\gggg^2$
true. In other words, the postcondition $\xx=\gggg^2$
is equivalent to just saying that $\xx$
is a square.  Now, in the classical case this is not very impressive:
$\xx=\gggg^2$
is just equivalent to $\exists z.\, \xx=z^2$.
And any other predicate involving ghosts can also be
rewritten into a regular predicate by using existential
quantifiers. So, at least if we allow existential quantifiers in
predicates (and there is no reason why we should not), \emph{ghost
  variables are useless for classical Hoare logic}.\fullonly{\footnote{Which is,
  most likely, why they have not been considered before.}} However,
this argument does not apply in the quantum case. A quantum ghost
variable cannot just be simulated using an existential quantifier
(e.g., because ghost variables might be entangled with quantum
variables).

So, how can we formalize ghost variables in the quantum setting?  A
classical memory $\mm$
(containing only program variables) satisfies a predicate $\PA$
involving ghosts iff there exists a larger memory
$\mm^\circ$ containing both program and ghost variables such that
$\mm^\circ$
satisfies $\PA$,
and $\mm$
is the result of removing all ghosts from $\mm^\circ$.
The quantum analogue of removing variables is the partial trace. That
is, if we have a mixed memory $\rho$
on $\XX\GG$,
then $\rho_{\XX}:=\partr\GG\rho$
is the result of removing all ghosts $\GG$
from~$\rho$.
Thus, we are ready for our first tentative definition: $\sats\rho\PA$
iff there exists a density operator $\rho^\circ$ on $\XX\GG$ such that $\partr\GG\rho^\circ=\rho$ and
$\suppo\rho^\circ\subseteq\PA$.

Note that in the previous definition, the program variables $\XX$
and the ghost variables $\GG$
can be entangled in arbitrary ways (since we put no restriction on
$\rho^\circ$).
However, there is a different possibility of defining ghost variables:
We could additionally require that $\rho^\circ$
is $(\XX,\GG)$-separable.
That would mean that ghost and program variables may not be
entangled. This will lead to a very different behavior of ghost
variables. It will turn out that both variants have their uses, so in
our logic we will simply consider both variants: We consider two kinds
of ghost variables, entangled ghost variables $\EE$
and unentangled ghost variables $\UU$.
That is, $\PA$
may contain both entangled and unentangled ghosts, and
$\rho^\circ$
is required to be $(\XX\EE,\UU)$-separable.
This means that the variables $\UU$
cannot be entangled with the program variables $\XX$,
but the variable $\EE$ can be!

\paragraph{Formal definitions.}
We can now mold all these ideas into a formal definition:
\begin{definition}[Satisfying a predicate with ghosts]\label{def:satisfy}
  Let $\rho$
  be a mixed memory over $\XX$.
  Let $\PA$
  be a predicate over $\XX\EE\UU$.
  Then a density operator $\rho$ over $\XX$
  \emph{satisfies}%
  \index{satisfy!a predicate} $\PA$
  (written \symbolindexmark\sats{$\sats\rho\PA$})
  iff there exists a $(\XX\EE,\UU)$-separable mixed memory $\rho^\circ$ over $\XX\EE\UU$ such that
  $\suppo\rho^\circ\subseteq\PA$ and $\partr{\EE\UU}\rho^\circ=\rho$.
\end{definition}

Recall that we use different letters for different kinds of
variables (cf.~\autopageref{page:variable.conventions}),
so the above definition implicitly assumes that
$\XX$
are program variables, $\EE$
are entangled ghost variables, and $\UU$
are unentangled ghost variables. In the remainder of this work, we
assume that these conventions are understood.
Note that if $\EE=\UU=\varnothing$,
then \autoref{def:satisfy} specializes to the definition given in the recap above,
namely $\sats\rho\PA\iff\suppo\rho\subseteq\PA$.

Given the definition of satisfying a predicate, it is straightforward
to define our Hoare logic:
\begin{definition}[Hoare logic with ghosts]\label{def:hoare}
  Let $\PA$
  be a predicate over $\XXall\EE\UU$,
  and $\PB$
  a predicate over $\XXall\EE'\UU'$,
  and $\bc$ a program.

  Then \symbolindexmark\hl{$\hl\PA\bc\PB$}
  iff for all mixed memories $\rho$
  over $\XXall$
  with $\sats\rho\PA$, we have that $\sats{\denot\bc(\rho)}\PB$.
\end{definition}
Note that $\PA$
and $\PB$
do not need to use the same ghosts. Ghosts are local
to the interpretation of a given predicate. In particular, if ghost
variables are chosen in a particular way when showing $\sats\rho\PA$,
this does not mean that they have to be chosen in a related way in
$\sats\rho\PB$!

\paragraph{Example.}\pagelabel{page:pred.example} Consider the following situation. We have two
variables $\xx,\yy$
of type $\bit$.
Initially, they are entangled in the state
$\psi:=\fsq\ket{00}+\fsq\ket{11}$.
Now we initialize $\yy$
with $\ket0$.
What do we know about $\xx$?
The initial state of $\xx,\yy$
is represented by the predicate $\xx\yy\quanteq\psi$
in our notation (see \autopageref{page:quanteq}). So, we are asking for a predicate
$\PB$
involving $\xx$
such that $\hl{\xx\yy\quanteq\psi}{\init\yy}\PB$
holds. Since $\xx$
is initially entangled with $\yy$,
and $\yy$
is ``overwritten'' (thus effectively deleted), $\xx$
is afterwards entangled with a ghost (in a sense, the ghost
of the deleted $\yy$).
That is, $\PB=(\xx\ee\quanteq\psi)$.
(The \rulerefx{Init} rule below will allow us to make this reasoning
rigorous.) And, as we will see later, $\xx\ee\quanteq\psi$
means that $\xx$
is a uniformly random bit. So we have derived that after deleting half
of an EPR pair $\psi$,
we get a uniformly random bit. (This matches what we know about EPR
pairs.)

\paragraph{About the sets of ghost variables.} The careful reader may
have noticed that there is an ambiguity in our notation.  On
\autopageref{page:pred.identify}, we said that we identify predicates
over different sets of variables as long as they coincide on their
free variables. But that means that for a state $\rho$
over $\XX$,
when interpreting $\sats\rho\PA$
according to \autoref{def:satisfy}, we do not know what the sets $\EE,\UU$
are in that definition (we only know that
$\EE\UU\supseteq\fv(\PA)\setminus\XX$).
Fortunately, the following lemma shows that the choice of $\EE\UU$
is irrelevant, so our notational choice is justified. (Note that there
is no ambiguity concerning the set $\XX$
of program variables since that set is determined by the type of the
mixed memory $\rho$.)

\begin{lemma}[Irrelevance of sets of ghosts]\label{lemma:ghost.sets}
  Let $\rho$
  be a mixed memory over $\XX$.
  Let $\sats{\rho}{^{\EE\UU}\PA}$
  denote $\sats\rho\PA$
  (as in \autoref{def:satisfy}) where $\PA$
  is interpreted as a predicate over $\XX\EE\UU$.
  
  Assume that $\EE_1\UU_1, \EE_2\UU_2 \supseteq\fv(\PA)\setminus\XX$.
  Then $\sats\rho{^{\EE_1\UU_1}\PA}$ iff $\sats\rho{^{\EE_2\UU_2}\PA}$.
\end{lemma}

\begin{proof}
  For clarity, we write $\PA_{\VV}$ when we interpret $\PA$ as a  predicate over $\VV$.

  Due to symmetry, we only need to prove
  $\sats\rho{^{\EE_1\UU_1}\PA}\implies\sats\rho{^{\EE_2\UU_2}\PA}$.
  Assume $\sats\rho{^{\EE_1\UU_1}\PA}$.
  Then there exists a $(\XX\EE_1,\UU_1)$-separable
  mixed memory $\rho^\circ$
  over $\XX\EE_1\UU_1$
  with $\suppo\rho^\circ\subseteq\PA_{\XX\EE_1\UU_1}$
  and $\partr{\EE_1\UU_1}\rho^\circ=\rho$.
  
  Let
  $\tilde\rho^\circ:=\partr{\EE_1\UU_1\setminus\EE_2\UU_2}\rho^\circ$.
  Then $\tilde\rho^\circ$
  is a mixed memory over
  $\XX(\EE_1\cap\EE_2)(\UU_1\cap\UU_2)$.
  And since $\rho^\circ$
  was $(\XX\EE_1,\UU_1)$-separable,
  $\tilde\rho^\circ$
  is $(\XX(\EE_1\cap\EE_2),\UU_1\cap\UU_2)$-separable.
  And since $\suppo\rho^\circ\subseteq\PA_{\XX\EE_1\UU_1}$,
  we have
  $\suppo\tilde\rho^\circ\subseteq\PA_{\XX(\EE_1\cap\EE_2)(\UU_1\cap\UU_2)}$. (The last step uses that
  $\XX(\EE_1\cap\EE_2)(\UU_1\cap\UU_2)\supseteq\fv(\PA)$.)

  Let
  $\hat\rho^\circ:=\tilde\rho^\circ \otimes
  \sigma_{\EE_2\setminus\EE_1} \otimes \sigma_{\UU_2\setminus\UU_1}$
  where $\sigma_{\EE_2\setminus\EE_1}$
  and $ \sigma_{\UU_2\setminus\UU_1}$
  are arbitrary mixed memories of trace $1$
  over $\EE_2\setminus\EE_1$
  and $\UU_2\setminus\UU_1$,
  respectively.  Then $\hat\rho^\circ$
  is is
  $(\XX\EE_2,\UU_2)=(\XX(\EE_1\cap\EE_2)(\EE_2\setminus\EE_1),(\UU_1\cap\UU_2)(\UU_2\setminus\UU_1))$-separable
  since $\tilde\rho^\circ$
  is $(\XX(\EE_1\cap\EE_2),\UU_1\cap\UU_2)$-separable.
  And
  $\suppo\hat\rho^\circ
  \subseteq\suppo\tilde\rho^\circ\otimes\elltwov{\EE_2\UU_2\setminus\EE_1\UU_1}
  \subseteq\PA_{\XX(\EE_1\cap\EE_2)(\UU_1\cap\UU_2)}\otimes\elltwov{\EE_2\UU_2\setminus\EE_1\UU_1}
  =\PA_{\XX\EE_2\UU_2} $.
  Furthermore
  \[
    \partr{\EE_2\UU_2} \hat\rho^\circ
    =
    \partr{(\EE_1\cap\EE_2)(\UU_1\cap\UU_2)} \tilde\rho^\circ
    =
    \partr{(\EE_1\cap\EE_2)(\UU_1\cap\UU_2)}
    \partr{\EE_1\UU_1\setminus\EE_2\UU_2} \rho^\circ
    =
    \partr{\EE_1\UU_1} \rho^\circ = \rho.
  \]
  Thus $\sats\rho{^{\EE_2\UU_2}\PA}$.
\end{proof}

\section{Predicates with Ghosts}
\label{sec:pred.ghosts}

In this section, we describe three important kinds of predicates that
can be expressed using ghosts. 

\subsection{Variables with a certain distribution}

First, we show that entangled ghosts can be used to express
that a variable has a certain distribution. Given a distribution $D$
on $T$,
we define
\symbolindexmarkonly\psiDD$\symbolindexmarkhighlight{\psiDD D}:=\sum_i
\sqrt{D(i)}\, \ket{i}\otimes\ket i$, a state on two variables of type $T$.
This state has the property that, if we erase (or measure) the second
part, we get a $D$-distributed
classical value $\ket i$
in the first part. So, if $\xx\ee$
is in state $\psiDD D$,
then, since $\ee$
is a ghost, $\ee$
is, in effect, erased. Thus the predicate $\xx\ee=\psiDD D$
effectively means that $\xx$
is $D$-distributed. Thus we introduce syntactic sugar for predicates:
\begin{itemize}
\item \symbolindexmark\distr{$\distr\XX D$}
  ($\XX$
  is $D$-distributed).
  $\distr\XX D$
  is short for $\XX\ee \quanteq \psiDD D$
  where $\ee$
  is a fresh entangled ghost (i.e., one that does not occur
  elsewhere in the predicate we are formulating) of the same type as $\XX$.
\item \symbolindexmark\uniform{$\uniform\XX$}
  ($\XX$
  is uniformly distributed).  This is short for $\distr\XX D$
  where $D$ is the uniform distribution on the type of $\XX$.
\end{itemize}

(As a special case, if $D$
is the uniform distribution on a single bit, then $\psiDD D$
is the state $\psi$
from the example in the previous section.  So the
postcondition $\PB$
in that example can indeed be written as $\uniform\xx$
as was already hinted there.)

So far, we gave only a relatively hand-waving explanation why
$\distr\XX D$
means that $\XX$
is $D$-distributed. But the following lemma makes this formal:
\begin{lemma}[Distribution predicates]\label{lemma:distrib}
  Let $D$
  be a distribution over $T$.
  Let $\rho$
  be a mixed memory over $\YY$.
  Let $\XX\subseteq\YY$
  have type $T$.
  Let $\rho_D:=\sum_i D(i)\proj{\ket i}$
  be a mixed memory over $\XX$.
  (I.e., $\rho_D$
  contains a $D$-distributed
  classical value $i$.) Then the following are equivalent:
  \begin{compactitem}
    \item $\sats\rho{\distr\XX D}$.
    \item There exists a mixed memory $\rho'$
      over $\YY\setminus\XX$ such that $\rho=\rho'\otimes\rho_D$.
  \end{compactitem}
\end{lemma}
In other words, $\distr\XX D$
means that $\XX$
is $D$-distributed \emph{and independent of other variables}.

\begin{proof}
  First, we check that $\rho_D$
  is the result of removing variable $\ee$
  from the quantum memory $\psiDD D$ 
  (we interpret $\psiDD D$ as a quantum memory over $\XX\ee$):
  \begin{equation}
    \label{eq:psiDD.rhoD}
    \partr{\ee}\proj{\psiDD D} =
    \partr{\ee}\pB\proj{\sum\nolimits_i\sqrt{D(i)}\ket i_{\XX}\otimes\ket i_{\ee}}
    = \sum\nolimits_i D(i) \proj{\ket i_{\XX}} = \rho_D.
  \end{equation}
  $\distr\XX D$
  is syntactic sugar for $\XX\ee\quanteq\psiDD D$
  and thus is a predicate over some variables $\YY\EE\UU$
  with $\XX\subseteq\YY$ and $\ee\in\EE$.

  \medskip\noindent\textbf{``$\pmb\Longrightarrow$'':}
  First we show that if $\sats\rho{\distr\XX D}$,
  then $\rho=\rho'\otimes\rho_D$
  for some $\rho'$.
  Since $\sats\rho{\distr\XX D}$,
  there exists a $\rho^\circ$
  over $\YY\EE\UU$
  with $\suppo\rho^\circ\subseteq\distr\XX D$ and $\partr{\EE\UU}\rho^\circ=\rho$. Thus $\rho^\circ=\sum_i\proj{\psi_i}$ for some $\psi_i$ with
  \[
    \psi_i \in \distr\XX D
    =
    (\XX\ee\quanteq\psiDD D)
    =
    \SPAN\braces{\psiDD D}\otimes\elltwov{\XX\ee^\complement}.
  \]
  Hence $\psi_i=\psiDD D\otimes\psi_i'$ for some $\psi_i'$ over $\YY\EE\UU\setminus\XX\ee$. Thus
  \[
    \rho
    =
    \partr{\EE\UU}\rho^\circ
    =
    \partr{\EE\UU}\sum_i\proj{\psiDD D}\otimes\proj{\psi_i'}
    =
    \partr{\ee}\proj{\psiDD D} \otimes
    \underbrace{
      \partr{\EE\UU\setminus\ee}\sum\nolimits_i\proj{\psi_i'}
    }_{=:\rho'}
    \eqrefrel{eq:psiDD.rhoD}=
    \rho_D\otimes\rho'.
  \]
  This shows the $\Longrightarrow$-direction.
  
  \medskip\noindent\textbf{``$\pmb\Longleftarrow$'':}
  We next show that if $\rho=\rho'\otimes\rho_D$,
  then $\sats\rho{\distr\XX D}$.
  Let
  $\rho^\circ := \rho' \otimes \proj{\psiDD D} \otimes
  \sigma_{\EE\setminus\ee} \otimes\sigma_{\UU}$ for some arbitrary
  mixed memories $\sigma_{\EE\setminus\ee},\sigma_{\UU}$
  of trace $1$
  on $\EE\setminus\ee$
  and $\UU$,
  respectively. The $\rho^\circ$
  is $(\YY\EE,\UU)$-separable
  and
  $\suppo\rho^\circ\subseteq \SPAN\braces{\psiDD
    D}\otimes\pb\elltwov{\paren{\XX\ee}^\complement} = \distr\XX D$. Furthermore,
  $\partr{\EE\UU}\rho^\circ=\rho'\otimes\partr\ee\proj{\psiDD
    D}\eqrefrel{eq:psiDD.rhoD}=\rho'\otimes\rho_D$.
  This shows the $\Longleftarrow$-direction.
\end{proof}

We will see examples of this predicate in the \ruleref{Sample} for
sampling statements ($\sample\yy D$), and in our analysis of the quantum one-time-pad in
\autoref{sec:sec.qotp}.

\subsection{Separable variables}

A concept specific to the quantum setting is for a variable to be
separable, i.e., not entangled with any other variables. (But a
separable variable may be probabilistically correlated!)

Expressing that a variable $\xx$
is separable seems, at the first glance, impossible to do using
predicates (that are modeled as subspaces): Such a predicate would
have to contain, e.g., the states $\ket{00}_{\xx\yy}$
and $\ket{11}_{\xx\yy}$
(since in both cases, $\xx$
and $\yy$
are equal but not entangled) but not the state
$\fsq\ket{00}_{\xx\yy}+\fsq\ket{11}_{\xx\yy}$.
But that would mean that the predicate is not closed under linear
combinations, hence not a subspace.

Yet, by introducing ghosts, we can model separable variables.
To understand how, we first need to recall a concept from
\cite{qrhl}, namely the quantum equality (between two variables):

\begin{definition}[Quantum equality {\cite[Defs.~\ref*{qrhl=>def:quanteq}, \ref*{qrhl=>def:quanteq.simple}]{qrhl}}]%
  \index{quantum equality}%
  \index{equality!quantum}
  Let $\WW,\WW'\subseteq\VV$
  be disjoint lists of quantum variables. ($\WW$ and $\WW'$ have the
  same type.)  Let \symbolindexmark\SWAP{$\SWAP$}
  be the unitary that swaps the content of $\WW$
  and $\WW'$.
  That is,
  $\SWAP\pb\paren{\ket i_{\WW}\otimes\ket j_{\WW'}\otimes \psi''} :=
  \ket j_{\WW}\otimes\ket i_{\WW'}\otimes \psi'' $ for all $i,j\in T$
  and all quantum memories $\psi''$ over $\VV\setminus\WW\WW'$.

  Then \symbolindexmark\QUANTEQ{$\WW\QUANTEQ\WW'$}
  is the set of all quantum memories $\psi$
  on $\VV$
  such that $\SWAP\psi=\psi$.\fullonly{\footnote{The
    original definition of $\QUANTEQ$
    is more general because we can write something like
    $U\xx\QUANTEQ U'\yy$
    meaning that $\xx$
    and $\yy$
    are equal up to operations $U,V$.
    Since we will no explicitly make use of this in this paper, we
    only gave the definition of the special case here. But the more
    general definition is, of course, also admissible in predicates as
    defined here.}}
\end{definition}

\fullonly{(\cite{qrhl} also presents a number of useful lemmas for rewriting and
  simplifying predicates involving $\QUANTEQ$.)}

In other words, we consider $\WW$
and $\WW'$
to have equal content ($\WW\QUANTEQ\WW'$)
iff a state is invariant under swapping $\WW$
and $\WW'$.
Now, it turns out that if $\WW\QUANTEQ\WW'$,
but $\WW$
and $\WW'$
are not entangled with each other, then $\WW$
and $\WW'$
also cannot be entangled with any other variables:
\begin{lemma}[Quantum equality \& separable states {\cite[Coro.~\ref*{qrhl=>coro:quanteq}]{qrhl}}]\label{lemma:quanteq.sep}
  Fix quantum memories $\psi$
  over $\VV\supseteq\WW$
  and $\psi'$
  over $\VV'\supseteq\WW'$.
  Then $\psi\otimes\psi'\in(\WW\QUANTEQ\WW')$
  iff $\psi$
  and $\psi'$
  are of the form $\psi=\psi_{\WW}\otimes\psi_{\VV\setminus\WW}$
  and $\psi'=\psi'_{\WW'}\otimes\psi'_{\VV'\setminus\WW'}$
  and $\psi_{\WW}=\psi'_{\WW'}$\,\footnote{Up to renaming of variables, formally
    $\psi_{\WW}=\Urename{\WW'}{\WW}\,\psi'_{\WW'}$}
  for some quantum memories
  $\psi_{\WW},\psi_{\VV\setminus\WW},\psi_{\WW'},\psi_{\VV'\setminus\WW'}$
  over $\WW,\VV\setminus\WW,\WW',\VV'\setminus\WW'$, respectively.
\end{lemma}

But this means that a program variable $\xx$
is separable iff $\xx\QUANTEQ\uu$
for some unentangled ghost! (Remember from
\autoref{def:satisfy} that an unentangled ghost $\uu$
will, by definition, not be entangled with $\xx$)
Thus we can introduce the following syntactic sugar for predicates:
\begin{itemize}
\item \symbolindexmark\separable{$\separable\XX$} ($\XX$ is separable).
    $\separable\XX$
  is short for $\XX\QUANTEQ\uu$
  where $\uu$
  is a fresh unentangled ghost (i.e., one that does not occur
  elsewhere in the predicate we are formulating) of the same type as $\XX$.
\end{itemize}

The following lemma formalizes our informal reasoning above,
$\separable\XX$ indeed characterizes separability:
\begin{lemma}[Separability predicates]\label{lemma:separable}
  Let $\rho$
  be a mixed memory on $\YY$.
  Let $\XX\subseteq\YY$. Then the following are equivalent:
  \begin{compactitem}
    \item $\sats\rho{\separable\XX}$.
    \item $\rho$ is $(\XX,\YY\setminus\XX)$-separable.
  \end{compactitem}
\end{lemma}

\begin{proof}
  $\separable\XX$
  is syntactic sugar for $\XX\QUANTEQ\uu$
  and thus is a predicate over some variables $\YY\EE\UU$
  with $\XX\subseteq\YY$ and $\uu\in\UU$.
  And $\XX$ and $\uu$ have the same type.

  \medskip\noindent\textbf{``$\pmb\Longrightarrow$'':}
  First we show that if $\sats\rho{\separable\XX}$,
  then $\rho$
  is $(\XX,\YY\setminus\XX)$-separable.
  Since $\sats\rho{\separable\XX}$,
  there exists a $(\YY\EE,\UU)$-separable
  $\rho^\circ$
  with $\suppo\rho^\circ\subseteq{\separable\XX}$
  and $\partr{\EE\UU}\rho^\circ=\rho$.
  Thus $\rho^\circ=\sum_i\proj{\psi_{\YY\EE,i}\otimes\psi_{\UU,i}}$
  for some $\psi_{\YY\EE,i},\psi_{\UU,i}$
  over $\YY\EE$
  and $\UU$, respectively. Then
  \[
    \psi_{\YY\EE,i}\otimes\psi_{\UU,i}
    \in\suppo\rho^\circ \subseteq \separable\XX = \paren{\XX\QUANTEQ\uu}.
  \]
  By \autoref{lemma:quanteq.sep}, this implies that
  $\psi_{\YY\EE,i}=\psi_{\XX,i}\otimes\psi_{\YY\EE\setminus\XX,i}$
  for some $\psi_{\XX,i},\psi_{\YY\EE\setminus\XX,i}$
  over $\XX$
  and $\YY\EE\setminus\XX$,
  respectively. Thus $\proj{\psi_{\YY\EE,i}}$
  is $(\XX,\YY\EE\setminus\XX)$-separable.
  Thus $\partr{\EE}\proj{\psi_{\YY\EE,i}}$ is  $(\XX,\YY\setminus\XX)$-separable.
  Hence
  \[
    \rho = \partr{\EE\UU} \rho^\circ = \sum_i \partr{\EE}\proj{\psi_{\YY\EE,i}}
  \]
  is $(\XX,\YY\setminus\XX)$-separable
  as well.  This shows the $\Longrightarrow$-direction.
  
  \medskip\noindent\textbf{``$\pmb\Longleftarrow$'':}
  We next show that if $\rho$
  is $(\XX,\YY\setminus\XX)$-separable,
  then $\sats\rho{\separable\XX}$.
  Since $\rho$
  is $(\XX,\YY\setminus\XX)$-separable,
  $\rho=\sum_i\proj{\psi_{\XX,i}\otimes\psi_{\YY\setminus\XX,i}}$
  for some quantum memories $\psi_{\XX,i},\psi_{\YY\setminus\XX,i}$
  over $\XX$ and $\YY\setminus\XX$, respectively. Without loss of generality, $\norm{\psi_{\XX,i}}=1$.
  Let
  \[
    \psi^\circ_i := \psi_{\XX,i} \otimes \psi_{\YY\setminus\XX,i}
    \otimes
    \psi_{\EE,i}
    \otimes
    \psi_{\uu,i}
    \otimes
    \psi_{\UU\setminus\uu,i}
  \]
  where $\psi_{\EE,i},\psi_{\UU\setminus\uu,i}$
  are arbitrary quantum memories of norm $1$
  on $\EE$
  and $\UU\setminus\uu$,
  respectively, and $\psi_{\uu,i}:=\psi_{\XX,i}$
  except that $\psi_{\uu,i}$
  is a quantum memory over $\uu$
  and not over $\XX$.
  By \autoref{lemma:quanteq.sep},
  $\psi_i^\circ\in\paren{\XX\QUANTEQ\uu}$.
  Let $\rho^\circ:=\sum_i\proj{\psi_i^\circ}$.
  Then $\rho^\circ$
  is $(\YY\EE,\UU)$-separable
  and
  $\suppo\rho^\circ=\SPAN\{\psi_i^\circ\}_i\subseteq\paren{\XX\QUANTEQ\uu}=\separable\XX$. And
  \[
    \partr{\EE\UU}\rho^\circ = \sum_i \proj{\psi_{\XX,i}\otimes\psi_{\YY\setminus\XX,i}} = \rho.
  \]
  (Using that $\psi_{\EE,i},\psi_{\uu,i},\psi_{\UU\setminus\uu,i}$ all have norm $1$.)
  Hence $\sats\rho{\separable\XX}$.
  This shows the $\Longleftarrow$-direction.
\end{proof}

As an example for the relationship between different predicates,
notice that $\distr\xx D$
implies $\separable\xx$
since $\distr\xx D$
implies that $\xx$
is distributed independently from all other variables
(\autoref{lemma:distrib}). \fullonly{This also follows within our logic by an
application of the \ruleref{Transmute} below,
see the example after \ruleref{Transmute}.}

\subsection{Classical variables}

A \fullshort{third}{second} application of ghost variables is to formulate predicates
that imply that a variable has a classical state. We say a mixed
memory $\rho$
over $\YY$ \emph{is classical in $\XX\subseteq\YY$}%
\index{classical!in $\XX$}
iff it is of the form $\rho=\sum_i\proj{\ket i_{\XX}}\otimes\rho_i$
for some mixed memories $\rho_i$
over $\YY\setminus\XX$.
(This is often called a \emph{cq-state}\index{cq-state}.)

Expressing that a variable $\xx$
is classical seems, at the first glance, impossible to do using
predicates (that are modeled as subspaces): Such a predicate would
have to contain, e.g., the states $\ket 0_{\xx}$
and $\ket 1_{\xx}$
(since those are classical) but not the state
$\fsq\ket 0_{\xx}+\fsq\ket 1_{\xx}$.
But that would mean that the predicate is not closed under linear
combinations, hence not a subspace.

Yet, by introducing ghosts, we can model classicality. In
order to see how, we introduce \fullshort{a different}{an} equality notion between
quantum variables, $\CLASSEQ$.
Intuitively, two variables are classically equal iff measuring both in
the \emph{computational} basis will always give the same outcome. (So,
$\ket0_{\xx}$
and $\ket0_{\yy}$
would be classically equal, but $\fsq\ket 0_{\xx}+\fsq\ket 1_{\xx}$
and $\fsq\ket 0_{\yy}+\fsq\ket 1_{\yy}$
would not be.\fullonly{\footnote{\label{foot:classeq.entangled}Somewhat counterintuitively, $\xx$
and $\yy$
are also classically equal if they are in the entangled state
$\fsq\ket{00}_{\xx\yy}+\fsq\ket{11}_{\xx\yy}$.
But this will not matter in our setting since we will apply $\CLASSEQ$ only
to unentangled variables.}})
Formally:
\begin{definition}[Classical equality]%
  \index{classical equality}%
  \index{equality!classical}
  Let $\WW,\WW'\subseteq\VV$
  be disjoint lists of quantum variables. ($\WW$ and $\WW'$ have the
  same type $T$.)  
  
  Then \symbolindexmark\CLASSEQ{$\WW\CLASSEQ\WW'$}
  is the span of all quantum memories of the form
  $\ket i_{\WW}\otimes\ket i_{\WW'}\otimes\psi$
  with $i\in T$
  and $\psi$ a quantum memory on $\VV\setminus\WW\WW'$.
\end{definition}
If we think of two variables $\xx,\uu$
both having the same state $\psi$,
then $\xx\CLASSEQ\uu$
holds if $\psi=\ket i$
for some $i$
(i.e., if $\psi$
is a classical state). But if $\psi$
is a superposition of different $\ket i$,
then measuring both $\xx$
and $\uu$
in the computational basis gives different results with non-zero
probability. Hence $\xx\not\CLASSEQ\uu$
in that case. This suggests that $\xx$
is classical iff it is classically equal to some unentangled ghost
variable $\uu$.\fullonly{\footnote{Classical
  equality to some entangled ghost would not be sufficient due to the
  situation described in \autoref{foot:classeq.entangled}.}} That is,
we introduce the following syntactic sugar for predicates:
\begin{itemize}
\item \symbolindexmark\class{$\class\XX$} ($\XX$ is classical).
    $\class\XX$
  is short for $\XX\CLASSEQ\uu$
  where $\uu$
  is a fresh unentangled ghost (i.e., one that does not occur
  elsewhere in the predicate we are formulating) of the same type as $\XX$.
\end{itemize}

The following lemma formalizes our informal reasoning above,
$\class\XX$ indeed characterizes classicality:
\begin{lemma}[Classicality predicates]\label{lemma:classical}
  Let $\rho$
  be a mixed memory over $\YY$.
  Let $\XX\subseteq\YY$. Then the following are equivalent:
  \begin{compactitem}
    \item $\sats\rho{\class\XX}$.
    \item $\rho$ is classical in $\XX$. (As defined at the beginning of this section.)
  \end{compactitem}
\end{lemma}

\begin{proof}
  $\class\XX$
  is syntactic sugar for $\XX\CLASSEQ\uu$
  and thus is a predicate over some variables $\YY\EE\UU$
  with $\XX\subseteq\YY$ and $\uu\in\UU$.

  \medskip\noindent\textbf{``$\pmb\Longrightarrow$'':}
  First we show that if $\sats\rho{\class\XX}$,
  then $\rho$
  is classical in $\XX$.
  Since $\sats\rho{\class\XX}$,
  there exists a $(\YY\EE,\UU)$-separable
  $\rho^\circ$
  with $\suppo\rho^\circ\subseteq{\class\XX}$ and $\partr{\EE\UU}\rho^\circ=\rho$.
  Thus $\rho^\circ=\sum_i\proj{\psi_{\YY\EE,i}\otimes\psi_{\UU,i}}$
  for some $\psi_{\YY\EE,i},\psi_{\UU,i}$
  over $\YY\EE$
  and $\UU$,
  respectively. Without loss of generality, $\norm{\psi_{\UU,i}}\neq0$ for all $i$.
  And
  $\psi_{\YY\EE,i}\otimes\psi_{\UU,i}\in\suppo\rho^\circ\subseteq\class\XX$.
  Fix some $i$.
  (We will omit $i$
  from the subscripts for now.)  We can write $\psi_{\YY\EE}$
  and $\psi_{\UU}$
  as
  $\psi_{\YY\EE}=\sum_{j}\lambda_{j}\ket
  j_{\XX}\otimes\psi_{\YY\EE\setminus\XX,j}$ and
  $\psi_{\UU}=\sum_j\lambda_{j}'\ket
  j_{\uu}\otimes\psi_{\UU\setminus\uu,j}$.  Since
  $\psi_{\UU}\neq0$,
  we have $\lambda_{j^*}'\neq0$
  for some $j^*$.
  If $\lambda_j\neq0$
  for some $j\neq j^*$,
  then $\psi_{\YY\EE}\otimes\psi_{\UU}$
  is not in the span of states
  $\ket\nu_{\XX}\otimes\ket\nu_{\uu}\otimes\dots$
  and thus not in $\class\XX=(\XX\CLASSEQ\uu)$.
  Hence $\lambda_j\neq0$
  for all $j\neq j^*$.
  Thus
  $\psi_{\YY\EE}=\ket{j^*}_{\XX}\otimes\lambda_{j^*}\psi_{\YY\EE\setminus\XX,j^*}$.
  Hence $\proj{\psi_{\YY\EE}}$
  is classical in $\XX$. Hence $\partr\EE\proj{\psi_{\YY\EE}}$ is classical in $\XX$.
  Now we ``unfix'' $i$.
  Thus all $\partr\EE\proj{\psi_{\YY\EE,i}}$ are classical in $\XX$. Hence
  \[
    \rho = \partr{\EE\UU}\rho^\circ = \sum_i\partr{\EE}\proj{\psi_{\YY\EE,i}}
  \]
  is classical in $\XX$.  This shows the $\Longrightarrow$-direction.

  \medskip\noindent\textbf{``$\pmb\Longleftarrow$'':}
  We next show that if $\rho$
  is classical in $\XX$
  then $\sats\rho{\class\XX}$.
  Since $\rho$
  is classical in $\XX$,
  $\rho=\sum_i\proj{\ket i_{\XX}}\otimes\proj{\psi_{\YY\setminus\XX,i}}$ 
  for some quantum memories $\psi_{\YY\setminus\XX,i}$
  over $\YY\setminus\XX$.
  Let
  \[
    \rho^\circ := \sum_i \proj{\ket i_{\XX}}\otimes\proj{\psi_{\YY\setminus\XX,i}}
    \otimes \sigma_{\EE} \otimes\proj{\ket i_{\uu}}
    \otimes \sigma_{\UU\setminus\uu}
  \]
  for arbitrary mixed memories
  $\sigma_{\EE},\sigma_{\UU\setminus\uu}$
  of trace $1$
  over $\EE$
  and $\UU\setminus\uu$,
  respectively.  Each summand has support in
  $\paren{\XX\CLASSEQ\uu}=\class\XX$, hence $\suppo\rho^\circ\subseteq\class\XX$. And $\rho^\circ$ is $(\YY\EE,\UU)$-separable. Finally,
  $\partr{\EE\UU}\rho^\circ=\sum_i\proj{\ket i}_{\XX}\otimes\proj{\psi_{\YY\setminus\XX,i}}=\rho$.
  Thus $\sats\rho{\class\XX}$.
  This shows the $\Longleftarrow$-direction.
\end{proof}

The predicate $\class\XX$ occurs for example in the rules
\rulerefx{Measure*}, \rulerefx{MeasureForget*}, and \rulerefx{Sample*}
for measurements and random sampling.
We discuss the predicate $\class\XX$
and its uses in greater depth in
\autoref{sec:deriv:classical}.

\section{Core Rules}
\label{sec:core.rules}

In this section, we present the core reasoning rules for our
logic. Since we have defined the logic semantically
(\autoref{def:hoare}), the set of rules is not fixed a priori (since
we can always prove additional rules sound). Nevertheless, we identify
a set of important rules (one per language primitive, plus some useful
structural rules) that form the basis of the rest of this paper. In
particular, all ``derived rules'' in \autoref{sec:deriv.rules} are a consequence
of these core rules. That is, after this section we can ``forget''
\autoref{def:hoare} and build only on the rules from this
section. (Convenient additional rules will be derived in later
sections as corollaries.)
\shortonly{All proofs for this section are deferred to \autoref{app:proofs-core-rules}.}

\subsection{Rules for individual statements}

For each command of our language (sequence, skip, initialization,
application, if, while), we introduce one rule that derives a Hoare
judgment for that command from judgments about its subterms. The
rules for sequence and skip are quite obvious and follow directly from
the definition:
\begin{ruleblock}
  \RULE{Seq}{
    \hl\PA\bc\PB \\
    \hl\PB\bd\PC
  }{
    \hl\PA{\bc;\bd}\PC
  }
  \RULE{Skip}{
    \PA \subseteq \PB
  }{
    \hl\PA\SKIP\PB
  }
\end{ruleblock}
More interesting are the rules for operations on quantum states (isometries, initialization):
\begin{ruleblock}
  \RULE{Apply}{}{
    \pb\hl\PA{\apply U\XX}{\oppred{(\opon U\XX)}\PA}
  }
  \RULE{Init}{}{
    \pb\hl\PA{\init\xx}{
      \psubst\PA\ee\xx,\
      {\xx\quanteq\ket0}}
  }
\end{ruleblock}
\rulerefx{Apply} says that applying an isometry $U$
to variables $\XX$
has the effect of multiplying the predicate $\PA$
with $U$
(after suitably lifting $U$
to operate on quantum memories, see \autopageref{page:opon}
for the definition of $\opon U\XX$). \rulerefx{Init} is more interesting
because it is the first rule that introduces ghosts. Since
initialization ``overwrites'' the original value of $\xx$,
$\xx$
becomes an engangled ghost, thus the precondition $\PA$
is replaced by $\psubst\PA\ee\xx$,
i.e., $\xx$
is replaced by a fresh ghost $\ee$.
($\ee$ is fresh, i.e., $\notin\fv(\PA)$, because otherwise $\AA\{\ee/\xx\}$ would
not be welltyped.)
Additionally, $\xx$
will afterwards be in the state~$\ket0$,
so the postcondition additionally contains $\xx\quanteq\ket0$.
\fullonly{The rules \rulerefx{Apply} and \rulerefx{Init} are shown in
\cref{lemma:Apply,lemma:Init} in \autoref{sec:core.proofs}.}

The rules for if and while do not introduce ghosts
and are the same as in ``regular'' quantum Hoare logic:
\begin{center}
  \smaller
  \noindent
  \RULE{If}{
    \pb\hl{\oppred{\pb\paren{\opon{\proj{\ket1}}\xx}}\PA}\bc\PB
    \\
    \pb\hl{\oppred{\pb\paren{\opon{\proj{\ket0}}\xx}}\PA}\bd\PB
  }{
    \hl\PA{\ifte\xx\bc\bd}\PB
  }
  \qquad
  \RULE{While}{
    \pb\hl{\oppred{\pb\paren{\opon{\proj{\ket1}}\xx}}\PA}\bc\PA
  }{
    \pb\hl\PA{\while\xx\bc}{\oppred{\pb\paren{\opon{\proj{\ket0}}\xx}}\PA}
  }
\end{center}
Since the if-statement measures $\xx$
before executing $\bc$
or $\bd$
(see \autoref{sec:qprogs}), the precondition $\PA$
becomes $\oppred{\pb\paren{\opon{\proj{\ket1}}\xx}}\PA$
when that measurement returns $1$
and $\bc$
is executed (as $\proj{\ket1}$
is the projector corresponding to measurement outcome $1$),
and it becomes $\oppred{\pb\paren{\opon{\proj{\ket0}}\xx}}\PA$
if the measurement returns $0$
and $\bd$
is executed. Thus analyzing $\ifte\xx\PA\PB$
reduces to analyzing $\bc$
and $\bd$ with those two respective preconditions.

Similarly, $\while\xx\bc$
executes $\bc$
after measuring $\xx$
and getting $1$.  Thus, if we use $\PA$
as the loop invariant, the postcondition for the loop body becomes
$\oppred{\paren{\opon{\proj{\ket1}}\xx}}\PA$.
And to end the loop, the measurement of $\xx$
must return $0$,
hence we get the postcondition $\oppred{\paren{\opon{\proj{\ket0}}\xx}}\PA$
for the overall loop.
\fullonly{The rules \rulerefx{If} and~\rulerefx{While} are proven in 
\cref{lemma:If,lemma:While} in \autoref{sec:core.proofs}.}

\subsection{Further core rules}

Besides the per-statement rules from the previous section, we will use
five more rules, related to case-distinctions and to the modification
of ghosts. First, we consider case-distinctions. In classical Hoare
logic, we can easily show the following rule:
$\forall z.\,\hl{\xx=z,\ \PA}\bc{\PB}\implies\hl \PA\bc\PB$.
That is, to show $\hl\PA\bc\PB$,
it is sufficient to consider each possible value $z$
of $\xx$
separately and prove $\hl\PA\bc\PB$
under the additional assumption that $\xx=z$
holds in the precondition. An immediate quantum analogue would be:
$\forall \psi.\,\hl{\xx\quanteq\psi,\ \PA}\bc{\PB}\implies\hl
\PA\bc\PB$. There are two problems with such this rule. First, it does
not hold in this generality: $\xx\quanteq\psi$
implies that $\xx$
is not entangled with any other variables (because it is in the
specific pure state $\psi$),
so proving $\hl{\xx\quanteq\psi,\ \PA}\bc{\PB}$
for all $\psi$
does not guarantee anything about the behavior of $\bc$
in the presence of entanglement.\fullonly{\footnote{Formally, a counterexample
  would be: $\PA:=(\xx\yy\quanteq\ket{00}+\ket{11})$,
  $\PB:=\bot$,
  and $\bc:=\SKIP$.
  Then for all $\psi$,
  $(\xx\quanteq\psi,\ \PA)=\bot$,
  hence $\hl{\xx\quanteq\psi,\ \PA}\SKIP\PB$.
  But $\hl\PA\bc\PB$
  does not hold.}}
And even if we fix this by adding suitable extra
conditions, the rule will force us to always quantify over all
possible $\psi$.
But if, for example, $\xx$
is guaranteed to be classical (e.g., $\PA=\class\xx$)
then we would like to only consider the cases $\xx=\ket z$.
To formulate a rule that solves both problems, we introduce an
additional concept:
\begin{definition}[Disentangling]\label{def:disentangling}
  A predicate $\PA$ on $\XX\UU$ is
  \emph{$M$-disentangling}\index{disentangling}
  (for a \emph{set} $M\subseteq\elltwov\XX$) iff:
  For all sets of variables $\VV$ (disjoint from $\XX\UU$),
  all quantum memories $\psi_{\VV\XX}\neq0$ over $\VV\XX$,
  and all quantum memories $\psi_{\UU}\neq0$ over $\UU$
  with $\psi_{\VV\XX}\otimes\psi_{\UU}\in \PA$,
  we have that $\psi_{\VV\XX}=\psi_{\VV}\otimes\psi_{\XX}$
  for some $\psi_{\VV}\in\elltwov{\VV}$ and some $\psi_{\XX}\in M$.
\end{definition}
What does this definition mean? Roughly speaking, it means that if
variables $\XX$
and $\UU$,
jointly, satisfy $\PA$,
and variables $\UU$
are not entangled with variables $\XX$
or $\VV$,
then we know that variables $\XX$
are also not entangled with variables $\VV$,
and additionally that variables $\XX$ will be in one of the states in $M$.

A trivial example would be $\PA:=(\xx\uu=\phi_1\otimes\phi_2)$
which is $\{\phi_1\}$-disentangling.
That is, if $\xx\uu$
are in state $\phi_1\otimes\phi_2$,
then $\xx$
is in state $\phi_1$
(unsurprisingly).  Similarly, for any non-separable $\phi$,
$S=\SPAN\{\phi\}$
is $\varnothing$-disentangling
(as the variables $\xx\uu$ cannot at the same time be non-entangled and in state $\phi$).
The following lemma gives \fullshort{two more interesting examples}{a more interesting example} of disentangling predicates:
\begin{lemma}\label{lemma:disentangling}
  Let $T$
  be the type of $\XX$.
  \fullonly{Then $\XX\QUANTEQ\UU$ and $\separable\XX$
  are $\elltwov\XX$-disentangling.}
  And $\XX\CLASSEQ\UU$ and $\class\XX$ are $\{\ket i\}_{i\in T}$-disentangling.
\end{lemma}

\begin{proof}
  We first show that $\XX\QUANTEQ\UU$
  is $\elltwov\XX$-disentangling.
  (This also implies that $\separable\XX$
  is $\elltwov\XX$-disentangling
  since $\separable\XX$
  is simply syntactic sugar for $\XX\QUANTEQ\uu$.)
  Fix some variables~$\VV$,
  and quantum memories $\psi_{\VV\XX}$
  and $\psi_{\UU}$
  over $\VV\XX$
  and $\UU$,
  respectively, with
  $\psi_{\VV\XX}\otimes\psi_{\UU}\in(\XX\QUANTEQ\UU)$.
  By \autoref{lemma:quanteq.sep} (with $\VV:=\VV\XX$,
  $\WW:=\XX$,
  $\VV':=\WW':=\UU$),
  this implies that $\psi_{\VV\XX}$
  can be written as $\psi_{\VV\XX}=\psi_{\VV}\otimes\psi_{\XX}$
  for some quantum memories $\psi_{\VV},\psi_{\XX}$
  over $\VV,\XX$,
  respectively. And trivially, $\psi_{\XX}\in\elltwov{\XX}$. Thus 
  $\XX\QUANTEQ\UU$
  is $\elltwov\XX$-disentangling by \autoref{def:disentangling}.

  \medskip
  
  Now we show that $\XX\CLASSEQ\UU$
  is $\{\ket i\}_{i\in T}$-disentangling.
  (This also implies that $\class\XX$
  is $\{\ket i\}_{i\in T}$-disentangling
  since $\class\XX$
  is simply syntactic sugar for $\XX\CLASSEQ\uu$.)
  Fix some variables~$\VV$,
  and quantum memories $\psi_{\VV\XX}$
  and $\psi_{\UU}$
  over $\VV\XX$
  and $\UU$,
  respectively, with
  $\psi_{\VV\XX}\otimes\psi_{\UU}\in(\XX\CLASSEQ\UU)$ and $\psi_{\UU}\neq0$.
  By \autoref{def:disentangling}, we need to show that
  $\psi_{\VV\XX}=\psi_{\VV}\otimes\ket i_{\XX}$
  for some $i\in T$
  and $\psi_{\VV}$
  over $\VV$.
  We decompose
  $\psi_{\VV\XX}=\sum\lambda_i\psi_{\VV,i}\otimes\ket i_{\XX}$
  and $\psi_{\UU}=\sum\lambda_i'\ket i_{\UU}$ for some $\psi_{\VV,i}\neq0$ over $\VV$.
  Since $\psi_{\UU}\neq0$,
  there exists a $j$
  such that $\lambda'_j\neq0$.
  If $\lambda_i\neq 0$
  for some $i\neq j$,
  then $\psi_{\VV\XX}\otimes\psi_{\UU}$
  is not in the span of states of the form
  $\ket \nu_{\XX}\otimes\ket \nu_{\UU}\otimes\dots$,
  in contradiction to
  $\psi_{\VV\XX}\otimes\psi_{\UU}\in(\XX\CLASSEQ\UU)$.
  Thus $\lambda_i=0$
  for all $i\neq j$, hence $\psi_{\VV\XX}=\lambda_j\psi_{\VV,j}\otimes\ket j_{\XX}$ as desired. Thus 
  $\XX\CLASSEQ\UU$
  is $\{\ket i\}_{i\in T}$-disentangling.
\end{proof}

Armed with the definition of disentangling predicates, we can formulate the rule for case
distinctions:
\[
  \RULE{Case}{
    \text{$\PC$ is $M$-disentangling predicate on $\XX\UU$}\\
    \PA\subseteq\PC \\
    \forall\psi\in M.\
    \hl{\XX\quanteq\psi,\ \PA}
    \bc\PB
  }{
    \hl{\PA}\bc\PB
  }
\]
As a special case (with $\PC:=\class\xx$
and using \autoref{lemma:disentangling}), we can recover a rule for
case distinction over classical variables:
$\forall z.\hl{\XX\quanteq\ket z,\ \class\XX,\ \PA}\bc\PB \implies
\hl{\class\XX,\ \PA}\bc\PB$.  See the derived \ruleref{CaseClassical} on
\autopageref{rule:CaseClassical} for details.  Notice that we would
not have been able to even state such a case rule without using ghosts!
\fullonly{The \rulerefx{Case} rule is proven in \autoref{lemma:Case} in
\autoref{sec:core.proofs}.}

The \rulerefx{Case} rule has the disadvantage that we need to have a
disentangling predicate in the precondition. As described above, this
is necessary because the variable under consideration might be
entangled with other variables. However, if we make a case distinction
over the state of \emph{all} variables, then this requirement
disappears. In fact, it turns out that it is enough to make a case
distinction over the state of the free variables in program and
pre-/postconditions plus one extra variable $\xx$ (this is not obvious because those variables might
still be entangled with other variables that are not used but
nevertheless present, even variables with uncountable type):
\[
  \RULE{Universe}{
    \XX\EE\UU\supseteq\fv(\PA,\bc)
    \\
    \XX\supseteq\progvars{\fv(\PB)}
    \\
    \xx\notin\XX
    \\
    \text{type of $\xx$ is infinite}
    \\\\
    \forall\psi\in\elltwov{\XX\xx\EE},\psi'\in\elltwov{\UU},\ \psi,\psi'\neq0.\
    \hl{\XX\xx\EE\quanteq\psi,\ \UU\quanteq\psi',\ \PA}\bc\PB
  }{
    \hl\PA\bc\PB
  }
\]
(We call this rule \rulerefx{Universe} since we do a case distinction
over the state of all variables, i.e., of the whole universe.)  \fullshort{We
will see an example where the \rulerefx{Universe} rule is useful in
the analysis of the quantum one-time pad (\autoref{sec:sec.qotp}, general case).}{
The \rulerefx{Universe} rule is used, for example, in the analysis of the quantum one-time pad.
(The general case that is deferred to \autoref{app:qotp-general}.)}
Note that it is important in this rule that we can fix one
concrete set $\XX\xx\EE\UU$
of variables to quantify over. Otherwise, we would have to quantify
over all states \emph{over all possible sets of variables}; depending
on the precise formalization the ``set'' of all possible sets of
variables might not even be a set, and a rigorous formalization of the
rule may not be possible in logical foundations that do not allow us
to quantify over large classes (e.g., higher-order logic as formalized
in Isabelle/HOL \cite{isabelle}).  \fullonly{The rule is proven in
\autoref{lemma:Universe} in \autoref{sec:core.proofs}.}

For stating the next rules more readably, we introduce another
notation: We write \symbolindexmark\impl{$\PA\impl\PB$}
for $\hl\PA\SKIP\PB$
(which in turn is equivalent to
$\forall\rho.\ \sats\rho\PA\implies\sats\rho\PB$).
By rules \rulerefx{Seq} and \rulerefx{Skip} (and the fact that $\SKIP$
is the neutral element of $;$)
we immediately have that $\impl$
is a preorder that refines $\subseteq$.
Also note that rule \rulerefx{Seq} implies that $\PA\impl\PA',\
\hl{\PA'}\bc{\PB'},\
\PB'\impl\PB
\implies\hl\PA\bc\PB$,
so $\impl$ can be used for rewriting Hoare judgments.

The next three rules are specific to ghost variables and allow us to
rewrite predicates. 
\begin{ruleblock}
  \RULE{Rename}{}{A\impl\psubstii\PA{\EE'}\EE{\UU'}\UU}
  \RULE{Transmute}{
    \forall i. \rank M_i \leq 1 \\
    \sum\nolimits_i \adj{M_i}{M_i} = \id \\\\
    \text{$\GG$ either all entangled or all unentangled ghosts}\\
    \text{$\GG'$ either all entangled or all unentangled ghosts}
  }{
    \PA
    \impl
    \bigvee\nolimits_{\!i} \pB\paren{
      \oppred{
      \paren{\opon{M_i}{\GG'}}
    }{
      \psubst\PA{\GG'}{\GG}
    }
    }
  }
  \RULE{ShapeShift}{
    \partr{\EE}\proj{\psi} =
    \partr{\EE'}\proj{\psi'}
    \\
    (\EE\cup\EE')\cap\fv(A)=\varnothing
  }{
    \paren{ \XX\EE\quanteq\psi,\ \PA}
    \impl
    \paren{ \XX\EE'\quanteq\psi',\ \PA }
  }
\end{ruleblock}
\Ruleref{Rename} simply allows us to rename ghosts, this mainly allows us to tidy up
judgments. \Ruleref{Rename} follows directly from \cref{def:satisfy,def:hoare}.
When reading \rulerefx{Transmute}, recall that $\GG,\GG'$
may refer to both entangled and unentangled
ghosts. The purpose of the \rulerefx{Transmute} rule is to change an
entangled ghost into an unentangled ghost or vice versa. (That is, we
will usually have $\GG=\ee$
and $\GG'=\uu$
or vice versa.) Ideally, we would like to have something like
$\PA\impl\psubst\PA{\uu}\ee=:\PA'$
and vice versa, i.e., being able to change the kinds of ghost
variables freely. But of course, that would mean that entangled and
unentangled ghosts are equivalent, and we would not have to had to
distinguish between those different kinds of variables in the first
place. Instead, we get a somewhat more complicate rule where, after
replacing $\uu$
by $\ee$
or vice versa, we also need to replace $\PA'$
by
$\bigvee\nolimits_{\!i} \pb\paren{ \oppred{ \paren{\opon{M_i}{\GG'}}
  }{ \PA' } }$. (Recall that $\vee$
is the disjunction of predicates, i.e., the sum of subspaces,
see \autopageref{page:vee}. Hence $\bigvee_{\!i}$
is a disjunction of a family of predicates
$ \oppred{ \paren{\opon{M_i}{\GG'}} }{ \PA' }$.)
The rule will be most useful if we can chose the $M_i$
in such a way that
$\bigvee\nolimits_{\!i} \pb\paren{ \oppred{ \paren{\opon{M_i}{\GG'}}
  }{ \PA' } }=\PA'$.  We will see later that this is often possible
when classical variables are involved (i.e., when the precondition
contains $\class\XX$).

An example of using \ruleref{Transmute} analyzes the predicate $\distr\XX D$:
\begin{align*}
  \distr\XX D
  &=\paren{\XX\ee\quanteq\psiDD D}
    \impl {\textstyle\bigvee_i \oppred{\oponp{\proj{\ket i}}\uu}{\paren{\XX\uu\quanteq\psiDD D}}}
    \\&= {\textstyle\bigvee_i \pb\paren{\XX\uu\quanteq \sqrt{D(i)}\,\ket i\otimes\ket i}}
  \subseteq
       \begin{cases}
         \paren{\XX\QUANTEQ\uu} = \separable\XX \\
         \paren{\XX\CLASSEQ\uu} = \class\XX.
       \end{cases}
\end{align*}
Here \ruleref{Transmute} is applied with $M_i:=\proj{\ket i}$, $\GG:=\ee$, $\GG':=\uu$.
Thus $\distr\XX D$ implies that $\XX$ is \fullonly{separable and }classical.

The main purpose of \rulerefx{ShapeShift}, in contrast, is to rewrite
the state $\psi$
in predicates of the form $\XX\EE\quanteq\psi$.
The \rulerefx{ShapeShift} rule
has the precondition
$ \partr{\EE}\proj{\psi} = \partr{\EE'}\proj{\psi'} $.
That is, after tracing out (erasing) $\EE,\EE'$,
the two states $\psi,\psi'$
(interpreted as density operators by applying $\proj\cdot$)
should be identical. Or, stated differently, looking only at $\XX$,
$\psi$
and $\psi'$
have to look identical. Thus the rule says, roughly, that in a predicate
$\XX\EE\quanteq\psi$,
we can replace $\psi$
by any state that looks identical from the point of view of $\XX$.
For example, $\xx\ee\quanteq\ket{00}+\ket{11} \ \impl\
\xx\ee\quanteq\ket{01}+\ket{10}$.

Both \rulerefx{Transmute} and \rulerefx{ShapeShift} are be
extensively used in the derivations of derived rules in
\autoref{sec:deriv.lang}.
\fullonly{We refer to those derivations for examples as to how and where
\rulerefx{Transmute} and \rulerefx{ShapeShift} can be used.}
\rulerefx{ShapeShift} is also used as the core step in the security
proof of the quantum one-time pad (\autoref{sec:sec.qotp}).
\fullonly{The rules are proven in \cref{lemma:Transmute,lemma:ShapeShift},
respectively.}

\subsection{Proofs of core rules}
\label{sec:core.proofs}

We begin with some auxiliary lemmas:

\begin{lemma}\label{lemma:suppo.partr}
  Let $\rho$
  be a mixed memory over $\VV\WW$
  and $\PA$
  a predicate over $\VV$.
  Then $\suppo\rho\subseteq{{\PA\otimes\elltwo{\WW}}}$
  iff $\suppo{\partr{\WW}\rho}\subseteq\PA$.
\end{lemma}

\begin{proof}
  For a predicate $\PB$,
  let $P_{\PB}$
  denote the projector onto $\PB$.
  Then for a mixed memory $\sigma=\sum_i\proj{\psi_i}$,
  $\suppo\sigma\subseteq\PB$
  iff $\forall i.\ \psi_i\in\PB$
  iff $\forall i.\ \norm{P_{\PB}\psi_i}=\norm{\psi_i}$
  iff
  $\forall i.\tr P_{\PB}\proj{\psi_i}\adj{P_{\PB}}=\tr\proj{\psi_i}$
  iff $\tr P_{\PB}\sigma\adj{P_{\PB}}=\tr\sigma$. (The last step uses that
  $\tr P_{\PB}\proj{\psi_i}\adj{P_{\PB}}\leq\tr\proj{\psi_i}$ for all $i$.)

  We have
  \begin{multline*}
    \tr P_{\PA\otimes\elltwov\WW}\rho\adj{P_{\PA\otimes\elltwov\WW}}
    =
    \tr \paren{P_{\PA}\otimes\id_{\WW}}\rho\adj{\paren{P_{\PA}\otimes\id_{\WW}}}
    \\=
    \tr \partr{\WW} \paren{P_{\PA}\otimes\id_{\WW}}\rho\adj{\paren{P_{\PA}\otimes\id_{\WW}}}
    =
    \tr  {P_{\PA}}\paren{\partr{\WW}\rho}\adj{{P_{\PA}}}
  \end{multline*}
  Thus $\suppo\rho\subseteq{{\PA\otimes\elltwo{\WW}}}$
  iff $\suppo\partr{\WW}\rho\subseteq\PA$.
\end{proof}

\begin{lemma}\label{lemma:sats.sum}
  If $\PA$ is a predicate,
  and $\rho_i$
  are a family of mixed memories with $\sats{\rho_i}\PA$
  for all $i$, and $\sum_i\rho_i$ exists,
  then $\sats{\sum_i\rho_i}\PA$.
\end{lemma}

\begin{proof}
  $\PA$ is a predicate over $\XX\EE\UU$, and $\rho_i$ are mixed memories over $\XX$ for some $\XX\EE\UU$.
  Since $\sats{\rho_i}\PA$,  there are $(\XX\EE,\UU)$-separable $\rho^\circ_i$
  with $\suppo\rho_i^\circ\subseteq\PA$
  and $\partr{\EE\UU}\rho^\circ_i=\rho_i$.
  We have
  $\sum\tr\rho_i^\circ = \sum\tr\rho_i < \infty$, since $\sum\rho_i$ exists.
  Thus 
  $\hat\rho^\circ:=\sum\rho_i^\circ$
  exists.

  Then
  $\partr{\EE\UU}\hat\rho^\circ=\sum\partr{\EE\UU}\rho_i^\circ =
  \sum\rho_i$. And since all $\rho^\circ_i$
  are $(\XX\EE,\UU)$-separable,
  so is $\hat\rho^\circ$. Finally, $\suppo\hat\rho^\circ=\suppo\sum\rho_i^\circ=\sum\suppo
  \rho_i^\circ\subseteq\PA$.
  Thus $\sats{\sum\rho_i}\PA$, as desired.
\end{proof}

\begin{lemma}\label{lemma:apply.rho.pred}
  Let $\rho$
  be a mixed memory over $\XX$, let $\PA$ be a predicate, and assume $\sats\rho\PA$.
  \begin{compactenum}[(i)]
  \item\label{item:general} Let $M$
    be an operator from $\elltwov{\XX}$
    to $\elltwov{\XX'}$.
    Then $\sats{M\rho\adj M}{\oppred{\paren{\opon M\XX}}\PA}$.
  \item\label{item:disjoint} Let $\YY\subseteq\XX$ and $\fv(\PA)\cap\YY=\varnothing$
    and $N$
    be an operator from from $\elltwov{\YY}$
    to $\elltwov{\YY'}$.
    Then $\sats{\paren{\opon N\YY}\rho\adj{\paren{\opon N\YY}}}\PA$.
  \end{compactenum}
\end{lemma}

\begin{proof}
  We first prove \eqref{item:general}.
  Let $\mathcal E(\sigma):=M\sigma\adj M$
  for all $\sigma$.
  Then $\rho':=\mathcal E(\rho)$
  is a mixed memory over $\XX'$,
  and we need to show $\sats{\rho'}{\oppred{\oponp M\XX}\PA}$.
  And in the judgment $\sats\rho\PA$, $\PA$
  is a predicate over some variables $\XX\EE\UU$ for some variables $\EE\UU$,
  and in the judgment $\sats{\rho'}{\oppred M\PA}$,
  it is interpreted as a predicate over variables $\XX'\EE\UU$
  (see \autopageref{page:pred.identify}).

  Since $\sats\rho\PA$,
  there exists a $(\XX\EE,\UU)$-separable
  $\rho^\circ$
  with $\suppo\rho^\circ\subseteq\PA$
  and $\partr{\EE\UU}\rho^\circ=\rho$.
  Let
  $\tilde\rho^\circ := (\mathcal E\otimes\id_{\EE\UU})(\rho^\circ)$.
  Then $\tilde\rho^\circ$
  is $(\XX'\EE,\UU)$-separable.
  Furthermore,
  \[
    \partr{\EE\UU}\tilde\rho^\circ
    = (\id_{\XX}\otimes\tr)\circ(\mathcal E\otimes\id_{\EE\UU})(\rho^\circ)
    = (\mathcal E\otimes\tr)(\rho^\circ)
    = \mathcal E (\partr{\EE\UU}\rho^\circ)
    = \mathcal E (\rho) = \rho'.
  \]
  (Here $\tr$
  is seen as a superoperator from the trace-class operators over $\EE\UU$
  to the trace-class operators on the $1$-dimensional space $\setC$.)

  And finally, we can write $\rho^\circ=\sum_i\proj{\psi_i}$
  for some $\psi_i\in\elltwov{\XX\EE\UU}$ and thus
  \begin{align*}
    \suppo\tilde\rho^\circ &=
                             \suppo\,(\mathcal E\otimes\id_{\EE\UU})(\rho^\circ) =
                             \suppo\,\paren{\opon M\XX} \rho^\circ  \adj{\paren{\opon M\XX}} \\
                           &=
                             \suppo\sum\nolimits_i\pb\proj{\paren{\opon M\XX}\psi_i}
                             = \SPAN\braces{\paren{\opon M\XX}\psi_i}_i
                             = \oppred{\paren{\opon M\XX}}
                             {\SPAN\braces{\psi_i}_i}
                             \\
                           &= \oppred{\paren{\opon M\XX}}{\suppo\rho^\circ}
                             \subseteq
                             \oppred{\paren{\opon M\XX}}{\PA}
  \end{align*}
  Thus $\sats{\rho'}{\oppred{\paren{\opon M\XX}}{\PA}}$.
  This shows \eqref{item:general}.

  \medskip

  We now show \eqref{item:disjoint} by reduction to
  \eqref{item:general}. Let $M:=\paren{\opon N\YY}$. Then by \eqref{item:general},
  ${\paren{\opon N\YY}\rho\adj{\paren{\opon N\YY}}}
  =
  \sats{M\rho\adj M}{\oppred{\paren{\opon M\XX}}\PA}$. Thus we need to show that
  ${\oppred{\paren{\opon M\XX}}\PA}\subseteq\PA$.

  Fix $\psi\in{\oppred{\paren{\opon M\XX}}\PA}$.
  Then there exists $\psi'\in\PA$
  with $\psi= \oppred{\paren{\opon M\XX}}\psi'$.
  $\PA$ is a predicate on $\XX\EE\UU$ for some $\EE\UU$.
  Let $P$ be the projector onto $\PA$. Since $\fv(\PA)\cap\YY=\varnothing$,
  we can write $P=\id_{\YY}\otimes P'$ for some projector $P'$ on $\XX\EE\UU\setminus\YY$.
  And when $\PA$ is interpreted as a predicate on $\XX\EE\UU\setminus\YY\cup\YY'$
  (as, e.g., in the judgment $\sats{\oponp N\YY\rho\adj{\oponp N\YY}}\PA$),
  then the projector onto $\PA$ is $\id_{\YY'}\otimes P'$.
  We have
  \begin{align*}
    (\id_{\YY'}\otimes P')\psi
    &
            = P (N\otimes\elltwov{\XX\EE\UU\setminus\YY})\psi'
            = \paren{\id_{\YY'}\otimes P'}  (N\otimes\elltwov{\XX\EE\UU\setminus\YY})\psi'
    \\
    &= (N \otimes P') \psi'
      = \oponp M\XX \paren{\id_{\YY}\otimes P'}\psi'
            = {\paren{\opon M\XX}} P\psi'
            \starrel= {\paren{\opon M\XX}} \psi' = \psi.
  \end{align*}
  Here $(*)$
  follows since $\psi'$
  is in $\PA$,
  the image of the projector $P$.
  Thus $\paren{\id_{\YY'}\otimes P'}\psi=\psi$,
  hence $\psi$ is in the image of $\id_{\YY'}\otimes P$, hence $\psi\in \PA$.

  Since this holds for all $\psi\in{\oppred{\paren{\opon M\XX}}\PA}$, this implies ${\oppred{\paren{\opon M\XX}}\PA}\subseteq\PA$.
  With 
  $
  \sats{\paren{\opon N\YY}\rho\adj{\paren{\opon N\YY}}}{\oppred{\paren{\opon M\XX}}\PA}$, we get
  $
  \sats{\paren{\opon N\YY}\rho\adj{\paren{\opon N\YY}}}{\PA}$. This shows \eqref{item:disjoint}.
\end{proof}

\begin{lemma}\label{lemma:Apply}
  \Ruleref{Apply} is sound.
\end{lemma}

\begin{proof}%
  \newcommand\prog{\apply U\XX}%
  \newcommand\Ux{\paren{\opon U\XX}_{\XXall}}%
  \newcommand\Uxe{\paren{\opon U\XX}_{\XXall\EE}}%
  \newcommand\Uxx{\paren{\opon U\XX}_{\XXall\EE\UU}}%
  The predicate $\PA$
  is a space of quantum memories over some variables $\XXall\EE\UU$
  with $\XX\subseteq\XXall$.
  In this proof, we will encounter both the term $\opon U\XX$
  interpreted as an operator on quantum memories over $\XXall$,
  and over $\XXall\EE\UU$.
  Since the syntax $\opon U\XX$
  does not disambiguate between the two (the space we are operating on
  is left implicit), we write $\Ux$ and $\Uxx$, respectively.
  Note that $\Uxx=\Ux\otimes\id_{\EE\UU}$.

  \medskip We need to show that for any mixed memory $\rho$
  over $\XX$,
  $\sats\rho\PA$
  implies $\sats{\denot\prog(\rho)}{\oppred{\Uxx}\PA}$.
  Since $\sats\rho\PA$, there exists an $(\XXall\EE,\UU)$-separable $\rho^\circ$
  with $\suppo\rho^\circ\subseteq\PA$
  and $\partr{\EE\UU}\rho^\circ=\rho$.
  Since $\rho^\circ$
  is $(\XXall\EE,\UU)$-separable,
  we can write it as
  $\rho^\circ=\sum_i\proj{\psi_{\XXall\EE,i}\otimes\psi_{\UU,i}}$.
  Since $\suppo\rho^\circ\subseteq\PA$, this implies that
  $\psi_{\XXall\EE,i}\otimes\psi_{\UU,i}\in\PA$ for all $i$.

  Let
  \[
    \hat\rho^\circ := \sum_i\pB\proj{
      \pb\paren{\Uxe\ \psi_{\XXall\EE,i}} \otimes \psi_{\UU,i}
    }.
  \]
  Then $\hat\rho^\circ$ is $(\XXall\EE,\UU)$-separable.
  And
  \[
    \pb\paren{\Uxe\ \psi_{\XXall\EE,i}} \otimes \psi_{\UU,i}
    =
    \Uxx\paren{\psi_{\XXall\EE,i} \otimes \psi_{\UU,i}}
    \in \oppred\Uxx\PA.
  \]
  So $\suppo\hat\rho^\circ\in\oppred\Uxx\PA$. Finally,
  \begin{align*}
    \partr{\EE\UU}\hat\rho^\circ
    &=
      \Ux
      \pB\paren{
      \partr{\EE\UU}
      \sum\nolimits_i \proj{\psi_{\XXall\EE,i}\otimes\psi_{\UU,i}}
      }
      \adj\Ux
    \\
    &=
      \Ux
      \paren{
      \partr{\EE\UU}
      \rho^\circ
      }
      \adj\Ux 
    =
      \Ux
      \rho
      \adj\Ux
    = \denot\prog(\rho).
  \end{align*}
  Hence $\sats{\denot\prog(\rho)}{\oppred\Uxx\PA}$.
\end{proof}

\begin{lemma}\label{lemma:Init}
  \Ruleref{Init} is sound.
\end{lemma}

\begin{proof}
  \newcommand\Asub{\psubst\PA\ee\xx_{\XXall\EE\ee\UU}} $\PA$
  is a predicate on some variables $\XXall\EE\UU$.
  We also have that $\ee\notin\EE$
  and that $\ee$
  and $\xx$
  have the same type because otherwise $\psubst\PA\ee\xx$
  would not be well-typed.  Let
  $\PB:=\paren{ \psubst\PA\ee\xx,\ {\xx\quanteq\ket0}}$ (the postcondition).
  Note that there is an implicit conversion happening:
  By definition,
  $\psubst\PA\ee\xx$
  is a predicate on $\XXall\EE\ee\UU\setminus\xx$,
  so it does not make sense to intersect it ($\wedge$)
  with $\xx\quanteq\ket0$.
  But, as discussed at the end of \autoref{sec:var.mem.pred}
  (\autopageref{page:pred.identify}), we identify any predicate $\PA'$
  with $\PA'\otimes\elltwov\xx$.
  In particular, $\psubst\PA\ee\xx$
  is identified with $\psubst\PA\ee\xx\otimes\elltwov\xx=:\Asub$
  on $\XXall\EE\ee\UU$.
  In that notation, $\PB$
  is actually $\PB= \paren{\Asub,\ {\xx\quanteq\ket0}}$,
  a predicate over $\XXall\EE\ee\UU$.

  \medskip

  To show the rule, we need to show that for any mixed memory
  $\rho$
  on $\XXall$,
  $\sats\rho\PA$
  implies $\sats\rho\PB$.
  $\sats\rho\PA$
  implies that there is a $(\XXall\EE,\UU)$-separable
  $\rho^\circ$
  such that $\suppo\rho^\circ\subseteq\PA$
  and $\partr{\EE\UU}\rho^\circ=\rho$.

  Let $\mathcal E$ be the canonical mapping from mixed memories over $\xx$
  to mixed memories over $\ee$. (Formally,
    $\mathcal E(\sigma)=\Urename{\xx}\ee\adj{\sigma\Urename\xx\ee}$
    where $\Urename\xx\ee$
    was defined on \autopageref{page:def:Urename}.)  Then
  $\mathcal E\otimes\id_{\XXall\EE\UU\setminus\xx}$
  maps mixed memories over $\XXall\EE\UU$
  to mixed memories over $\XXall\EE\ee\UU\setminus\xx$
  by renaming $\xx$ to $\ee$. We define:
  \[
    \hat\rho^\circ := \paren{\mathcal E\otimes\id_{\XXall\EE\UU\setminus\xx}}(\rho^\circ)
    \qquad\text{and}\qquad
    \tilde\rho^\circ := \hat\rho^\circ \otimes \proj{\ket0_{\xx}}.
  \]
  Since $\rho^\circ$
  is $(\XXall\EE,\UU)$-separable,
  and ${\mathcal E\otimes\id_{\XXall\EE\UU\setminus\xx}}$
  is the identity on $\UU$, we have that $\hat\rho^\circ$ is 
  $(\XXall\EE\ee\setminus\xx,\UU)$-separable, and thus $\tilde\rho^\circ$ is 
  $(\XXall\EE\ee,\UU)$-separable.

  We have
  \begin{align*}
    \suppo\hat\rho^\circ &=
    \suppo \paren{\Urename\xx\ee\otimes\id_{\XXall\EE\UU\setminus\xx}}\rho^\circ
    \adj{\paren{\Urename\xx\ee\otimes\id_{\XXall\EE\UU\setminus\xx}}}
    \\ &
         =
         \oppred{\paren{\Urename\xx\ee\otimes\id_{\XXall\EE\UU\setminus\xx}}}{\suppo\rho^\circ}
         \subseteq
         \oppred{\paren{\Urename\xx\ee\otimes\id_{\XXall\EE\UU\setminus\xx}}}{\PA}
         \starrel=
         \psubst\PA\xx\ee.
  \end{align*}
  Here $(*)$ is the definition of $\psubst\PA\xx\ee$ (\autopageref{page:def:psubst}).
  Thus
  \[
    \suppo\tilde\rho^\circ\subseteq\suppo\hat\rho^\circ \otimes \elltwov{\xx}
    \subseteq \psubst\PA\xx\ee \otimes \elltwov{\xx} = \Asub.
  \]
  And $\suppo\tilde\rho^\circ\subseteq(\xx\quanteq\ket0)$ by definition of $\tilde\rho^\circ$.
  Thus $\suppo\tilde\rho^\circ\subseteq \PB$.
  Finally,
  \begin{align*}
    \partr{\EE\ee\UU}\tilde\rho^\circ
    &=
    \paren{\partr{\EE\ee\UU}\hat\rho^\circ}
    \otimes
      \proj{\ket0_{\xx}}
      \starrel=
      \paren{\partr{\EE\xx\UU}\rho^\circ}
      \otimes
      \proj{\ket0_{\xx}}
    \\&
      =
      \paren{\partr\xx{\partr{\EE\UU}\rho^\circ}}
      \otimes
      \proj{\ket0_{\xx}}
      =
      \paren{\partr\xx\rho}
      \otimes
    \proj{\ket0_{\xx}}
    =
    \denot{\init\xx}(\rho).
  \end{align*}
  Here $(*)$
  uses that $\hat\rho^\circ$
  is the result of renaming $\xx$
  to $\ee$
  in $\rho^\circ$,
  so tracing out $\ee$
  in $\hat\rho^\circ$
  is the same as tracing out $\xx$ in $\rho^\circ$.

  Altogether, we have $\sats{\denot{\init\xx}(\rho)}\PB$.
\end{proof}

\begin{lemma}\label{lemma:If}
  \Ruleref{If} is sound.
\end{lemma}

\begin{proof}
  The predicate $\PA$
  is a space of quantum memories over some variables $\XXall\EE\UU$.
  The predicate $\PB$ is a space of quantum memories over some variables
  $\XXall\EE'\UU'$.

  \medskip
  
  We need to show that for any mixed memory $\rho$
  on $\XXall$,
  if $\sats\rho\PA$,
  then $\sats{\denot{\ifte\xx\bc\bd}(\rho)}\PB$.
  Recall that
  $\denot{\ifte\xx\bc\bd}(\rho)=
  \denot\bc(\rho_1)
  +
  \denot\bd(\rho_0)$ where
  $\rho_i:={\restrict i(\rho)}=
  \paren{\opon{\proj{\ket i}}\xx}\rho\adj{\paren{\opon{\proj{\ket i}}\xx}}$.

  By \autoref{lemma:apply.rho.pred}, we have
  $\rho_i=\sats{\paren{\opon{\proj{\ket i}}\xx}\rho\adj{\paren{\opon{\proj{\ket i}}\xx}}}
  {\oppred{\oponp{\proj{\ket i}}\xx}\PA}$.  
  Since $\pb\hl{\oppred{\paren{\opon{\proj{\ket1}}\xx}}\PA}\bc\PB$
  and $\pb\hl{\oppred{\paren{\opon{\proj{\ket0}}\xx}}\PA}\bd\PB$
  by assumption of the \rulerefx{If} rule, it follows that
  $\sats{\denot\bc(\rho_1)}{\PB}$
  and $\sats{\denot\bd(\rho_0)}{\PB}$.
  Thus with \autoref{lemma:sats.sum},
  $\denot{\ifte\xx\bc\bd}(\rho)
  =
  \sats{\denot\bc(\rho_1) +
    \denot\bd(\rho_0)}\PB
  $.
  \end{proof}

\begin{lemma}\label{lemma:While}
  \Ruleref{While} is sound.
\end{lemma}

\begin{proof}
  The predicate $\PA$
  is a subspace of quantum memories over $\XXall\EE\UU$ for some $\EE\UU$.

  \medskip
  
  We need to show that for any mixed memory $\rho$
  on $\XXall$,
  if $\sats\rho\PA$,
  then $\sats{\denot{\while\xx\bc}(\rho)}{\oppred{\paren{\opon{\proj{\ket
            0}}\xx}}\PA}$.
  Recall that
  $\denot{\while\yy\bc}(\rho) = \sum_{n=0}^\infty \restrict0 (\rho_n)$
  with $\rho_n:=\denot\bc\paren{\restrict 1(\rho_{n-1})}$
  and $\rho_0:=\rho$,
  where $\restrict 1,\restrict 0$
  are defined by
  $\restrict i(\sigma):= \paren{\opon{\proj{\ket
        i}}\xx}\sigma\adj{\paren{\opon{\proj{\ket i}}\xx}}$.

  We show $\sats{\rho_n}\PA$
  for $n\geq0$ by induction. The base case follows since $\rho_0=\sats{\rho}\PA$.
  For the induction step, assume that $\sats{\rho_n}\PA$.
  Then $\restrict 1(\rho_n)=
  \sats{\paren{\opon{\proj{\ket1}}\xx}\rho_n\adj{\paren{\opon{\proj{\ket1}}\xx}}}
  {\oppred{\oponp{\proj{\ket1}}\xx}\PA}$
  by \autoref{lemma:apply.rho.pred}\,\eqref{item:general}.
  From the premise of the \rulerefx{While} rule, it then follows that
  $\rho_{n+1}=\sats{\denot\bc\pb\paren{\restrict1(\rho_n)}}\PA$.

  Since $\sats{\rho_n}\PA$,
  we have $\sats{\restrict 0(\rho_n)}{\oppred{\paren{\opon{\proj{\ket
            0}}\xx}}\PA}$
  (again by \autoref{lemma:apply.rho.pred}\,\eqref{item:general}).
  Since this holds for all $n$,
  with \autoref{lemma:sats.sum} we have that 
  $\denot{\while\yy\bc}(\rho) =
  \sats { \sum_{n=0}^\infty \restrict0 (\rho_n) }
  {\oppred{\paren{\opon{\proj{\ket 0}}\xx}}\PA}$.
\end{proof}

\begin{lemma}\label{lemma:Case}
  \Ruleref{Case} is sound.
\end{lemma}

\begin{proof}
  The predicate $\PA$
  is a subspace of quantum memories over 
  $\XXall\EE\Tilde\UU$ for some variables $\EE\Tilde\UU$
  with $\UU\subseteq\Tilde\UU$.
  And $\PB$
  is a subspace of quantum memories over 
  $\XXall\EE'\UU'$ for some variables $\EE'\UU'$.

  \medskip
  
  We need to show that if $\sats\rho\PA$
  then $\sats{\denot\bc(\rho)}\PB$.
  Since $\sats\rho\PA$,
  there is a $(\XXall\EE,\Tilde\UU)$-separable
  $\rho^\circ$
  with $\suppo\rho^\circ\subseteq\PA$
  and $\partr{\EE\Tilde\UU}\rho^\circ=\rho$.
  For making the notation more compact in the remainder of the proof,
  let $\VV:=\XXall\EE\setminus\XX$
  and $\WW:=\Tilde\UU\setminus\UU$.
  In that notation, $\rho^\circ$
  is $(\XX\VV,\UU\WW)$-separable.
  Thus $\rho^\circ$
  can be written as
  $\rho^\circ=\sum_i\proj{\psi_{\mathit{all},i}}$
  where $\psi_{\mathit{all},i}:=\psi_{\XX\VV,i}\otimes\psi_{\UU\WW,i}$
  for some quantum memories $\psi_{\XX\VV,i}\neq0$
  over $\XX\VV$ and $\psi_{\UU\WW,i}\neq0$ over $\UU\WW$.

  Fix some $i$.
  (We will omit $i$
  from the subscripts for now.) We can write $\psi_{\UU\WW}\neq0$
  as a nonempty sum
  $\psi_{\UU\WW} = \sum_j\phi_{\UU,j}\otimes\phi_{\WW,j}$
  with quantum memories $\phi_{\UU,j},\phi_{\WW,j}\neq0$
  over $\UU$
  and $\WW$
  respectively where the $\phi_{\WW,j}$ are orthogonal.
  We have $\psi_\mathit{all}\in\suppo\rho^\circ \subseteq\PA\subseteq\PC$.
  Thus
  \begin{equation*}
    \psi_{\XX\VV} \otimes \phi_{\UU,j} \otimes \phi_{\WW,j}
    =
    \oponp{\proj{\phi_{\WW,j}}}\WW\,\psi_{\mathit{all}}
    \in 
    \oppred{\oponp{\proj{\phi_{\WW,j}}}\WW}\PC
    \starrel\subseteq \PC.
  \end{equation*}
  Here $(*)$ follows from $\fv(\PC)\cap\WW=\varnothing$.
  Hence $\psi_{\XX\VV}\otimes\phi_{\UU,j}\in \PC$ for all $j$.
  (Note here that the notation $\PC$
  is overloaded both as a predicate over $\XX\VV\UU\WW$
  and over $\XX\VV\UU$, see \autopageref{page:pred.identify}.)
  Thus $\psi_{\XX\VV}\otimes\phi_{\UU,j^*}\in\PC$ for some $j^*$.
  (Recall that the sum $\psi_{\UU\WW} = \sum_j\phi_{\UU,j}\otimes\phi_{\WW,j}$ was nonempty.)

  Since $\PC$
  is $M$-disentangling,
  this implies that $\psi_{\XX\VV}=\psi_{\XX}\otimes\psi_{\VV}$
  for some $\psi_{\VV}\in\elltwov\VV$
  and $\psi_{\XX}\in M$.
  Thus
  $\psi_\mathit{all}=\psi_{\XX}\otimes\psi_{\VV}\otimes\psi_{\UU\WW}\in(\XX\quanteq\psi_{\XX})$.
  Since also $\psi_\mathit{all}\in\PA$,
  we have that
  $\suppo\proj{\psi_\mathit{all}}\subseteq\paren{\XX\quanteq\psi_{\XX},\ \PA}$.
  And since $\XX\VV=\XXall\EE$
  and $\UU\WW=\Tilde \UU$,
  $\proj{\psi_\mathit{all}}$ is $(\XXall\EE,\Tilde\UU)$-separable. Hence
  $\rho_\mathit{all}:=\sats{\partr{\EE\Tilde\UU} \proj{\psi_\mathit{all}}}{(\XX\quanteq\psi_{\XX})\land\PA}$.
  By assumption of the rule \rulerefx{Case}, 
  $\hl{\XX\quanteq\psi_{\XX},\ \PA}\bc\PB$ (since $\psi_{\XX}\in M$). Thus $\sats{\denot\bc(\rho_\mathit{all})}\PB$.
  
  We now ``unfix'' $i$. We thus have $\sats{\denot\bc(\rho_{\mathit{all},i})}\PB$ for
  $\rho_{\mathit{all},i}:=\partr{\EE\Tilde\UU} \proj{\psi_{\mathit{all},i}}$. Furthermore,
  \[
    \denot\bc(\rho)=\denot\bc\paren{\partr{\EE\TIlde\UU}\rho^\circ}
    =
    \denot\bc\pB\paren{\partr{\EE\TIlde\UU}\sum\nolimits_i\proj{\psi_{\mathit{all},i}}}
    =
    \denot\bc\pB\paren{\sum\nolimits_i\rho_{\mathit{all},i}}
    =
    \sum\nolimits_i \denot\bc\paren{\rho_{\mathit{all},i}}.
  \]
  By \autoref{lemma:sats.sum},  $\forall i.\, \sats{\denot\bc(\rho_{\mathit{all},i})}\PB$ implies
  $\denot\bc(\rho)=\sats{\sum\nolimits_i\denot\bc\paren{\rho_{\mathit{all},i}}}\PB$.
\end{proof}

The following is an auxiliary lemma needed for the proof of
\ruleref{Universe}. But it is also of independent interest because it
says that the choice of the set of program variables with respect to
which we evaluate a program (denoted $\XXall$
on \autopageref{page:XXall}) does not matter as long as it is large enough.

\begin{lemma}[Changing the set of program variables]\label{lemma:change.prog}
  Let $\hl\PA\bc\PB^{\XX}$
  denote Hoare judgments $\hl\PA\bc\PB$ as in \autoref{def:hoare}, except that the
  set $\XX$ is used instead of $\XXall$.
  (I.e., the semantics of the program $\bc$
  are defined with respect to memories containing variables $\XX$,
  not $\XXall$.)

  Assume that $\fv(\bc),\progvars{\fv(\PA)},\progvars{\fv(\PB)}\subseteq \XX_1,\XX_2$.
  Let $\YY_i:=\XX_i\setminus\fv(\bc)\setminus\fv(\PA)\setminus\fv(\PB)$.
  Let $T_i$
  be the type of $\YY_i$.
  Assume that $\abs{T_1}\geq\abs{T_2}$
  or ${T_1}$
  is infinite. Then
  $\hl\PA\bc\PB^{\XX_1}\implies\hl\PA\bc\PB^{\XX_2}$.
\end{lemma}

\begin{proof}
  We fix some $\EE,\UU$
  such that $\EE\UU$
  contains all ghosts from $\fv(\PA),\fv(\PB)$.
  By \autoref{lemma:ghost.sets}, we can interpret $\PA,\PB$
  in judgments $\sats\rho\PA,\sats\rho\PB$
  as predicates over $\XX\EE\UU$
  (where $\rho$
  is over $\XX$),
  i.e., we can without loss of generality use the same $\EE\UU$
  everywhere.
  Note that $\XX_1\setminus\YY_1=\XX_2\setminus\YY_2$ since
  both are equal to $\fv(\bc)\cup\progvars{\fv(\PA)}\cup\progvars{\fv(\PB)}$.

  \medskip
  
  Assume $\hl\PA\bc\PB^{\XX_1}$.
  To show $\hl\PA\bc\PB^{\XX_2}$,
  we fix a mixed memory $\rho$
  over $\XX_2$
  with $\sats\rho\PA$,
  and we need to show $\sats{\denot\bc(\rho)}\PB$.

  Let
  $S:=\suppo\partr{\XX_2\setminus\YY_2}\rho\subseteq\elltwov{\YY_2}$.
  The operator $\partr{\XX_2\setminus\YY_2}\rho$
  can be written as $\sum_{i\in I}p_i\proj{\psi_i}$
  with some orthonormal $\psi_i\in\elltwov{\YY_2}$
  and some $p_i > 0$
  with
  $\sum_{i\in I}
  p_i=\tr\paren{\partr{\XX_2\setminus\YY_2}\rho}<\infty$. Thus $I$
  is countable (otherwise $\sum_{i\in I} p_i$
  cannot converge). Hence $\dim S=\dim\SPAN\{\psi_i\}_{i\in I}=\abs I$
  is countable.  Furthermore
  $\dim S\leq \dim\elltwov{\YY_2}=\abs{T_2}$.
  Since $\abs{T_1}\geq\abs{T_2}$
  or $\abs{T_1}$
  is infinite, it follows that
  $\dim S\leq\abs{T_1}=\dim\elltwov{\YY_1}$.
  Thus there exists an isometry $U$
  from $S$
  to $\elltwov{\YY_1}$.
  We extend $U$
  to an operator from $\elltwov{\YY_2}$
  to $\elltwov{\YY_1}$
  by setting $U=0$
  on the orthogonal complement of $S$.
  Then $\adj UU$
  is the projector $P_S$
  onto $S$.
  Since $S=\suppo\partr{\XX_2\setminus\YY_2}\rho$,
  $S \otimes \elltwov{\XX_2\setminus\YY_2}\supseteq\suppo\rho$.
  And ${\oponp{P_S}{\YY_2}}$
  is the projector onto $S \otimes
  \elltwov{\XX_2\setminus\YY_2}$. Hence $\suppo\rho$
  is fixed by ${\oponp{P_S}{\YY_2}}$.

  Let $\mathcal E(\sigma):=U\sigma \adj U$
  for all $\sigma$
  over $\YY_2$.
  And let
  $\Hat{\mathcal E}(\sigma):=\paren{\opon
    U{\YY_2}}\sigma\adj{\paren{\opon U{\YY_2}}}$ for all $\sigma$
  over $\XX_2$,
  i.e.,
  $\Hat{\mathcal E}=\id_{\XX_2\setminus\YY_2}\otimes{\mathcal E}$.
  Let $\mathcal E^*(\sigma):=\adj U\sigma U$
  for all $\sigma$
  over $\YY_1$.
  And let
  $\Hat{\mathcal E}^*(\sigma):=\paren{\opon{\adj U}\YY_1}\sigma
  \adj{\paren{\opon{\adj U}\YY_1}}$ for all $\sigma$
  over $\XX_1$,
  i.e.,
  $\Hat{\mathcal E}^*=\id_{\XX_1\setminus\YY_1}\otimes\mathcal E^*$.
  We have
  $\Hat{\mathcal E}^*\circ\Hat{\mathcal E}(\rho)=\paren{\opon{\adj
      UU}{\YY_2}} \rho \adj{\paren{\opon{\adj
        UU}{\YY_2}}}=\paren{\opon{P_S}{\YY_2}}\rho
  \adj{\paren{\opon{P_S}{\YY_2}}}=\rho$.
  Here the last equality is
  because $\suppo\rho$ is fixed by $\oponp{P_S}{\YY_2}$.
  
  Since $\fv(\PA) \cap \YY_2 = \varnothing$ and $\sats\rho\PA$,
  by \autoref{lemma:apply.rho.pred}\,\eqref{item:disjoint},
  $\sats{\Hat{\mathcal E}(\rho)}\PA$.
  Note that ${\Hat{\mathcal E}(\rho)}$
  is a mixed memory over $\XX_2\setminus\YY_2\dotcup\YY_1=\XX_1$.
  Since $\hl\PA\bc\PB^{\XX_1}$,
  we have
  $\sats{\denot\bc^{\XX_1}\pb\paren{\Hat{\mathcal E}(\rho)}}\PB$.
  Here $\denot\bc^{\XX_1}$
  denotes the semantics $\denot\bc$ of $\bc$
  defined with respect to the set of variables $\XX_1$
  instead of $\XXall$. Note that
   $\denot\bc^{\XX_1}=\denot\bc^{\XX_1\setminus\YY_1}\otimes\id_{\YY_1}$
  and
   $\denot\bc^{\XX_2}=\denot\bc^{\XX_2\setminus\YY_2}\otimes\id_{\YY_2}$ (since $\fv(\bc)\cap\YY_1=\fv(\bc)\cap\YY_2=\varnothing$). 
Then
  \begin{align}
    \Hat{\mathcal E}^*
    \pB\paren {\denot\bc^{\XX_1}\pb\paren{\Hat{\mathcal E}(\rho)}}
    &=
    \paren{\id_{\XX_1\setminus\YY_1}\otimes\mathcal E^*}
    \circ
    \paren{ \denot\bc^{\XX_1\setminus\YY_1}\otimes\id_{\YY_1} }
    \circ 
    \paren{\id_{\XX_2\setminus\YY_2}\otimes\mathcal E}
      (\rho)
    \notag\\
    &\starrel=
      \pB
      \paren{ \denot\bc^{\XX_2\setminus\YY_2}
      \otimes
      \paren{\mathcal E^*\circ\mathcal E} }
      (\rho)
      = \pb\paren{ \denot\bc^{\XX_2\setminus\YY_2}
      \otimes \id_{\YY_2} }
      \circ
      \pb\paren{ \id_{\XX_2\setminus\YY_2} \otimes  \paren{\mathcal E^*\circ\mathcal E} }
      (\rho)
    \notag\\
    &
      = \pb\paren{ \denot\bc^{\XX_2\setminus\YY_2}
      \otimes \id_{\YY_2} }
      \circ
      \paren{\Hat{\mathcal E}^*\circ\Hat{\mathcal E}}
      (\rho)
      \starstarrel=
      \pB
      \paren{ \denot\bc^{\XX_2\setminus\YY_2}
      \otimes
      \id_{\YY_2} }
      (\rho)
      =  \denot\bc^{\XX_2}(\rho).
      \label{eq:E*X1}
  \end{align}
  In $(*)$ we use that $\XX_1\setminus\YY_1=\XX_2\setminus\YY_2$.
  And $(**)$ follows since $\Hat{\mathcal E}^*\circ\Hat{\mathcal E}(\rho) = \rho$ (shown above).

  Since $\fv(\PB) \cap \YY_1 = \varnothing$
  and $\sats{\denot\bc^{\XX_1}\pb\paren{\Hat{\mathcal E}(\rho)}}\PB$,
  by \autoref{lemma:apply.rho.pred}\,\eqref{item:disjoint},
  $\sats{\Hat{\mathcal E}^*\pb\paren{{\denot\bc^{\XX_1}\paren{\Hat{\mathcal
            E}(\rho)}}}}\PB$.  By \eqref{eq:E*X1}, this implies
  $\sats{\denot\bc^{\XX_2}(\rho)}\PB$.
  Since $\rho$
  was an arbitrary mixed memory over $\XX_2$
  with $\sats\rho\PA$, this implies $\hl\PA\bc\PB^{\XX_2}$.
\end{proof}

\begin{lemma}\label{lemma:Universe}
  \Ruleref{Universe} is sound.
\end{lemma}

Most of the work for the proof has already been done in
\autoref{lemma:change.prog}.

\begin{proof}
  By \autoref{lemma:change.prog} (with $\XX_1:=\XX\xx$ and $\XX_2:=\XXall$), we have that  $\hl\PA\bc\PB^{\XX\xx}$
  (in the notation of \autoref{lemma:change.prog}) implies $\hl\PA\bc\PB$.
  (For this, note that $\XX\xx\setminus \fv(\bc)\setminus\fv(\PA)\setminus\fv(\PB)$
  contains $\xx$
  and thus has infinite type.)
  Thus to prove the soundness of \rulerefx{Universe}, it is sufficient to prove
  $\hl\PA\bc\PB^{\XX\xx}$.

  To show $\hl\PA\bc\PB^{\XX\xx}$,
  fix some $\rho$
  over $\XX\xx$
  with $\sats\rho\PA$.
  Then there is a $(\XX\xx\EE,\UU)$-separable
  $\rho^\circ$ over $\XX\xx\EE\UU$
  with $\partr{\EE\UU}\rho^\circ=\rho$
  and $\suppo\rho^\circ\subseteq\PA$.

  Since $\rho^\circ$
  is $(\XX\xx\EE,\UU)$-separable,
  we can write $\rho^\circ=\sum_i\proj{\psi_i\otimes\psi_i'}$
  with $\psi_i\in\elltwov{\XX\xx\EE}$
  and $\psi_i'\in\elltwov{\UU}$ and $\psi_i,\psi_i'\neq0$.
  Since $\suppo\rho^\circ\subseteq\PA$,
  we have $\psi_i\otimes\psi_i'\in\PA$
  for all $i$.
  Furthermore, we have
  $\psi_i\otimes\psi_i'\in\paren{\XX\xx\EE\quanteq\psi_i}$
  and $\psi_i\otimes\psi_i'\in\paren{\UU\quanteq\psi_i'}$.
  Thus
  $\psi_i\otimes\psi_i'\in \PA_i := \paren{\XX\xx\EE\quanteq\psi_i,\
    \UU\quanteq\psi_i',\ \PA}$.  We then have that
  $\proj{\psi_i\otimes\psi_i'}$
  is $(\XX\xx\EE,\UU)$-separable
  and $\suppo \proj{\psi_i\otimes\psi_i'}\subseteq \PA_i$.
  Thus $\rho_i:=\sats{\partr{\EE\UU}\proj{\psi_i\otimes\psi_i'}}{\PA_i}$.
  By assumption of the \ruleref{Universe} we have
  $\hl{\PA_i}\bc\PB$.
  By \autoref{lemma:change.prog} (with $\XX_1:=\XXall$ and $\XX_2:=\XX\xx$),
  $\hl{\PA_i}\bc\PB^{\XX\xx}$.  (For this, note that $\XX\xx\subseteq\XXall$ since otherwise
  $\hl{\XX\xx\EE\quanteq\psi,\ \UU\quanteq\psi',\ \PA}\bc\PB$ in the rule's premise would not be well-defined. Hence
  the cardinality of the type of $\XXall\setminus\fv(\bc,\PA,\PB)$ is at least as big as that of $\XX\xx\setminus\fv(\bc,\PA,\PB)$.)
  Since $\sats{\rho_i}{\PA_i}$, we have $\sats{\denot\bc(\rho_i)}\PB$. This holds for all $i$.

  We have
  $\rho=\partr{\EE\UU}\rho^\circ=\sum_i\partr{\EE\UU}\proj{\psi_i\otimes\psi_i'}=\sum_i\rho_i$.
  Thus $\denot\bc(\rho)=\sum_i\denot\bc(\rho_i)$.
  Since for all $i$,
  $\sats{\denot\bc(\rho_i)}\PB$,
  by \autoref{lemma:sats.sum} we have
  $\denot\bc(\rho)=\sats{\sum_i\denot\bc(\rho_i)}\PB$.
  Since this holds for all $\rho$ over $\XX\xx$
  with $\sats\rho\PA$,
  we have shown $\hl\PA\bc\PB^{\XX\xx}$.
  As mentioned in the beginning of the proof, by
  \autoref{lemma:change.prog} this implies $\hl\PA\bc\PB$.
\end{proof}

\begin{lemma}\label{lemma:Transmute}
  \Ruleref{Transmute} is sound.
\end{lemma}

\begin{proof}
  We distinguish four cases: $\GG,\GG'$
  are entangled ghosts (EE-case), $\GG,\GG'$
  are unentangled ghosts (UU-case), $\GG$
  are entangled and $\GG'$
  are unentangled ghosts (EU-case), and $\GG$
  are unentangled and $\GG'$
  are entangled ghosts (UE-case). Most of the proof is the same in all
  four cases, we make case distinctions in individual proof steps as
  necessary.

  $\PA$
  is a predicate over $\XXall\EE\UU$
  for some $\EE,\UU$.
  And $\PA':=\psubst\PA\GG{\GG'}$
  is a predicate over $\XXall\EE'\UU'$
  for some $\EE',\UU'$.
  (What variables are in $\EE'$
  and $\UU'$, respectively, depends on the case. E.g., in the EU-case,
  $\EE'=\EE\setminus\GG$ and $\UU'=\UU\dotcup\GG$.)

  To show
  $ \PA \impl \bigvee\nolimits_{\!i} \pb\paren{ \oppred{
      \paren{\opon{M_i}{\GG'}} }{ \psubst\PA{\GG'}{\GG} }}=:\PB $, we fix
  some mixed memory $\rho$
  over $\XXall$ with $\sats\rho\PA$. We have to show that $\sats\rho\PB$.
    
  Since $\sats\rho\PA$,
  there exists a $(\XXall\EE,\UU)$-separable
  $\rho^\circ$
  over $\XXall\EE\UU$
  with $\suppo\rho^\circ\subseteq\PA$
  and $\partr{\EE\UU}\rho^\circ=\rho$.

  Let
  $\tilde\rho^\circ:=\Urename{\GG}{\GG'}\rho^\circ\adj{\Urename{\GG}{\GG'}}$.
  We have that
  $\suppo\tilde\rho^\circ =
  \oppred{\Urename\GG{\GG'}}{\suppo\rho^\circ}\subseteq
  \oppred{\Urename\GG{\GG'}}\PA = \psubst\PA{\GG'}{\GG}
  = \PA'$.

  In the EE- and EU-case, $\GG\subseteq\EE$.
  That is, $\Urename\GG{\GG'}$
  renames only variables in $\EE$.
  Then since $\rho^\circ$ is $(\XXall\EE,\UU)$-separable,
  $\tilde\rho^\circ$ is $(\XXall\EE\setminus\GG\dotcup\GG',\UU)$-separable
  in that case.  Similarly, in the UE- and UU-case, $\tilde\rho^\circ$
  is then $(\XXall\EE,\UU\setminus\GG\dotcup\GG')$-separable.

  And since $\GG\cap\XXall=\varnothing$, we have $\partr{\EE'\UU'}\tilde\rho^\circ
  = \partr{\EE\UU}\rho^\circ = \rho$.

  Note that this does \emph{not} imply that $\sats\rho\PA'$
  because $\tilde\rho^\circ$
  is not necessarily $(\XXall\EE',\UU')$-separable.
  (The separability is along the wrong variables, at least in the EU-
  and UE-cases.)

  Let
  $\rho^\circ_i:=\oponp{M_i}{\GG'}\tilde\rho^\circ\adj{\oponp{M_i}{\GG'}}$ and $\rho_i:=\partr{\EE'\UU'}\rho^\circ_i$.
  Then
  $\suppo\rho^\circ_i =
  \oppred{\oponp{M_i}{\GG'}}{\suppo\tilde\rho^\circ} \subseteq
  \oppred{\oponp{M_i}{\GG'}}{\PA'} \subseteq \PB$.

  We show that $\rho^\circ_i$
  is $(\XXall\EE',\UU')$-separable
  by distinguishing four cases: \emph{In the EE-case,}
  $\tilde\rho^\circ$
  is $(\XXall\EE\setminus\GG\dotcup\GG',\UU)$-separable.
  And $\oponp{M_i}{\GG'}$
  operates only on $\XXall\EE\setminus\GG\dotcup\GG'$,
  hence $\rho^\circ_i$
  is also $(\XXall\EE\setminus\GG\dotcup\GG',\UU)$-separable.
  And $\EE'=\EE\setminus\GG\dotcup\GG'$,
  $\UU'=\UU$,
  hence $\rho^\circ_i$
  is $(\XXall\EE',\UU')$-separable.
  \emph{In the UU-case,} $\tilde\rho^\circ$
  is $(\XXall\EE,\UU\setminus\GG\dotcup\GG')$-separable.
  And $\oponp{M_i}{\GG'}$
  operates only on $\UU\setminus\GG\dotcup\GG'$,
  hence $\rho^\circ_i$
  is also $(\XXall\EE,\UU\setminus\GG\dotcup\GG')$-separable.
  And $\EE'=\EE$,
  $\UU'=\UU\setminus\GG\dotcup\GG'$,
  hence $\rho^\circ_i$ is $(\XXall\EE',\UU')$-separable.

  \emph{In the EU-case,} $\tilde\rho^\circ$
  is $(\XXall\EE\setminus\GG\dotcup\GG',\UU)$-separable.
  Thus we can write $\tilde\rho^\circ$
  as
  $\tilde\rho^\circ=\sum_j\proj{\psi_{\XXall\EE\setminus\GG\dotcup\GG'}
    \otimes\psi_{\UU}}$ for some
  $\psi_{\XXall\EE\setminus\GG\dotcup\GG'}\in\elltwov{\XXall\EE\setminus\GG\dotcup\GG'}$
  and $\psi_{\UU}\in\elltwov{\UU}$.
  Then
  $\rho^\circ_i =
  \sum_j\proj{\psi'_{\XXall\EE\setminus\GG\dotcup\GG',j}\otimes\psi_{\UU,j}}$
  with
  $\psi'_{\XXall\EE\setminus\GG\dotcup\GG',j}:=\oponp{M_i}{\GG'}\psi_{\XXall\EE\setminus\GG\dotcup\GG',j}$.
  Since $M_i$
  has rank $\leq1$,
  $\im M_i=\SPAN\{\psi'_{\GG'}\}$
  for some $\psi'_{\GG'}\in\elltwov{\GG'}$
  (in the case $\rank M_i=0$,
  we simply use $\psi'_{\GG'}=0$).
  Thus
  $\psi'_{\XXall\EE\setminus\GG\dotcup\GG',j}=\psi'_{\XXall\EE\setminus\GG,j}\otimes\psi'_{\GG'}$
  for some
  $\psi'_{\XXall\EE\setminus\GG,j}\in\elltwov{\XXall\EE\setminus\GG}$.
  Hence
  $\rho^\circ_i = \sum_j\pb\proj{\psi'_{\XXall\EE\setminus\GG,j}\otimes
    (\psi'_{\GG'} \otimes\psi_{\UU,j})}$, thus $\rho^\circ_i$
  is $(\XXall\EE\setminus\GG,\UU\GG')$-separable.
  Since in the EU-case, $\EE'=\EE\setminus\GG$
  and $\UU'=\UU\GG'$,
  we have that $\rho^\circ_i$ is $(\XXall\EE',\UU')$-separable.

  \emph{In the UE-case,} we show that
  $\rho^\circ_i=\sum_j\pb\proj{(\psi_{\XXall\EE,j}\otimes\psi'_{\GG'})\otimes
    \psi'_{\UU\setminus\GG,j}}$ for some
  $\psi_{\XXall\EE,j}\in\elltwov{\XXall\EE}$,
  $\psi'_{\GG'}\in\elltwov{\GG'}$,
  and $\psi'_{\UU\setminus\GG,j}\in\elltwov{\UU\setminus\GG,j}$ (analogous to the EU-case).
  Thus $\rho^\circ_i$
  is $(\XXall\EE\GG',\UU\setminus\GG)$-separable.
  Since in the UE-case, $\EE'=\EE\GG'$
  and $\UU'=\UU\setminus\GG$,
  we have that $\rho^\circ_i$ is $(\XXall\EE',\UU')$-separable.

  Thus in all four cases, $\rho^\circ_i$
  is $(\XXall\EE',\UU')$-separable.
  Furthermore, we showed above that $\suppo\rho^\circ_i\subseteq\PB$,
  and defined $\rho_i:=\partr{\EE'\UU'}\rho^\circ_i$. Thus $\sats{\rho_i}\PB$.

  Let
  $\mathcal E(\sigma):=\sum_i\oponp{M_i}{\GG'}\sigma
  \adj{\oponp{M_i}{\GG'}}$. Since $\sum_i \adj{M_i}M_i=\id$,
  $\mathcal E$
  is trace-preserving.
  $\sum_i\rho_i =
  \sum_i\partr{\EE'\UU'}\rho^\circ_i=\sum_i\partr{\EE'\UU'}\oponp{M_i}{\GG'}\tilde\rho^\circ
  \adj{\oponp{M_i}{\GG'}} = \partr{\EE'\UU'}\mathcal
  E(\tilde\rho^\circ) $.  Since $\mathcal E$
  is trace-preserving and operates only on $\GG'\subseteq\EE'\UU'$,
  we have
  $\partr{\EE'\UU'}\mathcal
  E(\tilde\rho^\circ)=\partr{\EE'\UU'}\tilde\rho^\circ=\rho$. Thus
  $\sum_i\rho_i=\rho$.
  By \autoref{lemma:sats.sum}, $\sats{\rho_i}\PB$
  then implies $\sats{\rho}\PB$.
  
  Since this holds for all $\sats\rho\PA$, we have $\PA\impl\PB$.
\end{proof}

\begin{lemma}\label{lemma:ShapeShift}
  \Ruleref{ShapeShift} is sound.
\end{lemma}

\begin{proof}
  Since $(\EE\cup\EE')\cap\fv(\PA)=\varnothing$,
  we can interpret $\PA$ as a predicate
  over $\XXall\Tilde\EE\Tilde\UU$
  for some $\Tilde\EE\Tilde\UU$ with $\Tilde\EE\cap(\EE\cup\EE')=\varnothing$. (See the discussion
  on \autopageref{page:pred.identify} about identifying predicates over different sets.)
  Then the pre- and postconditions
  ${\paren{\XX\EE\quanteq\psi,\ \PA}}$ and 
  ${\paren{\XX\EE'\quanteq\psi',\ \PA}}$ can be written more explicitly as predicates
  ${\paren{\XX\EE\quanteq\psi,\ \PA\otimes\elltwov{\EE}}}$ and 
  ${\paren{\XX\EE'\quanteq\psi',\ \PA\otimes\elltwov{\EE'}}}$
  over $\XXall\Tilde\EE\EE\Tilde\UU$ and $\XXall\Tilde\EE\EE'\Tilde\UU$, respectively. (Justified by 
  \autoref{lemma:ghost.sets}.)
  And $\psi,\psi'$ are quantum memories over $\XX\EE$
  and $\XX\EE'$, respectively.

  Fix a mixed memory $\rho$
  over $\XXall$
  with $\sats\rho {\paren{\XX\EE\quanteq\psi,\ \PA\otimes\elltwov{\EE}}}$.
  To show \ruleref{ShapeShift}, we need to show that
  $\sats\rho{\paren{\XX\EE'\quanteq\psi',\ \PA\otimes\elltwov{\EE'}}}$.

  Since
  $\sats\rho{\paren{\XX\EE\quanteq\psi,\ \PA\otimes\elltwov\EE}}$,
  there is an $(\XXall\Tilde\EE\EE,\Tilde\UU)$-separable
  $\rho^\circ$
  with
  $\suppo\rho^\circ\subseteq{\paren{\XX\EE\quanteq\psi,\
      \PA\otimes\elltwov\EE}}$ and $\partr{\Tilde\EE\EE\Tilde\UU}\rho^\circ=\rho$.

  Thus $\suppo\rho^\circ\subseteq {\paren{\XX\EE\quanteq\psi}}$.
  Hence $\rho^\circ = \paren{\partr{\XX\EE}\rho^\circ} \otimes
  \proj{\psi}$.
  Let
  $\tilde\rho^\circ := \paren{\partr{\XX\EE}\rho^\circ} \otimes
  \proj{\psi'}$.
  Then
  \begin{align}
    \partr{\EE}\rho^\circ
    &=
    \partr{\EE} \pb\paren{ \paren{\partr{\XX\EE}\rho^\circ} \otimes \proj{\psi}}
    =
    \paren{\partr{\XX\EE}\rho^\circ} \otimes \partr\EE \proj{\psi} \notag\\
    &\starrel=
    \paren{\partr{\XX\EE}\rho^\circ} \otimes \partr{\EE'} \proj{\psi'} 
    =
    \partr{\EE'} \pb\paren{ \paren{\partr{\XX\EE}\rho^\circ} \otimes \proj{\psi}}
    =
      \partr{\EE'} \tilde\rho^\circ.
      \label{eq:rel-rho-tilderho}
  \end{align}
  Here $(*)$ is by assumption from \ruleref{ShapeShift}.
  Then
  \begin{equation*}
    \partr{\Tilde\EE\EE'\Tilde\UU} \tilde\rho^\circ
    =
    \partr{\Tilde\EE\Tilde\UU} \partr{\EE'} \tilde\rho^\circ
    \eqrefrel{eq:rel-rho-tilderho}=
    \partr{\Tilde\EE\Tilde\UU} \partr{\EE} \rho^\circ
    =
    \partr{\Tilde\EE\EE\Tilde\UU} \rho^\circ
    =
    \rho.
  \end{equation*}

  Since $\rho^\circ$
  is $(\XXall\Tilde\EE\EE,\Tilde\UU)$-separable,
  $\partr{\XX\EE}\rho^\circ$
  is $(\XXall\Tilde\EE\setminus\XX,\Tilde\UU)$-separable,
  and hence
  $\tilde\rho^\circ = \paren{\partr{\XX\EE}\rho^\circ} \otimes
  \proj{\psi'}$ is
  $(\XXall\Tilde\EE\setminus\XX\dotcup\XX\EE',\Tilde\UU)=(\XXall\Tilde\EE\EE',\Tilde\UU)$-separable.

  Since
  $\suppo\rho^\circ\subseteq{\paren{\XX\EE\quanteq\psi,\
      \PA\otimes\elltwov\EE}}$, we have
  $\suppo\rho^\circ\subseteq{{ \PA\otimes\elltwov\EE}}$.
  By \autoref{lemma:suppo.partr}, this implies
  $\suppo\partr\EE{\rho^\circ}\subseteq{{\PA}}$.
  Then by \eqref{eq:rel-rho-tilderho},
  $\suppo\partr{\EE'}{\tilde\rho^\circ}\subseteq{{\PA}}$.
  And by \autoref{lemma:suppo.partr},
  $\suppo{\tilde\rho^\circ}\subseteq{{\PA}\otimes\elltwov{\EE'}}$.
  And since
  $\tilde\rho^\circ = \paren{\partr{\XX\EE}\rho^\circ} \otimes
  \proj{\psi'}$, we also have
  $\suppo\tilde\rho^\circ\subseteq\paren{\XX\EE'=\psi'}$.
  Thus
  $\suppo\tilde\rho^\circ\subseteq {\paren{\XX\EE'\quanteq\psi',\
      \PA\otimes\elltwov{\EE'}}}$.

  Since $\tilde\rho^\circ$
  is $(\XXall\Tilde\EE\EE',\Tilde\UU)$-separable
  and satisfies
  $\suppo\tilde\rho^\circ\subseteq {\paren{\XX\EE'\quanteq\psi',\
      \PA\otimes\elltwov{\EE'}}}$ and
  $\partr{\Tilde\EE\EE'\Tilde\UU}\tilde\rho^\circ=\rho$, it follows that 
  $\sats\rho{\paren{\XX\EE'\quanteq\psi',\ \PA\otimes\elltwov{\EE'}}}$.

  Since $\rho$
  was arbitrary with
  $\sats\rho {\paren{\XX\EE\quanteq\psi,\ \PA\otimes\elltwov{\EE}}}$,
  we have shown
  $ {\paren{\XX\EE\quanteq\psi,\ \PA\otimes\elltwov{\EE}}}
  \impl{\paren{\XX\EE'\quanteq\psi',\ \PA\otimes\elltwov{\EE'}}}$, the
  conclusion of the rule.
\end{proof}

\section{Derived rules}
\label{sec:deriv.rules}

In this section, we show that our eleven core rules are powerful enough
to derive a number of new rules without having to refer to the
semantics of Hoare judgments with ghosts from
\autoref{def:hoare}. (That is, the rules in this section would hold
for any definition of Hoare judgments satisfying the eleven core rules.)
\shortonly{All proofs for this section are deferred to \autoref{app:proofs-derived}.}

This first derived rule is relatively trivial but of high importance:
\[
  \RULE{Conseq}{
    \PA\subseteq\PA'\\
    \hl{\PA'}\bc{\PB'}
    \PB'\subseteq\PB\\
  }{
    \hl\PA\bc\PB
  }
\]
\begin{proof}
  From \rulerefx{Skip} and \rulerefx{Seq}, we get
  $\hl\PA{\SKIP;\bc;\SKIP}\PB$.
  Since $\SKIP$
  is the neutral element for~$;$ (\autopageref{page:seq.assoc.skip.neutral}),
  this implies $\hl\PA{\bc}\PB$.
\end{proof}

\subsection{Derived rules for derived language elements}
\label{sec:deriv.lang}

\paragraph{Initialization.} The next rules deal with the initialization of variables. They are
generalizations of \ruleref{Init}, dealing with the syntactic sugar
$\inits\XX\psi$
(initialization) and $\initc\XX z$
(classical initialization).
\begin{ruleblock}
  \RULE{InitQ}{}{\pb\hl\PA{\inits\XX\psi}{\psubst\PA\ee\XX,\ \XX\quanteq\psi}}
  \RULE{InitC}{}{\pb\hl\PA{\initc\XX z}{\psubst\PA\ee\XX,\ \XX\quanteq\ket z,\ \class\XX}}
\end{ruleblock}
\fullshort{The derivation of \ruleref{InitQ} is not difficult, but the proof of
\ruleref{InitQ} provides a nice first example of reasoning with
sequences of Hoare judgments. The derivation of \ruleref{InitC} is a little more
involved, and it gives an example how to use rules \rulerefx{ShapeShift} and \rulerefx{Transmute} to
show that the content of a variable is classical.}{
The derivation of \ruleref{InitQ} is straightforward, but
the derivation of \ruleref{InitC} requires the use of rules
\rulerefx{ShapeShift} and \rulerefx{Transmute} to
show that the $\XX$ is classical.}
\begin{proof}[of \rulerefx{InitQ}]
  $\inits\XX\psi$
  is syntactic sugar for
  $\init{\xx_1};\dots;\init{\xx_n};\apply{U_\psi}\XX$
  where $U_\psi\ket{0,\dots,0}=\psi$ and  $\xx_1\dots\xx_n:=\XX$.
  For variables $\EE:=\ee_1\dots\ee_n$ and $\ee$ of the same type as $\XX$, we have:
  \begin{align*}
    \pb\hlfrag{\PA}\
    &\init{\xx_1}\ 
      \pb\hlfrag{\psubst\PA{\ee_1}{\xx_1},\ \xx_1\quanteq\ket0} && \text{(\ruleref{Init})} \\
    &\init{\xx_2}\
      \pb\hlfrag{\psubst\PA{\ee_1\ee_2}{\xx_1\xx_2},\ \xx_1\quanteq\ket{0},\ \xx_2\quanteq\ket0}
     && \text{(\ruleref{Init})} \\
    &\dots\ 
      \pb\hlfrag{\psubst\PA{\ee_1\dots\ee_{n-1}}{\xx_1\dots\xx_{n-1}},\
      \xx_1\quanteq\ket0,\dots,\xx_{n-1}\quanteq\ket{0}} \\
    & \init{\xx_n}\ 
      \pb\hlfrag{\psubst\PA{\EE}{\XX},\ \XX\quanteq\ket0}  && \text{(\ruleref{Init})} \\
    & \apply{U_\psi}\XX\
      \pb\hlfrag{\oppred{\oponp{U_\psi}\XX}{\paren{\psubst\PA{\EE}{\XX},\ \XX\quanteq\ket0}}}
     && \text{(\ruleref{Apply})} \\
    & =\
      \pb\hlfrag{{\psubst\PA{\EE}{\XX},\ \XX\quanteq \psi}} \\
    & \impl\ 
      \pb\hlfrag{{\psubst\PA{\ee}{\XX},\ \XX\quanteq \psi}}
     && \text{(\ruleref{Rename})} 
  \end{align*}
  Then by \ruleref{Seq}, we get
  $\hl\PA{\inits\XX\psi}{{\psubst\PA{\ee}{\XX},\ \XX\quanteq \psi}}$.
\end{proof}

\begin{proof}[of \ruleref{InitC}]
  Recall that $\initc\XX z$ is syntactic sugar for $\inits\XX{\ket z}$. Thus we have
  \begin{align*}
    \pb\hlfrag\PA\
    & \initc\XX z\
      \pb\hlfrag{\psubst\PA\ee\XX,\ \XX\quanteq\ket z}
    && \text{(\ruleref{InitQ})}
    \\
    & \impl\
      \pb\hlfrag{\psubst\PA\ee\XX,\ \XX{\ee'}\quanteq\ket z\otimes\ket z}
    && \text{(\ruleref{ShapeShift} with $\EE:=\varnothing$, $\EE':=\ee'$)}
    \\
    & \impl\
      \pb\hlfrag{
      \bigvee\nolimits_i\pb\paren{
      \oppred{\oponp{\proj i}\uu}
      {\paren{\psubst\PA\ee\XX,\ \XX\uu\quanteq\ket z\otimes\ket z}}}}
    \hskip-2in\\
    &&& \hskip-1in\text{(\ruleref{Transmute} with $M_i:=\proj i$, $\GG:=\ee'$, $\GG':=\uu$)}
    \\
    & \starrel=\
      \pb\hlfrag{
      {{\psubst\PA\ee\XX,\ \XX\uu\quanteq\ket z\otimes\ket z}}}
    \ \starstarrel=\
      \pb\hlfrag{
      {{\psubst\PA\ee\XX,\ \XX\uu\quanteq\ket z\otimes\ket z,\ \XX\CLASSEQ\uu}}}
      \hskip-2in
    \\
    & \subseteq\ 
      \pb\hlfrag{
      {{\psubst\PA\ee\XX,\ \XX\quanteq\ket z,\ \XX\CLASSEQ\uu}}}
      \ =\
      \pb\hlfrag{
      {{\psubst\PA\ee\XX,\ \XX\quanteq\ket z,\ \class\XX}}}
      \hskip-2in
  \end{align*}
  Here $(*)$
  follows since
  $\oppred{\oponp{\proj i}\uu} {\pb\paren{\psubst\PA\ee\XX,\
      \XX\uu\quanteq\ket z\otimes\ket z}} = {\pb\paren{\psubst\PA\ee\XX,\
      \XX\uu\quanteq\ket z\otimes\ket z}}$ for $i=z$
  and $=0$
  otherwise.  And $(**)$
  follows since
  $\paren{\XX\uu\quanteq\ket z\otimes\ket z} \subseteq
  \paren{\XX\CLASSEQ\uu}$.

  By rules \rulerefx{Seq} and \rulerefx{Conseq}, it follows that
  $\pb\hl\PA{\initc\XX z}{{\psubst\PA\ee\XX,\ \XX\quanteq\ket z,\
      \class\XX}}$.
\end{proof}

\paragraph{Measurements.} More interesting is the rule for measurements because it actively makes use
of ghosts to record the distribution of outcomes. We first look at the
rule for measurements that forget the outcome ($\measuref\XX$
instead of $\measure\YY\XX$)
because it is a bit simpler, and the underlying ideas a the same:
\[
  \RULE{MeasureForget}{}{
    \pb\hl{\PA}{\measuref \XX}{\oppred{\Ucopy\XX\ee}\PA}}
\]
Here \symbolindexmark\Ucopy{$\Ucopy\VV\WW$}
is the isometry from $\VV$
to $\VV\WW$
defined by $\Ucopy\VV\WW\,\ket i_{\VV}=\ket i_{\VV}\otimes\ket
i_{\WW}$. Thus $\Ucopy\XX\ee$
is the operation that ``classically copies'' (in the computational
basis) the content of $\XX$
to the fresh entangled ghost $\ee$.
(We have $\ee\notin\fv(\PA)$
is fresh because otherwise $\oppred{\Ucopy\XX\ee}\PA$
would not be well-defined since it would contain two $\ee$'s.)

In other words, \ruleref{MeasureForget} says that after measuring
$\XX$,
the result is simply to get $\XX$
entangled with a fresh entangled ghost $\ee$.
Since $\ee$
is a ghost, being entangled with it effectively means that $\XX$
has been measured. (It is a well-known fact in quantum information
that entangling with a subsystem that is not observed any more
effectively measures a state.) Thus, the predicate
$\oppred{\Ucopy\XX\ee}\PA$
encodes the fact that $\XX$
has been measured. At the same time, this predicate does not forget
about the probabilities of the different measurement outcomes. This is
best illustrated by an example: Let $\xx$
be of type integer and $\psi:=\sqrt{2/3}\ket1+\sqrt{1/3}\ket2$.
We have
$\hl{\xx\quanteq\psi}{\measuref\xx}{\oppred{\Ucopy\xx\ee}{\paren{\xx\quanteq\psi}}}=
\hlfrag{\xx\ee\quanteq\Ucopy\xx\ee\psi}= \hlfrag{\xx\ee\quanteq\psi'}$
with $\psi':=\sqrt{2/3}\ket{11}+\sqrt{1/3}\ket{22}$.
As we see, the postcondition encodes the probabilities of measuring
$1$
and $2$
(namely, $2/3$
and $1/3$).
Note that the postcondition $\xx\ee\quanteq\psi'$
does not mean that $\xx$
is actually entangled with something. Since $\ee$
is a ghost, it only means that $\xx$
is in a state that can be seen as a hypothetical entanglement with
some $\ee$.
In fact, it is easy to see (\autoref{lemma:distrib}) that the only
mixed memory on $\xx$
satifying $\xx\ee\quanteq\psi'$
is $\rho=\frac23\proj{\ket1}+\frac13\proj{\ket2}$,
as expected. Thus the postcondition faithfully encodes the
probabilities of the measurement outcome, something that would not
have been possible without using ghosts.

Since measurements are merely syntactic sugar in our language, it
turns out that \ruleref{MeasureForget} can be easily derived from the
more basic rules we saw so far:
\begin{proof}[of \rulerefx{MeasureForget}]
  Recall from \autopageref{page:measuref} that $\measuref\XX$ is syntactic sugar for
  ``$\inits{\zz}{\ket0};\apply{\CNOT}{\XX\zz};\inits{\zz}{\ket0}$''
  for some fresh $\zz$ (in particular, $\zz\notin\fv(\PA),\XX$). 
  We then have for some fresh $\ee$:
  \begin{align*}
    \pb\hlfrag{\PA}\
    &\inits\zz{\ket0}\ 
      \pb\hlfrag{\psubst\PA\ee\zz,\ \zz\quanteq\ket0}
      \ =\ \pb\hlfrag{\PA,\ \zz\quanteq\ket0}
    && \text{(\ruleref{InitQ})} \\
    &\apply{\CNOT}{\XX\zz}\
      \pb\hlfrag{\oppred{\oponp{\CNOT}{\XX\zz}}{\paren{\PA,\ \zz\quanteq\ket0}}}
    && \text{(\ruleref{Apply})} \\
    & \!\starrel=\
      \pb\hlfrag{\oppred{\Ucopy\XX\zz}\PA}
      \\
    &\inits\zz{\ket0}\ 
      \pb\hlfrag{\oppred{\Ucopy\XX\ee}\PA,\ \zz\quanteq\ket0}
    && \text{(\ruleref{InitQ})} \\
    & \!\subseteq\ 
      \pb\hlfrag{\oppred{\Ucopy\XX\ee}\PA}.
  \end{align*}
  Here $(*)$ follows since $\oponp{\CNOT}{\XX\zz}(\psi\otimes\ket0)=\Ucopy\XX\zz\psi$.
  
  Then by rules \rulerefx{Seq} and \rulerefx{Conseq}, we get
  $\hl\PA{\measuref\XX}{\oppred{\Ucopy\XX\ee}\PA}$.
\end{proof}

The postcondition of \ruleref{MeasureForget} encodes both the
distribution of outcomes, as well as the state after the
measurement. Sometimes, it may not be necessary to remember the
distribution (only which outcomes are possible). In this case we can
use the following weaker rule:
\[
  \RULE{MeasureForget*}{}{
    \pB\hl{\PA}{\measuref \XX}{\class\XX,\ \textstyle \bigvee\nolimits_i\oppred{\oponp{\proj{\ket i}}\XX}\PA}}
\]
To understand the rule, it is easiest to look at the same example as
above.  Recall that $\psi=\sqrt{2/3}\ket1+\sqrt{1/3}\ket2$
and $\xx$
is of type integer. Then \rulerefx{MeasureForget*} implies
\[
  \pb\hl{\xx\quanteq\psi}{\measuref\xx}{\class\xx,\ \textstyle \bigvee_i
    \oppred{\oponp{\proj{\ket i}}\xx}{\paren{\xx\quanteq\psi}}}
  = \hlfrag{\class\xx,\ \textstyle \bigvee_i
    {\xx\quanteq\proj{\ket i}\psi}}.
\]
Since $\proj{\ket i}\psi=0$
for $i\notin\{1,2\}$
and $\proj{\ket i}\psi=\ket i$ up to scalar factor
for $i=1,2$, we have that $\bigvee_i {\xx\quanteq\proj{\ket i}\psi}$ equals $\xx\quanteq\ket1\vee\xx\quanteq\ket2$. Thus
\[
  \pb\hl{\xx\quanteq\psi}{\measuref\xx}{\class\xx,\
    \xx\quanteq\ket1\vee\xx\quanteq\ket2}.
\]
In other words, after measuring $\xx$,
$\xx$ will be classical and have a state $\ket1$ or $\ket2$.
\Ruleref{MeasureForget*} is derived from \rulerefx{MeasureForget} by \ruleref{Transmute}:
\begin{proof}[of \rulerefx{MeasureForget*}]
  \begin{align*}
    \pb\hlfrag\PA\
    & \measuref\XX\
      \pb\hlfrag{\oppred{\Ucopy\XX\ee}\PA}
    && \text{(\ruleref{MeasureForget})}
    \\
    & \impl\
      \pb\hlfrag{
      \textstyle \bigvee\nolimits_i\pb\paren{
      \oppred{\oponp{\proj{\ket i}}\uu}
      {\paren{\oppred{\Ucopy\XX\uu}\PA}}}}
    \hskip-2in\\
    &&& \hskip-1.5in\text{(\ruleref{Transmute} with $M_i:=\proj{\ket i}$, $\GG:=\ee$, $\GG':=\uu$)}
    \\
    & \starrel=\
      \pb\hlfrag{
      \textstyle \bigvee\nolimits_i\pb\paren{
      \oppred{\oponp{\proj{\ket i}}\XX}
      {\PA},\ \uu\quanteq\ket i}}
      \hskip-2in \\
    & \starstarrel=\
      \pb\hlfrag{
      \textstyle \bigvee\nolimits_i\pb\paren{
      \oppred{\oponp{\proj{\ket i}}\XX}
      {\PA},\ \XX\quanteq\ket i,\ \uu\quanteq\ket i}}
     \\
    & \subseteq\
      \pb\hlfrag{
      \textstyle \bigvee\nolimits_i\pb\paren{
      \oppred{\oponp{\proj{\ket i}}\XX}
      {\PA},\ \XX\CLASSEQ\uu}}
     \\
    & \subseteq\
      \pb\hlfrag{\XX\CLASSEQ\uu,\
      \textstyle \bigvee\nolimits_i{
      \oppred{\oponp{\proj{\ket i}}\XX}
      {\PA}}}
    \\
    & =\
      \pb\hlfrag{\class\XX,\
      \textstyle \bigvee\nolimits_i{
      \oppred{\oponp{\proj{\ket i}}\XX}
      {\PA}}}.
  \end{align*}
  Here $(*)$
  follows since
  ${\oponp{\proj{\ket i}}\uu}{\Ucopy\XX\uu}\phi=
  {\oponp{\proj{\ket i}}\XX}\phi\otimes\ket i_{\uu}$ for all $\phi$.
  And $(**)$ follows since 
  $\oppred{\oponp{\proj{\ket i}}\XX}
      {\PA}
  \subseteq
  \paren{\XX\quanteq\ket i}$.
\end{proof}
Above, we studied measurements that forget their outcome
($\measuref\XX$).
When we consider measurement that remember their outcome
($\measure\YY\XX$),
we get the following analogues to \rulerefx{MeasureForget} and
\rulerefx{MeasureForget*}:
\begin{ruleblock}
  \RULE{Measure}{}{
    \pb\hl{\PA}{\measure \YY\XX}{\oppred{\Ucopy\XX{\ee'}}{\oppred{\Ucopy\XX\YY}{\psubst\PA\ee\YY}}}}
  \RULE{Measure*}{}{
    \pb\hl{\PA}{\measure \YY \XX}{\class\XX,\ \class\YY,\ \textstyle \bigvee\nolimits_i\pb\paren{\oppred{\oponp{\proj{\ket i}}\XX}{\psubst\PA\ee\YY},\ \YY\quanteq\ket i}}}
\end{ruleblock}
As one can see, the only differences to \rulerefx{MeasureForget} and
\rulerefx{MeasureForget*} is that the measurement is additionally
written to $\YY$
(either via $\Ucopy\XX\YY$
or via $\YY\quanteq\ket i$).
And additionally, $\YY$
is replaced by $\ee$
in $\PA$
if it occurs there because it is overwritten (analogous to
\ruleref{Init}, \ruleref{InitQ}, \ruleref{InitC}).  For example, with
$\psi$
and $\xx$
as above, we get
$\hl{\xx\quanteq\psi}{\measure\yy\xx}{\xx\yy\ee\quanteq\psi''}$
with $\psi'' :=\sqrt{2/3}\ket{111}+\sqrt{1/3}\ket{222}$
and
$\hl{\xx\quanteq\psi}{\measure\yy\xx}{\class\xx,\ \class\yy,\
  \xx\yy\quanteq\ket{11} \vee \xx\yy\quanteq\ket{22}}$.  The proofs of
these rules are very similar to those of \rulerefx{MeasureForget} and
\rulerefx{MeasureForget*}:

\begin{proof}[of \rulerefx{Measure}]
  Recall from \autopageref{page:measure} that $\measure\YY\XX$ is syntactic sugar for
  ``$\inits\YY{\ket0};\inits{\zz}{\ket0};\apply\CNOT{\XX\YY};\apply{\CNOT}{\XX\zz};\inits{\zz}{\ket0}$''
  for some fresh $\zz$ (in particular, $\zz\notin\fv(\PA),\XX,\YY$). 
  We then have for some fresh $\ee,\ee'$:
  \begin{align*}
    \pb\hlfrag{\PA}\
    &\inits\YY{\ket0}\
      \pb\hlfrag{\psubst\PA{\ee}\YY,\ \YY\quanteq\ket0}
    && \text{(\ruleref{InitQ})} \\
    &\inits\zz{\ket0}\ 
      \pb\hlfrag{\psubst{\pb\paren{\psubst\PA{\ee}\YY,\ \YY\quanteq\ket0}}{\ee'}\zz,\ \zz\quanteq\ket0}
    && \text{(\ruleref{InitQ})} \\
    &\!=\
      \pb\hlfrag{\psubst\PA{\ee}\YY,\ \YY\quanteq\ket0,\ \zz\quanteq\ket0}
      && (\zz\notin\fv(A),\YY,\ee)
      \\
    &\apply{\CNOT}{\XX\YY}\
      \pb\hlfrag{\oppred{\oponp{\CNOT}{\XX\YY}}
      {\pb\paren{\psubst\PA{\ee}\YY,\ \YY\quanteq\ket0}},\ \zz\quanteq\ket0}
    && \text{(\ruleref{Apply}, $\zz\notin\XX\YY$)} \\
    &\!\starrel=\
      \pb\hlfrag{\oppred{\Ucopy\XX\YY}{\psubst\PA{\ee}\YY},\ \zz\quanteq\ket0}
      \\
    &\apply{\CNOT}{\XX\zz}\
      \pb\hlfrag{\oppred{\oponp{\CNOT}{\XX\zz}}{\pb\paren{\oppred{\Ucopy\XX\YY}{\psubst\PA{\ee}\YY},\ \zz\quanteq\ket0}}}
    && \text{(\ruleref{Apply})} \\
    & \!\starstarrel=\
      \pb\hlfrag{\oppred{\Ucopy\XX\zz}{\oppred{\Ucopy\XX\YY}{\psubst\PA{\ee}\YY}}}
      \\
    &\inits\zz{\ket0}\ 
      \pb\hlfrag{\oppred{\Ucopy\XX{\ee'}}{\oppred{\Ucopy\XX\YY}{\psubst\PA{\ee}\YY}},\ \zz\quanteq\ket0}
    && \text{(\ruleref{InitQ})} \\
    & \!\subseteq\ 
      \pb\hlfrag{\oppred{\Ucopy\XX{\ee'}}{\oppred{\Ucopy\XX\YY}{\psubst\PA{\ee}\YY}}}.
  \end{align*}
  Here $(*)$
  follows since
  $\oponp{\CNOT}{\XX\YY}(\psi\otimes\ket0)=\Ucopy\XX\YY\psi$
  for all $\psi$. And $(**)$ analogously.
  
  Then by rules \rulerefx{Seq} and \rulerefx{Conseq}, we get
  $\hl\PA{\measure\YY\XX}{\oppred{\Ucopy\XX{\ee'}}{\oppred{\Ucopy\XX\YY}{\psubst\PA{\ee}\YY}}}$.
\end{proof}

\begin{proof}[of \rulerefx{Measure*}]
  We first derive an auxiliary fact:
  \begin{equation}
    \pb\paren{\class\XX,\ \XX\CLASSEQ\YY,\ \PB}
    \impl
    \pb\paren{\class\XX,\ \class\YY,\ \PB}
    \qquad
    \text{for any $\PB$}
    \label{eq:intro.class}
  \end{equation}
  By \ruleref{Case} with $\PC:=\class\XX$,
  $M:=\{\ket z\}$
  and \autoref{lemma:disentangling} (or more simply using the
  \ruleref{CaseClassical} we introduce on
  \autopageref{rule:CaseClassical} below), \eqref{eq:intro.class} follows from:
  \begin{align*}
    &\paren{\XX=\ket z,\ \class\XX,\ \XX\CLASSEQ\YY,\ \PB}
     \subseteq
      \paren{\class\XX,\ \YY=\ket z,\ \PB}
    \\&
    \starrel\impl
    \paren{\class\XX,\ \YY=\ket z,\ \ee=\ket z,\ \PB}
    \starstarrel\impl
    \paren{\class\XX,\ \YY=\ket z,\ \uu=\ket z,\ \PB}
    \\&
    \subseteq
    \paren{\class\XX,\ \YY\CLASSEQ\uu,\ \PB}
    =
    \paren{\class\XX,\ \class\YY,\ \PB}.
  \end{align*}
  Here $(*)$
  follows from \rulerefx{ShapeShift} with $\XX:=\EE:=\varnothing$,
  $\EE':=\ee$,
  $\psi:=1$,
  $\psi':=\ket z_{\ee}$.
  And $(**)$
  follows from \rulerefx{Transmute} with $M_i=\proj{\ket i}$,
  $\GG:=\ee$,
  $\GG':=\uu$, and simplification.
  
  We proceed to the actual proof of \rulerefx{Measure*}. We calculate:
  \begin{align*}
    \pb\hlfrag\PA\
    & \measure\YY\XX\
      \pb\hlfrag{\oppred{\Ucopy\XX{\ee'}}{\oppred{\Ucopy\XX\YY}{\psubst\PA\ee\YY}}}
    && \text{(\ruleref{Measure})}
    \\
    & \impl\
      \pb\hlfrag{
      \textstyle \bigvee\nolimits_i
      \oppred{\oponp{\proj{\ket i}}\XX}
      {\oppred{\Ucopy\XX{\uu}}{\oppred{\Ucopy\XX\YY}{\psubst\PA\ee\YY}}}}
    \hskip-2in\\
    &&& \hskip-2.5in\text{(\ruleref{Transmute} with $M_i:=\proj{\ket i}$, $\GG:=\ee'$, $\GG':=\uu$)}
    \\
    & \starrel=\
      \pb\hlfrag{
      \textstyle \bigvee\nolimits_i\pb\paren{
      \oppred{\oponp{\proj{\ket i}}\XX}
      {{\psubst\PA\ee\YY}},\ \YY\quanteq\ket i,\ \uu\quanteq\ket i}}
      \hskip-2in \\
    & \starstarrel=\
      \pb\hlfrag{
      \textstyle \bigvee\nolimits_i\pb\paren{
      \oppred{\oponp{\proj{\ket i}}\XX}
      {\psubst\PA\ee\YY},\ \XX\quanteq\ket i,\ \YY\quanteq\ket i,\ \uu\quanteq\ket i}}
     \\
    & \subseteq\
      \pb\hlfrag{
      \textstyle \bigvee\nolimits_i\pb\paren{
      \oppred{\oponp{\proj{\ket i}}\XX}
      {\psubst\PA\ee\YY},\ \YY\quanteq\ket i,\ \XX\CLASSEQ\uu,\ \XX\CLASSEQ\YY}}
     \\
    & \subseteq\
      \pb\hlfrag{\XX\CLASSEQ\uu,\
      \XX\CLASSEQ\YY,\
      \textstyle \bigvee\nolimits_i{
      \pb\oppredp{\oponp{\proj{\ket i}}\XX}
      {\psubst\PA\ee\YY,\ \YY\quanteq\ket i}}}
    \\
    & =\
      \pb\hlfrag{\class\XX,\
      \XX\CLASSEQ\YY,\
      \textstyle \bigvee\nolimits_i{
\pb\oppredp{\oponp{\proj{\ket i}}\XX}
      {\psubst\PA\ee\YY,\ \YY\quanteq\ket i}}}
    \\ & \eqrefrel{eq:intro.class}\impl\
         \pb\hlfrag{\class\XX,\
         \class\YY,\
      \textstyle \bigvee\nolimits_i\pb\paren{
      \oppred{\oponp{\proj{\ket i}}\XX}
      {\psubst\PA\ee\YY,\ \YY\quanteq\ket i}}}
  \end{align*}
  Here $(*)$
  follows since
  ${\oponp{\proj{\ket i}}\uu}{\Ucopy\XX\uu}
  {\Ucopy\XX\YY}\phi=
  {\oponp{\proj{\ket i}}\XX}\phi\otimes\ket i_{\uu}\otimes\ket i_{\YY}$ for all $\phi$.
  And $(**)$ follows since 
  $\oppred{\oponp{\proj{\ket i}}\XX}
      {\PA}
  \subseteq
  \paren{\XX\quanteq\ket i}$.
  Then \rulerefx{Measure*} follows using \rulerefx{Conseq} and \rulerefx{Seq}.
\end{proof}

\paragraph{Sampling.} Finally, we consider sampling $\sample\XX
D$. Again, we have a \ruleref{Sample} that remembers the distribution
$D$,
and a \ruleref{Sample*} that only remembers which values can occur.
\begin{ruleblock}
  \RULE{Sample}{}{
    \pb\hl{\PA}{\sample{\XX}{D}}{\psubst\PA\ee\XX,\ \distr\XX D}}
  \RULE{Sample*}{}{
    \pb\hl{\PA}{\sample{\XX}{D}}{\psubst\PA\ee\XX,\ \class\XX,\ \textstyle \bigvee_{i\in\suppd D}
      \XX\quanteq\ket i}  }
\end{ruleblock}
Here \symbolindexmark\suppd{$\suppd D$}
is the support of the distribution $D$,
i.e., $\suppd D=\braces{i:D(i)\neq 0}$.
That is, the postcondition from \rulerefx{Sample} says that $X$
is distributed according to $D$
($\distr\XX D$),
while the postcondition from \ruleref{Sample*} merely says that $\XX$
is classical and has value $i$
($\XX\quanteq\ket i$) for some $i\in\suppd D$.
Both rules can be derived easily using the definition of 
$\distr\XX D$ as syntactic sugar and the rules we have derived above:
\begin{proof}[of \rulerefx{Sample}]
  Recall that $\sample\XX D$
  is syntactic sugar for ``$\inits\XX{\psiD D};\measuref\XX$''
  where $\psiD D=\sum_{i\in T}\sqrt{D(i)}\ket{i}$.
  And $\distr\XX D$
  is syntactic sugar for $\XX\ee\quanteq\psiDD D$
  where ${\psiDD D}=\sum_i \sqrt{D(i)} \ket{i,i}$ and $\ee$ is fresh. We have (with fresh $\ee,\ee'$):
  \begin{align*}
    \pb\hlfrag\PA
    & \inits\XX{\psiD D}\
      \pb\hlfrag{\psubst\PA\ee\XX,\ \XX\quanteq\psiD D}
    && \text{(\ruleref{InitQ})}
    \\
    & \measuref\XX\
      \pb\hlfrag{\oppred{\Ucopy\XX{\ee'}}{\paren{\psubst\PA\ee\XX,\ \XX\quanteq\psiD D}}}
    && \text{(\ruleref{MeasureForget})}
    \\
    & =\
      \pb\hlfrag{\psubst\PA\ee\XX,\ \oppred{\Ucopy\XX{\ee'}}{\paren{\XX\quanteq\psiD D}}}
    \\
    & =\
      \pb\hlfrag{\psubst\PA\ee\XX,\ {\XX\ee\quanteq\psiDD D}}
      \ =\
      \pb\hlfrag{\psubst\PA\ee\XX,\ \distr\XX D}.
      \hskip-5mm
  \end{align*}
  \Ruleref{Sample} then follows by \ruleref{Seq}.
\end{proof}

\begin{proof}[of \rulerefx{Sample*}]
  Recall that $\sample\XX D$
  is syntactic sugar for ``$\inits\XX{\psiD D};\measuref\XX$''
  where $\psiD D=\sum_{i\in T}\sqrt{D(i)}\ket{i}$.
  We have (with fresh $\ee$):
  \begin{align*}
    \pb\hlfrag\PA
    & \inits\XX{\psiD D}\
      \pb\hlfrag{\psubst\PA\ee\XX,\ \XX\quanteq\psiD D}
    && \text{(\ruleref{InitQ})}
    \\
    & \measuref\XX\
      \pb\hlfrag{\class\XX,\ \textstyle\bigvee\nolimits_i
      \oppred{\oponp{\proj{\ket i}}\XX}
      {\paren{\psubst\PA\ee\XX,\ \XX\quanteq\psiD D}}}
      \hskip-1in
      \\
    &&& \text{(\ruleref{MeasureForget*})}
    \\
    & =\
      \pb\hlfrag{\class\XX,\ \textstyle\bigvee\nolimits_i
      \pb\paren{{\psubst\PA\ee\XX,\ \XX\quanteq\proj{\ket i}\,\psiD D}}}
    \\
    & \subseteq\
      \pb\hlfrag{\psubst\PA\ee\XX,\ \class\XX,\ \textstyle\bigvee\nolimits_i
      {{\XX\quanteq\proj{\ket i}\,\psiD D}}}
    \\
    & \starrel=\
      \pb\hlfrag{\psubst\PA\ee\XX,\ \class\XX,\ \textstyle\bigvee\nolimits_{i\in\suppd D}
      {{\XX\quanteq\ket i}}}.
  \end{align*}
  Here $(*)$ follows since $\proj{\ket i}\,\psiD D=0$ for $i\notin\suppd D$,
  and $0\neq \proj{\ket i}\,\psiD D \propto \ket i$ for $i\in\suppd D$.
  
  \Ruleref{Sample*} then follows by rules \rulerefx{Seq} and \rulerefx{Conseq}.
\end{proof}

\paragraph{Final note.}  Without the concept of ghosts, we would not
have been able to express \ruleref{Init} and thus not have been able
to derive the above rules. Instead, we would have had to directly
prove rules for measurements and sampling directly from the
semantics. And without ghosts, those rules would not have been as
expressive, for example, \rulerefx{Sample} would not have been expressible
(i.e., we cannot express what distribution $\XX$
has after sampling), and \ruleref{Sample*} would lack the predicate
$\class\XX$,
i.e., we cannot express that $\XX$
is not a superposition between different $\ket i$ with $i\in\suppd D$.

\subsection{Programs with classical variables}
\label{sec:deriv:classical}
In this section, we show how programs using classical variables can be
conveniently treated in our logic, even though the definition of our
programming language does not contain classical variables. The lack of
classical variables in language and logic has, at the first glance,
a number of negative consequences:
\begin{compactenum}[(i)]
\item Any classical values in a program need to be encoded as quantum
  states. In particular, reasoning steps that hold only for classical
  variables cannot be applied. (E.g., a case distinction over the
  value of the classical variable.)
\item Quantum operations cannot be parametrized by classical
  values. For example, we might wish to model a program step such as
  $\apply{U_{\yy}}\xx$,
  i.e., $U_i$
  is a family of isometries, and the classical variable $\yy$
  selects which of them is applied to $\xx$.
  For example, in \cite{qrhl}, every program step can be parametrized
  by all classical variables, and this possibility is essential for
  expressing more complicated programs (e.g., the cryptographic
  schemes analyzed there).
\item Predicates cannot depend on classical values. For example, we
  might wish to say something like $\xx\quantin S_{\yy}$,
  i.e., $S_i$
  is a family of subspaces, and $\xx$
  lies in the subspaces $S_{\yy}$
  selected by the classical variable $\yy$.
  For example, \cite{qrhl} handles this by defining predicates to be
  families of subspaces indexed by the values of the classical
  variables (and not simply subspaces as is the case here). Such
  predicates are necessary for more complex analyses, e.g., think of
  Grover's algorithm \cite{grover} where the loop invariant would have
  to state that the quantum register is in a state that depends on how
  many iterations have been performed so far (the iteration counter
  being a classical variable).
\end{compactenum}
As we see, a special treatment of classical variables is almost
essential for convenient reasoning about hybrid programs (i.e.,
programs that contain both classical and quantum values), yet such a
special treatment comes with a large formal overhead (the semantics
are more complex, all proofs need to distinguish between classical and
quantum variables). In this section, we will see how ghost variables
allow us to recover the benefits of classical variables without the
formal overhead, simply by introducing additional syntactic sugar and
some derived rules.

\paragraph{Syntactic sugar for programs.}
In our language, $\apply U\XX$
requires $U$
to be a constant. Since $\inits\XX\psi$,
$\initc\XX z$,
and $\sample\XX D$
are all syntactic sugar based on $\applykw$,
they inherit this restriction, i.e., $\psi,z,D$
are constants as well. Thus we cannot even write something as simple
as $\initc\xx\yy$,
meaning we assign the content of the classical variable $\yy$
to $\xx$.
We introduce some syntactic sugar for $\applykw$
that solves this problem:

Consider the term $\apply U\XX$\pagelabel{page:applysugar}
where $U$
is an expression containing program variables $\YY$
of type $T$
(disjoint from $\XX$).
For any assignment $z$
to the variables $\YY$,
$U\{z/\YY\}$
defines an isometry $\UX Uz$
($U$
evaluated for $\YY:=z$).
Let \symbolindexmark\UX{$\UX U\YY$}
be the isometry on $\YY\XX$
defined by
$\UX U\YY(\ket{i}_{\YY}\otimes\psi):=\ket{i}_{\YY}\otimes
U_i\psi$.  (That is $\UX U\YY$
is a controlled operation, like CNOT.)  Finally,
\symbolindexmark\apply{$\apply U\XX$}
is syntactic sugar for $\apply{\UX U\YY}{\YY\XX}$.

This notation is best understood by looking at a typical example.
Consider $\apply{e^{-2i\pi\yy H}}\xx$
where $H$
is some fixed Hermitian operator, and $\yy$
has type $\setR$.
Since $U:=e^{-2\pi i\yy H}$
contains the variable $\yy$,
it defines the family $\UX{U}z:=e^{-2\pi iz H}$ of unitaries.
Then $\UX U\yy(\ket z_{\YY}\otimes\psi)=\ket z_{\YY}\otimes e^{-2\pi iz H}\psi$. Then 
$\apply{e^{-2i\pi\yy H}}\xx\ =\ \apply{\UX U\yy}{\yy\xx}$ applies
$\UX Uz:=e^{-2\pi iz H}$ to $\xx$ if $\yy$ is in state $\ket z$, as expected.

Note that this notation does not require that $\YY$
refers to classical variables. It is meaningful to use this notation
when $\YY$
does not contain classical data. However, in the remainder of this
paper, we will only use this notation when we think of $\YY$
as classical variables.

Since $\inits\XX\psi$,
$\initc\XX z$,
and $\sample\XX D$
are all syntactic sugar based on $\applykw$,
this notation automatically carries over to those constructs, too. For
example, $\inits\xx\yy$
is syntactic sugar for $\init\xx;\apply{U_{\ket\yy}}\xx$
with $U_{\ket z}\ket0:=\ket z$
which is syntactic sugar for $\init\xx;\apply{\UX{\paren{U_{\ket\yy}}}\yy}{\yy\xx}$
where ${\UX{\paren{U_{\ket\yy}}}\yy}\ket{z}\ket{0}=\ket{z}\ket{z}$.
Hence $\inits\xx\yy$
will initialize~$\yy$
with~$\ket z$
when~$\xx$
contains~$\ket z$,
as expected.
\fullonly{Similarly, had we chosen to define a more complicated measurement
command in \autoref{sec:qprogs} (instead of $\measure\YY\XX$) that takes the measurement
basis as an additional argument, then that measurement command would
generalize analogously and allow us specify a basis that depends on
classical variables.}

\paragraph{Syntactic sugar for predicates.}
We use similar syntactic sugar for writing predicates that depend on
classical variables. Without such syntactic sugar, predicates such as
$\xx\quanteq\psi$
can only contain a constant $\psi$,
i.e., $\psi$ cannot depend on classical variables.  An expression $\PA$
containing some (supposedly classical) variables $\YY$
of type $T$
defines a family $\PX\PA z$
($z\in T$)
of predicates with $\fv(\PX\PA z)\cap\YY=\varnothing$,
resulting from substituting $\YY$
by $z$
in the expression $\PA$.
We then define the predicate
\symbolindexmarkonly{\PX}$\symbolindexmarkhighlight{\PX\PA\YY}:=
\bigvee_z\pb\paren{\YY\quanteq\ket z,\ \PX\PA z}$.  (That
is, $\ket z_{\YY}\otimes\psi\in \PX\PA\YY$
iff $\psi\in\PX\PA z$.)
We can now use the notation $\PX\PA\YY$
in pre-/postconditions to parametrize predicates by the values of
classical variables. In most cases, we omit the $\PX{}\YY$,
writing simply the predicate $\PA$.
While this notation is potentially ambiguous, in most cases it will be
clear where the $\PX{}\YY$
has to be added since otherwise the pre-/postconditions will not be
welltyped.

We illustrate this by example: Consider the judgment
$\pb\hl{\class\yy,\ \xx\quanteq\ket0} {\apply{e^{-2i\pi\yy
      H}}\xx}{\class\yy,\ \xx\quanteq {e^{-2i\pi\yy H}}\ket 0}$.  We
have seen in the previous paragraph how to read
${\apply{e^{-2i\pi\yy H}}\xx}$.
The precondition does not contain any syntactic sugar related to
classical variables. We now translate the postcondition.  Without
omission of the implicitly understood $\PX{}\yy$,
it reads
$\class\yy,\ \PX{\paren{\xx\quanteq {e^{-2i\pi\yy H}}\ket
      0}}\yy$. \fullonly{(The $\PX{}\yy$
  cannot be placed elsewhere since including the $\class\yy$
  in it would be mean that our postcondition contain non-welltyped
  subterms terms $\class z$ for real $z$. And 
  we cannot have  $\PX{}{\xx\yy}$ instead of  $\PX{}\yy$ since that would lead to non-welltyped subterms
  $z_1\quanteq\dots$ where $z_1$ is not a variable.)}
Then ${\xx\quanteq {e^{-2i\pi\yy H}}\ket
      0}$ defines a family of predicates 
 ${\xx\quanteq {e^{-2i\pi z H}}\ket
   0}$ for real $z$,  and 
$\PX{\paren{\xx\quanteq {e^{-2i\pi\yy H}}\ket
    0}}\yy$ means 
$\bigvee_z\pb\paren{\yy\quanteq\ket z,\ {\xx\quanteq {e^{-2i\pi z H}}\ket
    0}}$.
Thus the postcondition would, without the syntactic sugar, read 
$\class\yy,\ \bigvee_z\pb\paren{\yy\quanteq\ket z,\ {\xx\quanteq {e^{-2i\pi z H}}\ket
    0}}$.
Of course, given appropriate rules such as \rulerefx{ApplyParam} below,
one rarely needs to actually explicitly unfold the syntactic sugar.

\paragraph{Derived rules.}
Since the syntactic sugar introduced in this section expands to
language constructs for which we already have introduced rules, we
could, in principle, reason about programs involving classical
variables with only the rules above. However, in practice this may be
cumbersome. Therefore we will now introduce a few derived rules
specifically for the dealing with such programs. The first is a simple
consequence of \ruleref{Case}:
\[
  \RULE{CaseClassical}{
    \forall z.\
    \hl{\YY\quanteq\ket z,\ \class\YY,\ \PX\PA z}\bc\PB
  }{
    \hl{\class\YY, \PX\PA\YY}\bc\PB
  }
\]
\begin{proof}
  Let $T$ be the type of $\YY$. For any $z\in T$, we have
  $\paren{\YY\quanteq\ket z,\ \PX\PA z}
  =\paren{\YY\quanteq\ket z,\ \bigvee_{z'}\paren{\YY\quanteq\ket{z'},\ \PX\PA {z'}}}
  =\paren{\YY\quanteq\ket z,\ \PX\PA\YY}$. (The first equality follows since
  $\paren{\YY\quanteq\ket z,\ \YY\quanteq\ket{z'},\ \PX\PA {z'}}=0$ for $z\neq z'$.)
  Thus the premise of \rulerefx{CaseClassical} becomes
  $    \forall z.\
  \hl{\YY\quanteq\ket z,\ \class\YY,\ \PX\PA\YY}\bc\PB$.
  By \rulerefx{Case}   (with $\PA:=\paren{\class\YY,\ \PX\PA\YY}$,
  $\PC:=\class\XX$, $M:=\{\ket z\}_{z\in T}$),
  we get 
  $\hl{\class\YY,\ \PX\PA\YY}\bc\PB$.
  (Using \autoref{lemma:disentangling} to show $\class\XX$ is 
  $\{\ket z\}_{z\in T}$-disentangling.)
\end{proof}
From this rule, we can derive a rule for our classically parametrized
$\applykw$-command:
\[
  \RULE{ApplyParam}{
    \text{$\fv(e)\subseteq\YY$}\\
    \YY\cap\XX=\varnothing
  }{
    \pb\hl{\class\YY,\ \PX\PA\YY}{\apply e\XX}
    {\class\YY,\ \PX{\pb\paren{\oppred{\oponp e\XX}\PA}}\YY}
  }
\]
Note that his rule is basically the same as \ruleref{Apply}
(especially if we write it with omitted~$\PX{}\YY$),
except that we allow expression $e$
that specifies the operation to apply to contain variables~$\YY$
that must be guaranteed to be classical in the precondition
($\class\YY$).

\begin{proof}
  Let $\YY':=\fv(e)$. Let $\YY'':=\YY\setminus\YY'$.
  Let $T,T',T''$ be the types of $\YY,\YY',\YY''$, respectively.
  Then $T=T'\times T''$ (up to a canonical bijection).
  Then $\apply e\XX$
  is syntactic sugar for $\apply{\UX e{\YY'}}{\YY'\XX}$
  where $\UX e{z'}:=e\{z'/\YY'\}$ for $z'\in T'$. For any $z=(z',z'')\in T$, we have:
  \begin{align*}
    &
      \pb\hlfrag{\YY\quanteq\ket{z},\ \class\YY,\ \PX\PA{ z}}
      \\&
    {\apply{\UX e{\YY'}}{\YY'\XX}}\
    \pb\hlfrag{\oppred{\oponp{\UX e{\YY'}}{\YY'\XX}}{\paren{\YY\quanteq\ket{z},\ \class\YY,\ \PX\PA{ z}}}}
    \\&\starrel=\
    \pb\hlfrag{\oppred{\oponp{\UX e{z'}}{\XX}}{\paren{\YY\quanteq\ket{z},\ \class\YY,\ \PX\PA{ z}}}}
    \\&=\
    \pb\hlfrag{{\YY\quanteq\ket{z},\ \class\YY,\ \oppred{\oponp{\UX e{z'}}{\XX}}{\PX\PA{ z}}}}
    \\&=\
    \pb\hlfrag{{\YY\quanteq\ket{z},\ \class\YY,\ \PX{\oppredp{\oponp{e}{\XX}}{\PA}}{z}}}
    \\&\starstarrel=\
    \pb\hlfrag{{\YY\quanteq\ket{z},\ \class\YY,\
    \textstyle\bigvee_{\tilde z}\pb\paren{\YY\quanteq\ket{\tilde z},\
    \PX{\oppredp{\oponp{e}{\XX}}{\PA}}{z}}}}
    \\&=\
    \pb\hlfrag{{\YY\quanteq\ket{z},\ \class\YY,\ \PX{\oppredp{\oponp{e}{\XX}}{\PA}}{\YY}}}.
    \\&\subseteq\
    \pb\hlfrag{{\class\YY,\ \PX{\oppredp{\oponp{e}{\XX}}{\PA}}{\YY}}}.
  \end{align*}
  Here $(*)$ follows from the fact that by definition, $\UX U{\YY'}$ operates
  as $\UX U{z}$ on quantum memories in $\YY\quanteq\ket z$. And $(**)$ follows since
  $\paren{\YY\quanteq\ket z,\YY\quanteq\ket{\tilde z},\PX\PA {\tilde z}}=0$ for $z\neq \tilde z$.

  Thus (using \ruleref{Conseq} and the syntactic sugar for $\applykw$):
  \[
    \hl
    {\YY\quanteq\ket{z},\ \class\YY,\ \PX\PA{ z}}
    {\apply e\XX}
    {{\class\YY,\ \PX{\oppredp{\oponp{e}{\XX}}{\PA}}{\YY}}}
  \]
  By \ruleref{CaseClassical} (with
  $\PB:=\pb\paren{\class\YY,\ \PX{\oppredp{\oponp{e}{\XX}}{\PA}}{\YY}}$),
  we get the conclusion of \ruleref{ApplyParam}.
\end{proof}

For a simple example of using this rule see the correctness of the
quantum one-time pad (\autoref{sec:qotp.correct}).

\section{Case Study: Quantum One-time Pad}
\label{sec:qotp}

In this section, we give a more advanced example of using Hoare logic
with ghosts. We analyze the \emph{quantum one-time pad}%
\index{quantum one-time pad}%
\index{one-time pad!quantum} (QOTP%
\index{QOTP|see{quantum one-time pad}}, \cite{boykin03qotp,mosca00qotp}),
a simple encryption scheme for quantum data. First, we analyze its
correctness (i.e., the fact that decryption correctly yields the
original plaintext). This is entirely unproblematic and can be done in
most variants of quantum Hoare logic. We include this case as a warm-up
example for reasoning with mixed quantum and classical data. Then we
turn to the security of the QOTP, i.e., the fact that an encrypted
qubit looks like random data if the key is not known.
\fullshort{
  For reasons
  described below (\autoref{sec:qotp.tricky}), this is hard or impossible with prior variants of
  quantum Hoare logic.
}{
  Proving this seems hard or impossible with prior variants of
  quantum Hoare logic.
  (We elaborate on this in \autoref{sec:qotp.tricky}.)
}
It thus shows nicely the power of ghosts.

The QOTP, presented here in its version for single qubits, is very
simple: The key $x=(x_1,x_2)$
are two uniformly random classical bits. The plaintext $\yy$
is a qubit. To encrypt, we apply the
\emph{Pauli-Z}\index{Pauli-Z}\index{Z!Pauli-} operator
\symbolindexmarkonly\pauliZ$\symbolindexmarkhighlight\pauliZ
:=
\begin{tinymatrix}1&0\\0&-1
\end{tinymatrix}$ iff $x_1=1$. Then we apply 
the \emph{Pauli-X}\index{Pauli-X}\index{X!Pauli-} operator
\symbolindexmarkonly\pauliX$\symbolindexmarkhighlight\pauliX
:=
\begin{tinymatrix}0&1\\1&0
\end{tinymatrix}$ iff $x_2=1$.
Or, written more compactly, to encrypt $\yy$,
we apply $\pauliX^{x_2}\pauliZ^{x_1}$ to it.

Decryption works by inverting the sequence of operations, i.e., by applying
$\pauliZ^{x_1}\pauliX^{x_2}$ to $\yy$.

In our language, the QOTP, consisting of key generation, encryption,
and decryption, is expressed as follows:%
\symbolindexmarkonly\Keygen\symbolindexmarkonly\Enc\symbolindexmarkonly\Dec
\begin{equation*}
  \symbolindexmarkhighlight\Keygen\ :=\ \sample\xx K,
  \qquad
  \symbolindexmarkhighlight\Enc\ :=\ \apply{\pauliX^{\xx_2}\pauliZ^{\xx_1}}{\yy}
  \qquad
  \symbolindexmarkhighlight\Dec\ :=\ \apply{\pauliZ^{\xx_1}\pauliX^{\xx_2}}{\yy}
\end{equation*}
Here $\xx$ has type $K:=\bits2$ (the \emph{key space}%
\index{key space!(of QOTP)}), $\yy$
has type $M:=\bit$ (the \emph{message space}%
\index{message space!(of QOTP)}). In slight abuse of notation, we also
use $K$
and $M$
for the uniform distributions over $K$
and $M$,
respectively. Note that the definitions of $\Enc,\Dec$
make use of the syntactic sugar from \autoref{sec:deriv:classical}.

\subsection{Correctness of the QOTP}
\label{sec:qotp.correct}

The correctness of the QOTP can be expressed by the following Hoare judgment:
\begin{equation}
  \forall \zz,\psi.\quad
  \pb\hl{\yy\zz\quanteq\psi}{\Keygen;\Enc;\Dec}{\yy\zz\quanteq\psi}
  \label{eq:qotp.correct}
\end{equation}
We introduced an extra variable $\zz$
here to model that even if $\yy$
is entangled with some other system $\zz$,
decryption correct restores the state of $\yy$
and its entanglement with $\zz$.\fullonly{\footnote{A
  more elementary statement would be
  $\forall\psi.\
  \hl{\yy\quanteq\psi}{\Keygen;\Enc;\Dec}{\yy\quanteq\psi}$.
  Alternatively, we could also state a stronger statement
  ``$\hl{\PA}{\Keygen;\Enc;\Dec}{\PA}$
  for all $\PA$
  with $\xx\notin\fv(\PA)$.''
  Both can be proven with essentially the same derivation as \eqref{eq:qotp.correct}.}}
Using the rules from this paper, the derivation of \eqref{eq:qotp.correct} is elementary:
\begin{align*}
  \pb\hlfrag{\yy\zz\quanteq\psi}\
  & \sample\xx K\
    \pb\hlfrag{\yy\zz\quanteq\psi,\ \class\xx,\ {\textstyle\bigvee_i\xx\quanteq\ket i}}
  &&
     \text{(\ruleref{Sample*})}
  \\
  & \subseteq\
    \pb\hlfrag{\class\xx,\ \yy\zz\quanteq\psi}
  \\
  & \apply{\pauliX^{\xx_2}\pauliZ^{\xx_1}}{\yy}\
    \pb\hlfrag{\class\xx,\ \oppred{\oponp{\pauliX^{\xx_2}\pauliZ^{\xx_1}}{\yy}}{\paren{\yy\zz\quanteq\psi}}}
  && \text{(\ruleref{ApplyParam})}
  \\
  & \apply{\pauliZ^{\xx_1}\pauliX^{\xx_2}}{\yy}\
    \pb\hlfrag{\class\xx,\ \oppred{\oponp{\pauliZ^{\xx_1}\pauliX^{\xx_2}}\yy}{\oppred{\oponp{\pauliX^{\xx_2}\pauliZ^{\xx_1}}{\yy}}{\paren{\yy\zz\quanteq\psi}}}}
    \hskip-1in
    \\
  &&& \text{(\ruleref{ApplyParam})}
  \\
  & =\
    \pb\hlfrag{\class\xx,\ \oppred{\oponp{\pauliZ^{\xx_1}\pauliX^{\xx_2}\pauliX^{\xx_2}\pauliZ^{\xx_1}}{\yy}}{\paren{\yy\zz\quanteq\psi}}}
  \\
  & \starrel=\
    \pb\hlfrag{\class\xx,\ \oppred{\oponp{\id}{\yy}}{\paren{\yy\zz\quanteq\psi}}}
    \ =\ 
    \pb\hlfrag{\class\xx,\ {\yy\zz\quanteq\psi}}
\end{align*}
Here $(*)$
uses that $\pauliZ\pauliZ=\pauliX\pauliX=\id$.
By rules \rulerefx{Seq} and \rulerefx{Conseq}, and the definitions of
$\Keygen,\Enc,\Dec$, we then get \eqref{eq:qotp.correct}.

\subsection{The Quantum One-time Pad and Other Logics}
\label{sec:qotp.tricky}

In this section, we explain why it is hard or even impossible to
analyze the security of the QOTP in existing quantum Hoare
logics (without ghosts).
This section is not
required for understanding the security proof and can be
skipped. However, it illustrates why we need ghosts. Since this
section is about logics from prior work, and it would be beyond the
scope of this section to introduce those logics in more detail, in
this section we assume some familiarity with the logics referenced
here. 

Security of the QOTP means that, after encrypting, the variable $\yy$
is indistinguishable from a uniformly random bit. (More precisely, after running $\Keygen;\Enc$.)
How can we model/prove this in different Hoare logics?

\paragraph{Quantum Hoare logic with subspace predicates.} Probably the
simplest and most obvious variant of quantum Hoare logic is quantum
Hoare logic with subspaces.  Here pre-/postconditions are modeled as
subspaces (a.k.a.~sharp predicates, Birkhoff-von Neumann quantum
logic).  See the ``Recap Hoare logic'' paragraph on
\autopageref{page:recap.hoare} for additional details. In this logic,
we cannot express that $\yy$
is uniformly distributed. Specifically, any predicate $\PA$
on $\yy$
that holds when $\yy$
is a uniformly distributed bit has to also hold for any other
distribution!  Namely, since $\yy$
can be $\ket0$
or $\ket1$
(or anything else), $\PA$
needs to contain $\ket0$
and $\ket1$.
But the only subspace containing both $\ket0$
and $\ket1$
is the full space $\elltwov\yy=\top$.
Thus, the only postcondition on $\yy$
that would be satisfied by the QOTP encryption is $\top$
which would also be satisfied by any insecure encryption scheme.  So
we cannot formulate (let alone prove) any judgment in this variant of
quantum Hoare logic that would express the security of the QOTP.

\paragraph{Quantum Hoare logic with expectations.}
A more expressive variant of quantum Hoare logic is Hoare logic with
expectations. Here, predicates are quantitative
``expectations''\index{expectation},\footnote{Analogous to the
  classical expectations by Kozen \cite{kozen83probabilistic}.}  that is, for a given
state $\rho$,
satisfaction of a predicate $\PA$
is not binary (true/false), but a predicate $\rho$
is satisfied to a certain degree. (Formally, $\PA$
is a Hermitian operator and the degree of satisfaction is defined as
$\tr\PA\rho$.)
And a Hoare judgment $\hl\PA\bc\PB$
means that $\denot\bc(\rho)$
satisfies $\PB$
at least as much as $\rho$
satisfies $\PA$.
Such a logic can be expressed in two ways, either via Hoare triples
(e.g., \cite{ying12floyd}) or equivalently in terms of weakest preconditions
\cite{dhondt06weakest}. 

We know that the weakest preconditions of a quantum program determines
the denotational semantics of said program \cite{dhondt06weakest}, that is, if
$\bc$ and $\bd$ have the same weakest preconditions, then $\denot\bc=\denot\bd$.
(Where $\denot\cdot$
is defined as in \autoref{sec:qprogs}.) Or equivalently:
$\denot\bc=\denot\bd$
iff for all predicates $\PA,\PB$,
$\hl\PA\bc\PB\Leftrightarrow\hl\PA\bd\PB$.
(This holds only if the predicates $\PA,\PB$
are expectations!)
Thus we could define security of the QOTP by
requiring that $\Keygen;\Enc$
has the same weakest preconditions as the program $\sample\yy M$.
(Where we do not allow predicates to refer to $\xx$
since security only holds when the key is secret. For simplicity, we
will assume that $\yy$ is the only variable left after execution.)
The weakest precondition for $\sample\yy M$
and postcondition $\PB$
is $\frac12\tr\PB\cdot\id$.
Thus we can define security of the QOTP as follows: The QOTP is secure
iff for all $\PB$, the weakest precondition of $\Keygen;\Enc$
for postcondition $\PB$
is $\frac12\tr B\cdot\id$.
Or in terms of Hoare judgments: The QOTP is secure iff for all
$\PA,\PB$
with $\PA\leq\frac12\tr B\cdot\id$,
we have $\hl\PA{\Keygen;\Enc}\PB$
and for all $\PA,\PB$
with $\PA\nleq\frac12\tr B\cdot\id$,
we do \emph{not} have $\hl\PA{\Keygen;\Enc}\PB$.

This is formally correct (except for the fact that we glossed over the
fact that there are variables beyond $\yy$),
and security of the one-time pad can be derived due to the
completeness of the calculus from \cite{dhondt06weakest}. However, the
definition is very awkward. In order to prove the security of the
one-time pad, not only do we need to prove that
$\hl\PA{\Keygen;\Enc}\PB$
holds for certain $\PA,\PB$
but also that is does \emph{not} hold for certain others. This is
problematic since usually we reason only in terms of judgments that
hold, and not in terms of judgments that do not hold. (Or, in terms of
weakest precondition, we reason about inequalities, not equalities.)
Especially in the presence of partial specifications, proving that
certain judgments do not hold might be very difficult.

What happens if we simply omit the requirement that some judgments do
not hold? I.e., we use the following definition: The QOTP is secure
iff for all $\PA,\PB$
with $\PA\leq\frac12\tr B\cdot\id$,
we have $\hl\PA{\Keygen;\Enc}\PB$.
As it turns out, this works for the QOTP (we can show that this
condition is equivalent to the original one). However, this is
accidental, and for slight variations of the QOTP, this might not work
any more.

To illustrate this, consider the following slightly artificial variant
$\frac12$QOTP
of the QOTP: This variant has an encryption algorithm $\Enchalf$
that terminates only with probability $\frac12$,
but that works correctly when it terminates. That is
\symbolindexmarkonly\Enchalf$\symbolindexmarkhighlight\Enchalf:=T;\Enc$
where $T$
is a program that does not touch any variables and terminates with
probability $\frac12$
(i.e., $\denot T(\rho)=\frac12\rho$,
e.g., implemented as a loop).  Analogous to the above, we can say that
the $\frac12$QOTP
is secure iff for all $\PA,\PB$
with $\PA\leq\frac14\tr B\cdot\id$,
we have $\hl\PA{\Keygen;\Enchalf}\PB$.
(Note that $\frac12$
was replaced by $\frac14$,
since we now compare with the program $T;\sample\yy M$
which only satisfies these judgments.) But now consider a program
$\bc$ that with probability $\frac12$ runs $\SKIP$,
and with probability $\frac12$ runs the original QOTP.\footnote{Formally,
  $\bc\ :=\ \sample\zz M;\ifte\zz{\Keygen;\Enc}{\SKIP}$.}
This is clearly not a secure encryption scheme (with probability
$\frac12$ the plaintext is leaked by the program $\SKIP$). Yet, we can check that 
when $\PA\leq\frac14\tr B\cdot\id$,
we have $\hl\PA\bc\PB$. Hence the insecure $\bc$ also satisfies our definition of security!
Thus, our security definition of the $\frac12$QOTP
does not guarantee any reasonable security.

Summarizing, if we want to analyze the security of the QOTP or
variants in quantum Hoare logic with expectations, it can be done in
principle, but we need complicated definitions (we cannot express the
security as a single judgment but through an infinite family).  And we
need to not only prove that judgments hold but also that some
judgments do not hold.  (Except in some cases like the QOTP proper
where positive judgments are sufficient. But seeing this needs
additional extra-logical reasoning and does not generalize, e.g., to
the $\frac12$QOTP.)
And even if we surmount these difficulties, it is not clear how we can
reason with such families of judgments in a larger context (e.g., if
the security of the QOTP is needed to derive some property of a larger
program).

\paragraph{Quantum relational Hoare logic.}
Another variant of Hoare logic that seems particularly suitable for
the analysis of the QOTP is quantum relational Hoare logic (qRHL
\cite{qrhl}).  This logic was specifically designed with cryptographic
proofs in mind, inspired by the success of probabilistic relational
Hoare logic \cite{certicrypt}.  In qRHL, Hoare judgments apply to pairs of
programs.  Very roughly speaking, a judgment such as
\symbolindexmark\rhl{$\rhl\PA\bc\bd\PB$}
means that, if the memories of $\bc$
and $\bd$
jointly satisfy the predicate $\PA$
before the execution of $\bc$
and $\bd$,
they will jointly satisfy $\PB$
afterwards.  For example, if $\XX_1,\XX_2$
are the variables of $\bc$,
$\bd$,
respectively, then
$\rhl{\XX_1\QUANTEQ\XX_2}\bc\bd{\XX_1\QUANTEQ\XX_2}$
means that for identical initial states, $\bc,\bd$
have identical final states.  In other words, $\denot\bc=\denot\bd$.
And $\rhl{\top}\bc\bd{\XX_1\QUANTEQ\XX_2}$
would mean that the final states are identical, no matter what the
initial states are, i.e.,
$\forall\rho,\rho'.\ \denot\bc(\rho)=\denot\bd(\rho')$.

Thus, at the first glance, it seems very easy to model the security of
the one-time pad.  To specify that $\yy$
is uniformly random after execution of the QOTP (no matter what its
initial state was), we simply require that $\yy$ after the QOTP is the same as $\yy$
after $\sample\yy M$.
That is, we say the QOTP is secure iff
$\rhl{\top}{\Keygen;\Enc}{\sample\yy
  M}{\yy_1\QUANTEQ\yy_2}$.\footnote{Or alternatively, we could
  define security as
  $\rhl{\top}{\Keygen;\Enc}{\Keygen;\Enc}{\yy_1\QUANTEQ\yy_2}$
  which means that the ciphertexts $\yy_1,\yy_2$
  have the same distribution, no matter what the plaintexts (initial
  values of $\yy_1,\yy_2$)
  are.
  The difficulties described here apply in the same way
  to that definition.} (Here $\yy_1,\yy_2$
refer to the $\yy$
from the left/right program, respectively.)  Unfortunately, this does
not work.  We can show that
$\rhl{\top}{\Keygen;\Enc}{\sample\yy M}{\yy_1\QUANTEQ\yy_2}$
does not hold in qRHL.
Intuitively, the
reason is that after encrypting, $\yy_1$ is still correlated with the key $\xx_1$.
This means that it could still be decrypted to the original plaintext,
and therefore qRHL does not consider it equivalent to a uniformly random~$\yy_2$
(that is independent of any other variables).

To resolve this, we need to erase the key after encrypting. So the
definition becomes: The QOTP is secure iff
\begin{equation}
  \rhl{\top}{\Keygen;\Enc;\initc\xx{00}}{\sample\yy M}{\yy_1\QUANTEQ\yy_2}.
  \label{eq:qrhl.qotp.sec}
\end{equation}
This judgment is indeed a good definition
for the security of the QOTP. It is not hard to prove that it holds by
explicitly computing the superoperators
$\denot{\Keygen;\Enc;\initc\xx{00}}$
and $\denot{\sample\yy M}$,
and then showing \eqref{eq:qrhl.qotp.sec} directly from the semantic definition of
qRHL. In that sense, qRHL is superior to the two Hoare logic variants
above: at least we can state the security of the QOTP concisely. (And
for variants of it such as $\frac12$QOTP, similar definitions work.)

Unfortunately, it seems hard (or impossible) to derive
\eqref{eq:qrhl.qotp.sec} \emph{within} the logic.  (That is, by an
application of a sequence of reasoning rules.)  While we do not have a
proof that \eqref{eq:qrhl.qotp.sec} cannot be derived from the rules
from \cite{qrhl}, a natural proof would seem to go along the following
lines: First, we show $\rhl{\top}{\Keygen;\Enc}{\sample\yy M}\PB$
for some $\PB$,
then we show $\rhl\PB{\initc\xx{00}}\SKIP{\yy_1\QUANTEQ\yy_2}$,
and then we use the qRHL-analogue to \ruleref{Seq} to conclude
\eqref{eq:qrhl.qotp.sec}. Unfortunately, we can show that there exists
no predicate $\PB$
such that both $\rhl{\top}{\Keygen;\Enc}{\sample\yy M}\PB$
and $\rhl\PB{\initc\xx{00}}\SKIP{\yy_1\QUANTEQ\yy_2}$
are true. Hence this proof approach is doomed.

To summarize, in qRHL, while it is easy to formulate the security of
the QOTP, there are reasons to believe that the security proof is
difficult or even impossible.

\subsection{Security of the Quantum One-time Pad}
\label{sec:sec.qotp}

We will now demonstrate how to prove the security of the QOTP using
quantum Hoare logic with ghosts.  Security of the QOTP means that,
after encrypting, the variable $\yy$
is indistinguishable from a uniformly random bit, as long as the key
is not known. We formalize this by requiring that after key
generation, encryption, and subsequent deletion of the key, $\yy$
is uniformly random. As a Hoare judgment, we write this as:
\begin{equation}
  \label{eq:sec.qotp}
  \pb\hl\top{\Keygen;\Enc;\initc\xx{00}}{\uniform\yy}.
\end{equation}
The $\initc\xx{00}$
overwrites the key $\xx$
and is added to model the fact that we do not know the
key.\footnote{One might think that it should be sufficient to simply not
  mention the key in the postcondition. I.e., to define security
  as $\pb\hl\top{\Keygen;\Enc}{\uniform\yy}$.
  However, from \autoref{lemma:distrib} we know that in a state
  satisfying $\uniform\yy$, $\yy$ is uniform \emph{and independent} of
  all other variables.
  This is clearly not the case after encryption ($\yy$~is not independent of the key $\xx$).
  Thus $\pb\hl\top{\Keygen;\Enc}{\uniform\yy}$ does not hold.
} The precondition is $\top$ since we do not want to make any assumption
about the initial state of~$\yy$, i.e., about the plaintext.
(In particular, the plaintext can be entangled with other variables.)

\paragraph{Warm up.}
\fullshort{
  Before we show \eqref{eq:sec.qotp},
  we prove a weaker claim as a warm up:
}{
  Instead of showing \eqref{eq:sec.qotp},
  we prove a weaker claim for simplicity
  (the general case is shown in \autoref{app:qotp.general}):
}
\begin{equation}
  \label{eq:sec.qotp.simp}
  \forall\psi\neq0.\quad
  \pb\hl{\yy\quanteq\psi}{\Keygen;\Enc;\initc\xx{00}}{\uniform\yy}.
\end{equation}
This equation say that the QOTP is secure as long as the plaintext
$\yy$
is some (arbitraty) state $\psi$.
That is, it only guarantees security for unentangled plaintexts.
\fullshort{
  We will do the general case~\eqref{eq:sec.qotp} below,
  but the simplified case is simpler and contains already many of the needed ideas.
}{
  This simplified case is simpler and its proof contains already many
  of the ideas needed for the general case~\eqref{eq:sec.qotp}.
}

Recall that $\Keygen\ =\ \sample\xx K$
and $\Enc\ =\ \apply{\pauliX^{\xx_2}\pauliZ^{\xx_1}}{\yy}$. We first
derive a postcondition for $\Keygen;\Enc;\initc\xx{00}$ by simply
applying the reasoning rules step by step:
\begin{align*}
  & \pb\hlfrag{\yy\quanteq\psi}\
   \sample\xx K\
    \pb\hlfrag{\yy\quanteq\psi,\ \uniform\xx}
  && \text{(\ruleref{Sample})}
  \\
  &\quad \apply{\pauliX^{\xx_2}\pauliZ^{\xx_1}}{\yy} \
    \pb\hlfrag{\oppred{\pb\oponp{\UX{\paren{\pauliX^{\xx_2}\pauliZ^{\xx_1}}}\xx}{\xx\yy}}
    {\pb\paren{\yy\quanteq\psi,\ \uniform\xx}}}
  && \text{(\ruleref{Apply})}
  \\
  & \quad\initc\xx{00} \
    \pb\hlfrag{\oppred{\pb\oponp{\UX{\paren{\pauliX^{\ee_2}\pauliZ^{\ee_1}}}\ee}{\ee\yy}}
    {\pb\paren{\yy\quanteq\psi,\ \uniform\ee}},\ \xx\quanteq\ket{00},\ \class\xx}
    \hskip-1in
  \\
  &&& \text{(\ruleref{InitC})}
  \\
  & \quad\mathord\subseteq\
    \pb\hlfrag{\oppred{\pb\oponp{\UX{\paren{\pauliX^{\ee_2}\pauliZ^{\ee_1}}}\ee}{\ee\yy}}
    {\pb\paren{\yy\quanteq\psi,\ \uniform\ee}}}
\end{align*}
For the application of \ruleref{Sample}, recall that
$\uniform\xx$ is syntactic sugar for $\distr\xx K$ where $K$ is the
uniform distribution on the type of $\xx$.
For the application of \ruleref{Apply}, recall that 
$\apply{\pauliX^{\xx_2}\pauliZ^{\xx_1}}{\yy}$ is syntactic sugar for
$\apply{ \UX{\paren{\pauliX^{\xx_2}\pauliZ^{\xx_1}}}\xx }{\xx\yy}$
(\autopageref{page:applysugar}).
By rules \rulerefx{Seq} and \rulerefx{Conseq}, we immediately get
\begin{equation}
  \label{eq:qotp.pc1}
  \pb\hl{\yy\quanteq\psi}
  {\Keygen;\Enc;\initc\xx{00}}
  {\oppred{\pb\oponp{\UX{\paren{\pauliX^{\ee_2}\pauliZ^{\ee_1}}}\ee}{\ee\yy}}
    {\pb\paren{\yy\quanteq\psi,\ \uniform\ee}}}
  =: \{\PB\}.
\end{equation}
While this is not yet the final result \eqref{eq:sec.qotp.simp} we
wanted, we see that the application of the \rulerefx{InitC} rule
already achieved one important thing: Since $\xx$
was turned into a ghost, the postcondition $\PB$
refers only to the ciphertext $\yy$
and not other variables, i.e., we got rid of the dependence between
$\yy$
and $\xx$.
What is left to do is to prove that the postcondition $\PB$
implies $\uniform\yy$.

Analyzing $\PB$
involves some calculations. (This is to be expected because the QOTP
relies on the properties of the involved matrices, so we have to
calculate somewhere.)  We first unfold the syntactic sugar.
$\uniform\ee$ means $\ee\ee'\quanteq\psiDD K$ for some fresh $\ee'$. Thus
$\paren{\yy\quanteq\psi,\ \uniform\ee}=
\paren{\yy\ee\ee'\quanteq \psi\otimes\psiDD K}$. Furthermore 
$\UX{\paren{\pauliX^{\ee_2}\pauliZ^{\ee_1}}}\ee
=
\sum_{k\in K}\proj{\ket k}\otimes
{\pauliX^{k_2}\pauliZ^{k_1}}$ (see \autopageref{page:applysugar}).
Thus
\begin{align*}
  \PB &= 
        {\oppred{\pBb\oponp{
        \pB\paren{\sum_{k\in K}\proj{\ket k}\otimes
        {\pauliX^{k_2}\pauliZ^{k_1}}}}{\ee\yy}}
        {\pb\paren{\yy\ee\ee'\quanteq \psi\otimes\psiDD K}}}
  \\
      &\starrel=
        \pBb\paren{
        \yy\ee\ee' \quanteq
        \underbrace{\pB\paren{\sum\nolimits_{k}
        {\pauliX^{k_2}\pauliZ^{k_1}}\otimes\proj{\ket k}\otimes\id}
        \paren{\psi\otimes\psiDD K}}_{{}=:\phi}}.
\end{align*}
(Note that in $(*)$,
the tensor product factors in the sum are written in a different order because
the $\oponp\dots{\ee\yy}$-term
and the $(\yy\ee\ee'\quanteq\dots)$-term
list the variables in a different order.)  Since $\PB$
is now of the form $\yy\ee\ee'=\phi$,
it is amenable to rewriting using \ruleref{ShapeShift}. Furthermore,
our intended postcondition $\uniform\yy$
is syntactic sugar for $\yy\ee''\quanteq\psiDD M$,
which is also compatible with \rulerefx{ShapeShift}. Specifically, if
we can show $\partr{\ee\ee'}\proj\phi=\partr{\ee''}\proj{\psiDD M}$,
then \ruleref{ShapeShift} implies
\begin{equation}
  \label{eq:qotp.ss}
  \PB = \paren{\yy\ee\ee'\quanteq\phi}
  \impl \paren{\yy\ee\quanteq\psiDD M}
  = \uniform\yy.
\end{equation}
(Recall that $M$ is the uniform distribution on the type of $\yy$,
i.e., on $\bit$.)
We now show  $\partr{\ee\ee'}\proj\phi=\partr{\ee''}\proj{\psiDD M}$
by computation.
Since $\psiDD K=\sum_{l\in K}\frac12\ket l\otimes\ket l$ by
definition, we have:
\[
  \phi =
  {\sum_{kl}\tfrac12
    {\pauliX^{k_2}\pauliZ^{k_1}}
    \psi\otimes\proj{\ket k}\,\ket l_{\ee}\otimes\ket l_{\ee'}}
  =
  {\sum_{k}\tfrac12
    {\pauliX^{k_2}\pauliZ^{k_1}}
    \psi\otimes\ket k_{\ee}\otimes\ket k_{\ee'}}.
\]
Thus
\begin{align*}
  \partr{\ee\ee'}\proj\phi
  &=
    \partr{\ee\ee'}\sum_{kk'}\tfrac14
    {\pauliX^{k_2}\pauliZ^{k_1}}\proj\psi
    \adjp{\pauliX^{k'_2}\pauliZ^{k'_1}}
    \otimes
    \ket k\bra{k'}_{\ee}
    \otimes
    \ket k\bra{k'}_{\ee'} \\
  &\starrel=
    \sum_{k}\tfrac14\,
    {\pauliX^{k_2}\pauliZ^{k_1}}\proj\psi
    \adjp{\pauliX^{k_2}\pauliZ^{k_1}}
    =
    \sum_{k}\tfrac14\,
    \proj{\pauliX^{k_2}\pauliZ^{k_1}\psi}.
\end{align*}
Here $(*)$
follows from the facts that
$\partr{\ee\ee'}\sigma\otimes\tau = \sigma\,\tr\tau$
and that $\tr\ket k\bra{k'}=1$
if $k=k'$
and $=0$
otherwise.  Without loss of generality, we can assume that
$\norm{\psi}=1$
(because the predicate $\yy\quanteq\psi$ does not change if we
multiply $\psi$ with a nonzero scalar.)
Thus $\psi=\begin{tinymatrix}\alpha\\\beta
\end{tinymatrix}
$ for some $\alpha,\beta\in\setC$
with $\alpha\adj\alpha+\beta\adj\beta=1$.
Then $\sum_{k}\tfrac14\, \proj{\pauliX^{k_2}\pauliZ^{k_1}\psi}$
can be explicitly computed (a sum of four $2\times 2$-matrices),
and simplifies to $\frac12\id$.
Furthermore, we easily compute that
$\partr{\ee''}\psiDD M=\frac12\id$.
Thus $\partr{\ee\ee'}\proj\phi=\partr{\ee''}\proj{\psiDD M}$.
Hence \eqref{eq:qotp.ss} follows by \ruleref{ShapeShift}.  From
\eqref{eq:qotp.pc1}, \eqref{eq:qotp.ss}, with \ruleref{Seq}, we get
\eqref{eq:sec.qotp.simp}.  This shows the security of the QOTP in the
special case that the plaintext is unentangled.

\paragraph{General case.}
We have shown the security of the QOTP in
the special case \eqref{eq:sec.qotp.simp} that the plaintext is not entangled with anything
else but is in a fixed but arbitrary state $\psi$.
We now show the general case \eqref{eq:sec.qotp}.
To do so, we first show something similar to the special case
\eqref{eq:sec.qotp.simp}, namely that the QOTP is secure when the
plaintext $\yy$ and one further variable $\zz$ are in a fixed state
$\psi$.
(And, for technical reasons we also include the variable $\xx$,
but that variable is less interesting since it is overwritten by $\Keygen$.)
Formally,
\begin{equation}
  \label{eq:sec.qotp2}
  \forall\psi\neq0.\quad
  \pb\hl{\yy\zz\xx\quanteq\psi}{\Keygen;\Enc;\initc\xx{00}}{\uniform\yy}.
\end{equation}
Here $\zz$
is a program variable of infinite cardinality (e.g., of type integer).
Intuitively, this already means that the QOTP is secure when the
plaintext $\yy$ is entangled.
And indeed, the general case \eqref{eq:sec.qotp} then is an immediate consequence
of \eqref{eq:sec.qotp2} and \ruleref{Universe} (with $\XX:=\xx\yy$,
$\xx:=\zz$, $\EE:=\UU:=\varnothing$, $\PA:=\top$, $\PB:=\uniform\yy$).

We are left to show \eqref{eq:sec.qotp2}. This is done similarly to
\eqref{eq:sec.qotp.simp}, except that the computations are a bit more
complex. First, we have
\begin{align*}
  & \pb\hlfrag{\yy\zz\xx\quanteq\psi}\
   \sample\xx K\
    \pb\hlfrag{\yy\zz\ee''\quanteq\psi,\ \uniform\xx}
  && \text{(\ruleref{Sample})}
  \\
  &\quad \apply{\pauliX^{\xx_2}\pauliZ^{\xx_1}}{\yy} \
    \pb\hlfrag{\oppred{\pb\oponp{\UX{\paren{\pauliX^{\xx_2}\pauliZ^{\xx_1}}}\xx}{\xx\yy}}
    {\pb\paren{\yy\zz\ee''\quanteq\psi,\ \uniform\xx}}}
  && \text{(\ruleref{Apply})}
  \\
  & \quad\initc\xx{00} \
    \pb\hlfrag{\oppred{\pb\oponp{\UX{\paren{\pauliX^{\ee_2}\pauliZ^{\ee_1}}}\ee}{\ee\yy}}
    {\pb\paren{\yy\zz\ee''\quanteq\psi,\ \uniform\ee}},\ \xx\quanteq\ket{00},\ \class\xx}
    \hskip-1in
  \\
  &&& \text{(\ruleref{InitC})}
  \\
  & \quad\mathord\subseteq\
    \pb\hlfrag{\oppred{\pb\oponp{\UX{\paren{\pauliX^{\ee_2}\pauliZ^{\ee_1}}}\ee}{\ee\yy}}
    {\pb\paren{\yy\zz\ee''\quanteq\psi,\ \uniform\ee}}}.
\end{align*}
Thus with \ruleref{Seq} and \ruleref{Conseq}:
\begin{equation}
  \label{eq:qotp.pc1'}
  \pb\hl{\yy\zz\xx\quanteq\psi}
  {\Keygen;\Enc;\initc\xx{00}}
  {\oppred{\pb\oponp{\UX{\paren{\pauliX^{\ee_2}\pauliZ^{\ee_1}}}\ee}{\ee\yy}}
    {\pb\paren{\yy\zz\ee''\quanteq\psi,\ \uniform\ee}}}
  =: \{\PB'\}.
\end{equation}
As in the special case, we unfold syntactic
sugar and we get:
\begin{equation*}
  \PB' =
  \pBb\paren{
    \yy\zz\ee''\ee\ee' \quanteq
    \underbrace{\pB\paren{\sum_{k\in K}
        {\pauliX^{k_2}\pauliZ^{k_1}}\otimes\id_{\zz\ee''}\otimes\proj{\ket k_{\ee}}\otimes\id_{\ee'}}
      \paren{\psi\otimes\psiDD K}}_{{}=:\phi'}}.
\end{equation*}
Quite analogous to the special case, we compute
\begin{equation}
  \phi'
  = \sum_k \tfrac12\paren{
    X^{k_2}Z^{k_1}\otimes\id_{\zz\ee''}}\psi\otimes\ket k_{\ee}\otimes\ket k_{\ee'}
  \quad\text{and}\quad
  \partr{\ee\ee'}\proj{\phi'}
  =
  \sum_k\tfrac14\pb\proj{\paren{
      X^{k_2}Z^{k_1}\otimes\id_{\zz\ee''}}\psi}.
  \label{eq:gen.ee''}
\end{equation}
(We will additionally need to trace out $\ee''$,
but the computation is easier if we do not do that yet.)  We can write
$\psi$
as $\ket0_{\yy}\otimes\psi_0+\ket1_{\yy}\otimes\psi_1$
for some $\psi_0,\psi_1$ over $\zz\ee''$. By substituting this in the
rhs of the second equation in \eqref{eq:gen.ee''}, and multiplying out
and canceling terms, we get
\begin{equation*}
\partr{\ee\ee'}\proj{\phi'}
=
\tfrac12\pb\paren{\proj{\ket0}_{\yy}+\proj{\ket1}_{\yy}}\otimes\pb\paren{\underbrace{\proj{\psi_0}+\proj{\psi_1}}_{{}=:\rho_{\zz\ee''}}}
= \tfrac12\id_{\yy}\otimes\rho_{\zz\ee''}.
\end{equation*}
Let $\ee_{\yy},\ee_{\zz}$
be additional entangled ghosts. Then
$\partr{\ee_{\yy}}\proj{\psiDD M}=\tfrac12\id_{\yy}$
(if we interpret $\psiDD M$
as a quantum memory over $\yy\ee_{\yy}$).
And there exists a $\gamma$
over $\zz\ee'\ee_{\zz}$
such that $\partr{\ee_{\zz}}\proj\gamma=\rho_{\zz\ee''}$.
Thus
$\partr{\ee\ee'}\proj{\phi'} = \partr{\ee_{\yy}\ee_{\zz}}\proj{\psiDD
  M \otimes \gamma}$. Hence
$\partr{\ee''\ee\ee'}\proj{\phi'} = \partr{\ee''\ee_{\yy}\ee_{\zz}}\proj{\psiDD
  M \otimes \gamma}$.
Using \ruleref{ShapeShift} for $(*)$, we thus have
\begin{multline*}
  \PB'
  = \paren{\yy\zz\ee''\ee\ee' \quanteq \phi'}
  \starrel\impl
  \paren{\yy\ee_{\yy}\zz\ee''\ee_{\zz} \quanteq \psiDD M\otimes\gamma} \\
  =
  \paren{\yy\ee_{\yy}\quanteq\psiDD M,\ \zz\ee''\ee_{\zz} \quanteq
    \gamma}
  \subseteq 
  \paren{\yy\ee_{\yy}\quanteq\psiDD M}
  =
  \uniform\yy
  .
\end{multline*}
Then \eqref{eq:sec.qotp2} follows with \eqref{eq:qotp.pc1'},
\ruleref{Seq} and \ruleref{Conseq}.  And, as mentioned above, the
general case \eqref{eq:sec.qotp} is an immediate consequence of
\eqref{eq:sec.qotp2} and \ruleref{Universe}. This shows the security
of the QOTP.

\appendix

\renewcommand\symbolindexentry[4]{
  \noindent\hbox{\hbox to 2in{$#2$\hfill}\parbox[t]{3.5in}{#3}\hbox to 1cm{\hfill #4}}\\[2pt]}

\clearpage
\printsymbolindex

\printindex  

\printbibliography

\end{document}
